\documentclass[a4paper,12pt]{article}

\usepackage{amsmath,amsfonts,amssymb,amsthm,amsmath,mathrsfs,stmaryrd}

\RequirePackage[dvipsnames]{xcolor}
\definecolor{arXiv}{named}{Blue} 
\definecolor{ColorCite}{named}{BrickRed}
\definecolor{ColorLink}{named}{NavyBlue}
\definecolor{ColorURL}{named}{RoyalBlue}

\usepackage{graphicx} 
\usepackage[colorlinks=false,
citecolor=ColorCite,
linkcolor=ColorLink,
urlcolor=ColorURL]{hyperref}
\usepackage{tikz,tikz-cd}
\usepackage{cite}
\usepackage{rotating}
\usepackage{url}
\usepackage{subcaption}
\usepackage{svg}
\usepackage{fullpage} 
\usepackage{wrapfig}
\usepackage{pdfpages}
\usepackage[normalem]{ulem}

\theoremstyle{definition}
\newtheorem{Example}{Example}[section]
\theoremstyle{definition}
\newtheorem{Definition}[Example]{Definition}
\theoremstyle{definition}
\newtheorem{Remark}[Example]{Remark}
\theoremstyle{definition}

\theoremstyle{definition}
\newtheorem{Theorem}[Example]{Theorem}
\theoremstyle{definition}
\newtheorem{Proposition}[Example]{Proposition}
\theoremstyle{definition}

\title{A geometrical invitation to BMS group theory}

\author{Xavier Bekaert$^a$, Yannick Herfray$^a$, Lea Mele$^b$, Noémie Parrini$^b$}

\date{$^a$ Institut Denis Poisson, Université de Tours, Université d’Orléans,\\ CNRS, IDP, UMR 7013, Parc de Grandmont,\\
37200 Tours, France\\\vspace{2mm}
$^b$ Physique de l’Univers, Champs et Gravitation,\\ Université de Mons – UMONS,
Place du Parc 20,\\ 7000 Mons, Belgium\\\vspace{2mm}
\texttt{\small xavier.bekaert@univ-tours.fr, yannick.herfray@univ-tours.fr, lea.mele@umons.ac.be, noemie.parrini@umons.ac.be}}

\begin{document}

\maketitle

\thispagestyle{empty} 

\abstract{In these lecture notes, a group-theoretical introduction to BMS symmetries is provided in a self-contained manner. More precisely, all definitions and structures are purely based on geometrical and group-theoretical notions defined at null infinity and valid in any dimension (four or higher), in a way that circumvents its traditional bulk realisation as asymptotic symmetries. The topics which are reviewed are: the definition of BMS transformations as conformal Carrollian isometries of null infinity, the semidirect structure of the BMS group, the holographic reconstruction of Minkowski spacetime in terms of good cuts, the one-to-one correspondence between good cut subspaces and Poincar\'{e} subgroups (aka vacua), as well as a basic introduction to the unitary representations of the BMS group.}

\pagebreak 

\tableofcontents

\thispagestyle{empty} 

\pagebreak

\setcounter{page}{1}

\section{Introduction}\label{section: introduction}

The Bondi-Metzner-Sachs (BMS) group was initially introduced in \cite{bondi_gravitational_1962,Sachs:1962zza} as the group of \emph{asymptotic symmetries} of asymptotically flat spacetimes. To their surprise, these asymptotic symmetries turned out not to be the expected Poincaré group but, rather, an infinite-dimensional extension of this group. 

The modern significance of the BMS group is that, as was shown 
by Strominger and collaborators \cite{Strominger:2013jfa,He:2014laa}, it is
a symmetry of the (perturbative) gravity S-matrix (see the review \cite{Strominger:2017zoo}). In this sense, the BMS group appears to be of a more fundamental nature than the Poincaré group, at least in the presence of gravitational interactions. Since our very notion of particles in Quantum Field Theory (QFT) relies on their realisation as unitary irreducible representations (UIRs) of the Poincaré group, these results suggest that BMS particles \cite{Barnich:2014kra,Barnich:2015uva,Oblak:2016eij,Bekaert:2024uuy,Bekaert:2025kjb}, i.e. UIRs of the BMS group, might in fact improve
the notion of asymptotic states in situations where gravity must be taken into account. This possibility --- that BMS particles might be more fundamental than Poincaré particles --- was in fact already contemplated very early on by McCarthy as a motivation for his seminal classification of UIRs of the BMS group \cite{Mccarthy:1972ry,McCarthy_72-I,McCarthy_73-II,McCarthy_73-III,McCarthy:1974aw,McCarthy_75,McCarthy_76-IV,McCarthy_78,McCarthy_78errata}.

These lecture notes intend to provide a self-contained\footnote{\label{Foot1}These lecture notes are aimed at theoretical physicists interested in the asymptotic properties of gravity. As such, on the one hand, they assume some basic knowledge of concepts in general relativity, such as conformal compactification \`a la Penrose (see, e.g., \cite{Geroch1977} for a pedagogical introduction) that will not be reviewed here; on the other hand, in the main text, we will generally be loose on details regarding functional analysis. At times, when we believe that these have some level of physical relevance, further details will be discussed in the form of remarks.} geometrical introduction to the BMS group, the focus being group-theoretical aspects. In particular, the present exposition will culminate in a basic introduction to the classification of UIRs of the BMS group and we hope that it will serve as an invitation to the subject.

Another emphasis of this review is the relationship between BMS symmetries and null geometries. As soon as Penrose introduced the notion of asymptotically simple spacetimes and associated conformal compactifications \cite{Penrose:1962ij}, the BMS group was realised to be intimately tied to the intrinsic geometry of null infinity (see, e.g., \cite{Geroch1977}). This fact was related to the Carroll (i.e. $c\to0$) limit in \cite{Duval:2014uva} and, since then, there has been a continuous resurgence of interest in Carrollian geometry and related Carrollian field theories (see, e.g., \cite{Figueroa-OFarrill:2022nui,Bergshoeff:2022eog,Bagchi:2025vri,Nguyen:2025zhg,Ciambelli:2025unn,Ruzziconi:2026bix} for some reviews).

Throughout these notes, we will also systematically present definitions and results valid in any dimension, higher or equal to four.\footnote{\label{Foot 3D gravity} Three-dimensional gravity and the group $BMS_3$ are very interesting but they are qualitatively very different. They will not be addressed in the present lecture notes. See \cite{Barnich:2014kra,Barnich:2015uva,Oblak:2016eij} for thorough analyses.} Some brief remarks might be in order, since some theorems have been put forward as no-go results against the existence of supertranslations or memory effects in higher dimensions \cite{Hollands:2016oma,Garfinkle:2017fre}. Of course, the weakness of any no-go theorem is the strength of its assumptions: A more positive way to phrase these results is that, in higher dimensions, supertranslations and memory effects are respectively over-leading and sub-leading in the $1/r$ expansion, with respect to the radiative order where gravitational waves appear. By contrast, in four dimensions, all these orders coincide. More precisely, memory effects do exist in any \textit{even} dimension \cite{Kapec:2015vwa,Mao:2017wvx,Pate:2017fgt,Satishchandran:2019pyc,Campoleoni:2020ejn} while supertranslations, and the BMS group, exist in any dimension. 
\\

The plan of these notes is as follows. Section \ref{BMSCarroll} reviews the relationships between the BMS group and the geometry of null infinity. In order to do so, a step-by-step approach to (conformal) Carrollian geometry is provided, with various useful notions being introduced along the way and illustrated with basic examples.
Section \ref{semiprodstrBMS} discusses the semidirect product structure of the BMS group and the geometrical significance of some physically relevant subgroups, in particular Poincaré subgroups and the related notion of gravity vacua.
The title of Section \ref{Alicebdyland} is ``Alice in Boundaryland'' and its subtitle could have been:
``How to reconstruct Minkowski spacetime purely from data at null infinity?''. On the one hand, this section reviews the notion of good cuts from the bulk perspective and their correspondence with points of Minkowski spacetime. On the other hand, it also presents the intrinsic realisation of good cuts as paraboloids of revolution in conformally flat Carrollian spacetime. Finally, Section \ref{BMSUIRs} gives a glimpse of the classification of UIRs of the BMS group. We particularly emphasise the hard/soft decomposition of supermomenta, with some results in higher dimensions presented here for the first time. A short appendix concludes this paper with various technical proofs. 

\section{BMS group and Carrollian geometries}\label{BMSCarroll}

To prepare the discussion of Carrollian geometries, we review various concepts of principal $\mathbb{R}$-bundles in Section \ref{principRbundes} and we illustrate them in the case of null infinity. Then, we review the main definitions of Carrollian spacetimes in Section \ref{Carrolliangeometry}. Finally, we conclude in Section \ref{ConformalCarrolliangeometry} with their conformal generalisation relevant for BMS symmetries.

\subsection{Principal bundle geometry}\label{principRbundes}

Let us consider a complete nowhere-vanishing vector field $n = n^\mu \partial_\mu \neq 0$ on a manifold $\mathscr{M}$ of dimension $d+1$. Taking the integral lines of  $n$, a congruence of parametrised open curves is obtained. 
In the language of fibre bundles, this vector field $n$ can, under mild assumptions\footnote{Namely, one needs to assume that the corresponding $\mathbb{R}$-action on ${\mathscr{M}}$ is proper in order to ensure that $\bar{\mathscr{M}}$ is a manifold.}, be thought of as the fundamental vector field of a principal bundle   
\begin{equation}\label{prbdle}
\pi\, :\, \mathscr{M} \twoheadrightarrow \bar{\mathscr{M}}
\end{equation}
with one-dimensional fibres being either diffeomorphic to $U(1)$ or $\mathbb{R}$.
We will focus on the latter case.
Consequently, the non-compact group $\mathbb{R}$ acts on $\mathscr{M}$ via shifts of the affine parameter $u$. Therefore, the integral lines of $n$ are the orbits of the $\mathbb{R}$-action.\footnote{We recall that, roughly speaking, a principal bundle has a typical fibre which can be identified (via a diffeomorphism) with its structure group via the group action.}
The space $\bar{\mathscr{M}}$ of orbits is the base space of the principal bundle \eqref{prbdle} and can be seen as the quotient of the total space $\mathscr{M}$ by $\mathbb{R}$:
\begin{equation}
    \bar{\mathscr{M}} = \mathscr{M}/ \,\mathbb{R}\,. 
\end{equation}

Without loss of generality, such a principal $\mathbb R$-bundle $\mathscr{M}$ can be assumed to be trivial, i.e. it is isomorphic to $\mathbb{R}\times\bar{\mathscr{M}} $ (see, e.g., Proposition 16.14.5 of \cite{Dieudonne1970}).
The situation outlined here is represented in Figure \ref{figure1}.
\begin{figure}[h]
\centering
\includegraphics[width=2in]{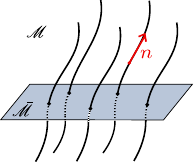}
\captionsetup{width=0.8\textwidth}
\caption{Principal $\mathbb R$-bundle $\mathscr{M}$ over $\bar{\mathscr{M}}$ with its fundamental vector field $n$}
\label{figure1}
\end{figure} 

A coordinate system on $\mathscr{M}$ that will find many applications in the remainder of these notes is $(u, x^i)$, 
where $u \in \mathbb{R}$ denotes an affine coordinate
on the fibres and $x^i \in \mathbb{R}^{d}$ are the coordinates on the base space ($i=1,\ldots,d$). In these coordinates, the expression of the fundamental vector field is simply $n = \partial/ \partial u$. The integral curves are characterised  by constant base coordinates $x^i = x_0^i$ and are parametrised by $u$. The $\mathbb{R}$-action on $\mathscr{M}$ is defined by shifting $u$ by a constant number. The fibration \eqref{prbdle} is locally defined as
\begin{equation}
    \pi\, : \,  (u, x^i)   \longmapsto x^i.
\end{equation}
The coordinates $(u, x^i)$ will be called \textit{adapted coordinates}.

\begin{Example}\label{exampleflatscri}
Let us consider the conformal compactification $\widehat{\mathbb{R}^{d+1,1}}$ of Minkowski spacetime (see, e.g., the section 2 of \cite{Eastwood:1991} for a review in any dimension). In suitable coordinates, the Minkowski metric $\eta$ on $\mathbb{R}^{d+1,1}$ can be rewritten as $\eta = \omega^{-2} \hat{\eta}$ with
\begin{align}
    \hat{\eta} &= -\frac{4dUdV}{\sin^{2}(V-U)} + \gamma\,,& \omega = \frac{2\cos(U) \cos(V)}{\sin(V-U)}.
\end{align}
Here $U,V\in [-\frac{\pi}{2},\frac{\pi}{2} ]$ with $V\geqslant U$ are doubly null coordinates, while $x$ and $\gamma$ are, respectively, coordinates and metric on the unit round sphere $S^d$. Future null infinity $\mathscr{I}_{d+1}^+$, represented in Figure \ref{fig: compactification}, then corresponds to the points for which $V=\frac{\pi}{2}$ and $U\in (-\frac{\pi}{2},\frac{\pi}{2} )$. The fundamental vector field along null infinity, $n := g^{-1} \nabla \omega \big|_{\mathscr{I}^+}$, reads
\begin{equation}
  n =  \cos^2(U) \partial_U.
\end{equation}
 It is a null vector field whose integral lines are the null rays generating the past null cone of $i^+$ ($U=V=\pi/2$). 
The fibration $\pi: (U,x^i) \mapsto x^i$ projects a point of future null infinity $\mathscr{I}^+_{d+1}$ onto the corresponding point on the celestial sphere $S^d$, see Figure \ref{fig: section}. Introducing the retarded time $u =\tan U$, we obtain an adapted coordinate system $(u,x^i)$, i.e. $n= \partial/\partial u$. 
\end{Example}

\begin{figure}
\centering

\includegraphics[width=2in]{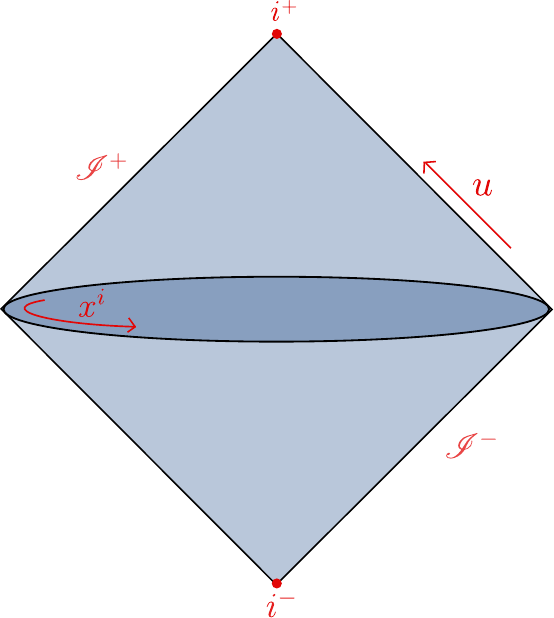}
\captionsetup{width=0.9\textwidth}
\caption{Compactified Minkowski spacetime $(\hat{\eta},\widehat{\mathbb{R}^{d+1,1}})$ with its conformal boundary $\mathscr{I}^\pm_{d+1}$}

\label{fig: compactification}
\end{figure}

\newpage

\begin{Example}\label{Example: radiative spacetimes}
 More generally, the intuition behind an asymptotically flat spacetime $(g,M)$ can be captured \cite{Penrose:1962ij} by the requirement that the metric be of the form $g= \omega^{-2}\hat{g}$, where the ``unphysical metric'' $\hat{g}$ is the metric of a spacetime $(\hat{g},\hat{M})$ with boundary and the conformal factor $\omega$ is the defining function of the boundary. More precisely, $\hat{M}=M\sqcup\partial \hat{M}$ with $M\equiv\{\omega>0\}$ and $\hat{M}\equiv\{\omega\geqslant 0\}$, hence $\partial\hat{M}\equiv\{\omega=0\}$. The conformal boundary is then either a principal $\mathbb{R}$-bundle or a principal $U(1)$-bundle.\footnote{The latter case, which happens for Taub-NUT solutions and in any situation where the so-called ``magnetic gravitational charge'' is non-zero \cite{Ramaswamy1981}, is usually considered to be a pathology, since this implies the existence of closed causal curves. We will exclude this possibility.} This boundary at infinity has the essential physical property of being the ultimate destination of all outgoing gravitational radiation, which is leaking away from the spacetime, see Figure \ref{fig: radiative spacetime}. The intuitive idea here is that the gravitational waves, which follow null geodesics, are eventually blocked on the boundary at infinity and stay stuck there. As we shall see later, Carrollian geometry allows for a neat geometrical description of this final stage of radiation.
\end{Example}

\begin{figure}[h]
\centering

\includegraphics[width=2in]{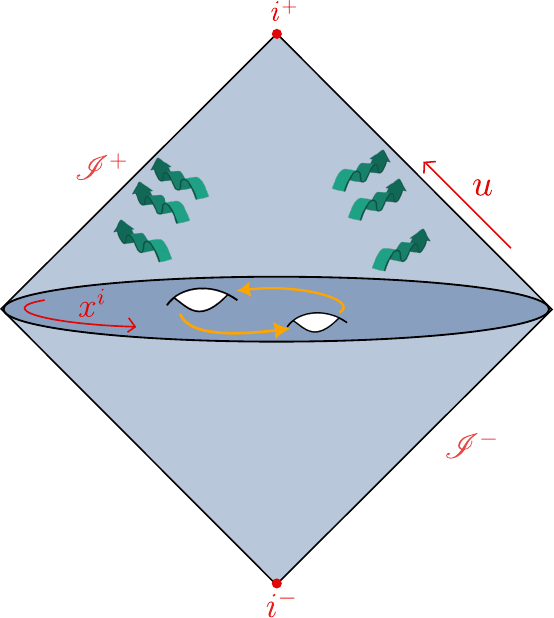}
\captionsetup{width=0.9\textwidth}
\caption{A schematic depiction of the conformal compactification $(\hat{g},\hat{M})$ of 
an asymptotically flat spacetime $(g,M)$: the gravitational radiation which has been emitted eventually ends up on the conformal boundary at null infinity.}
\label{fig: radiative spacetime}
\end{figure}

\begin{Definition}\label{defcuts}
    The sections $s:S^d\hookrightarrow\mathscr{I}_{d+1}$ of the ray  bundle $\pi:\mathscr{I}_{d+1}\twoheadrightarrow S^d$, i.e. the maps $s$ from the sphere $S^d$ into (either future or past) null infinity $\mathscr{I}_{d+1}$ which are such that $\pi\circ s=id_{S^d}$, are called \textit{cuts}. They can be identified with the transverse spheres $S^d$ inside null infinity $\mathscr{I}_{d+1}$. \
\end{Definition}

\noindent An example of cut is illustrated in Figure \ref{fig: section}.
Some other concepts will be useful later on.

\begin{figure}[h]
\centering
\includegraphics[width=0.45\textwidth]{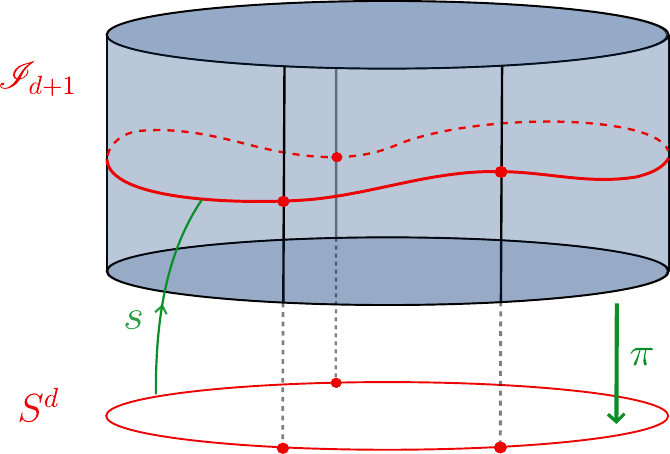}
\captionsetup{width=0.9\textwidth}
\caption{A cut $s$ of $\mathscr{I}_{d+1}$}
\label{fig: section}
\end{figure}

\subsubsection*{Invariant lift of a function}

Consider a principal $\mathbb{R}$-bundle $\pi : \mathscr{M} \twoheadrightarrow \bar{\mathscr{M}}$ with fundamental vector field $n$. An \textit{invariant function} is a function  $f:\mathscr{M}\rightarrow\mathbb{R}$ on the total space $\mathscr{M}$ which is constant along the orbits, i.e. $\mathcal{L}_n f = 0$. The set of such functions will be denoted $C^\infty_\text{inv}(\mathscr{M})$. In the adapted coordinates $(u, x^i)$, the invariance property simply means that the function does not depend on the vertical coordinate $u$. For any invariant function $f:\mathscr{M}\rightarrow\mathbb{R}$ on the total space, one can thus see that there exists a function $\bar{f}:\bar{\mathscr{M}}\to\mathbb{R}$ on the base space such that
\begin{equation}
    f = \pi^* \bar{f} = \bar{f} \circ \pi \,.
\end{equation}
From these definitions, it is then clear that there exists an isomorphism between the algebras of invariant functions $f$ on the total space $\mathscr{M}$ and of functions $\bar f$ on the base space $\bar{\mathscr{M}}$
\begin{equation}
    C^\infty_\text{inv}(\mathscr{M}) \cong C^\infty (\bar{\mathscr{M}}).
\end{equation}

\subsubsection*{Projection on the base space}

On a principal $\mathbb{R}$-bundle $\pi: \mathscr{M} \rightarrow \bar{\mathscr{M}}$ with fundamental vector field $n$, two types of vector fields can be introduced:
\begin{itemize}
\item \textit{Projectable vector field}: $X \in \mathfrak{X}(\mathscr{M})$ such that $\mathcal{L}_{n} X=f\cdot n$, where $f \in C^{\infty}(\mathscr{M})$,
\item  \textit{Invariant vector field}: $X \in \mathfrak{X}(\mathscr{M})$ such that $\mathcal{L}_{n} X=0$.
\end{itemize}

\begin{figure}[h]
\centering
\captionsetup{width=0.9\textwidth}
\includegraphics[width=2in]{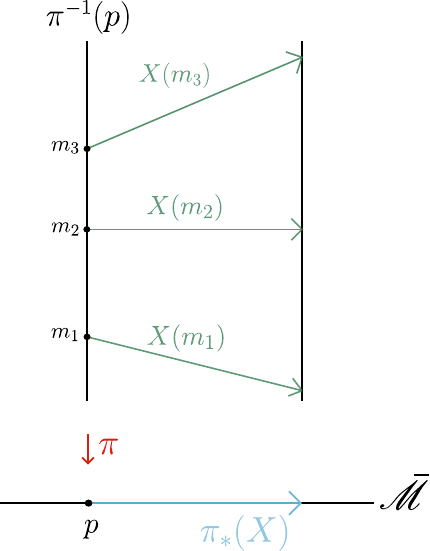}
\caption{
Values of a projectable vector field $X$ on $\mathscr{M}$ at various points $m_i\in\pi^{-1}(p)$ on the fibre above a base point $p\in\bar{\mathscr{M}}$}
\label{figure3}
\end{figure}

\begin{Remark}
Projectable vector fields are infinitesimal automorphisms of the fibre bundle. In fact, in the adapted coordinate system $(u, x^i)$ introduced above, they are of the form 
    \begin{equation*}
        X= F(u,x)\frac{\partial}{\partial u} +G^i(x)\frac{\partial}{\partial x^i}\,,
    \end{equation*}
    so they generate the coordinate transformations
    \begin{equation*}
        u' = u + \epsilon \ F(u,x), \quad x' = x + \epsilon \ G(x)\, ,
    \end{equation*}
where $\epsilon$ is an infinitesimally small parameter, $F\in C^\infty(\mathscr{M})$ and $G^i\in C_{\text{inv}}^\infty(\mathscr{M})$,

Intuitively, a projectable vector field is a vector field on the total space $\mathscr{M}$ which is such that there is no horizontal inconsistency when projected onto the base space $\bar{\mathscr{M}}$, cf. Figure \ref{figure3}. The only loss of information upon projection is the vertical part along $n$.
\end{Remark}

\begin{Remark}
 Since invariant vector fields commute with the fundamental vector field, they generate flows preserving the $\mathbb R$-action on the fibre. In other words,  invariant vector fields are locally of the form:
 \begin{equation*}
        X\,=\, f(x)\,\frac{\partial}{\partial u} \,+\,g^i(x)\,\frac{\partial}{\partial x^i}\,,
    \end{equation*}
 and they generate automorphisms of the principal $\mathbb{R}$-bundle:
\begin{equation}
    u' = u + \epsilon \ f(x), \quad x' = x + \epsilon \ g(x),
\end{equation}
where $f$ and $g^i \in C^\infty(\bar{\mathscr{M}})$. Note that any invariant vector field is projectable.
\end{Remark}

\begin{Example}\label{supertr}
To anticipate the use that will be made of these vector fields, let us consider an application of invariant vertical vector fields, $X = \mathcal{T} \cdot n$ with $\mathcal{T}\in C_\text{inv}^\infty(\mathscr{I}_{d+1})$, in the BMS context. In this case, these particular vector fields generate the \textit{supertranslations}\footnote{We here use the term supertranslation interchangeably with vertical automorphism. There is a slight subtlety in how these interact with base diffeomorphisms in the case of the BMS group at null infinity. This is due to the fact that the vertical vector field $n$ actually transforms under Weyl transformations of the Carrollian metric as a density of conformal weight minus $1$, see \eqref{Weak conformal Carrollian structure}. Accordingly the function $\mathcal{T}$ above ought to transform as a conformal density of weight $1$ (see Equation \eqref{BMSb}) and the complete BMS action will be given by \eqref{BMS group action on coordinates}.}, i.e. the vertical automorphisms of the ray bundle $\mathscr{I}_{d+1}\simeq \mathbb{R}\times S^d$:
\begin{equation}
    u' = u + \mathcal{T}(x), \quad x' = x.
\end{equation}
These transformations map a section of the principal bundle to another section, i.e. a cut to another one. There are as many supertranslations as functions $\bar{\mathcal{T}} \in C^\infty(S^d)$. The situation is outlined in Figure \ref{fig:supertranslation maps cut to cut}, where a cut $s$ is mapped to another cut $\hat{s}$ by a supertranslation. In fact, any two cuts are related by a unique supertranslation, so that there is an isomorphism between the affine space $\Gamma(\mathscr{I}_{d+1})$ of cuts and the vector space $C^\infty(S^d)$ of supertranslations. 
\end{Example}

\begin{figure}[h]
\centering
\captionsetup{width=0.9\textwidth}
\includegraphics[width=2.5in]{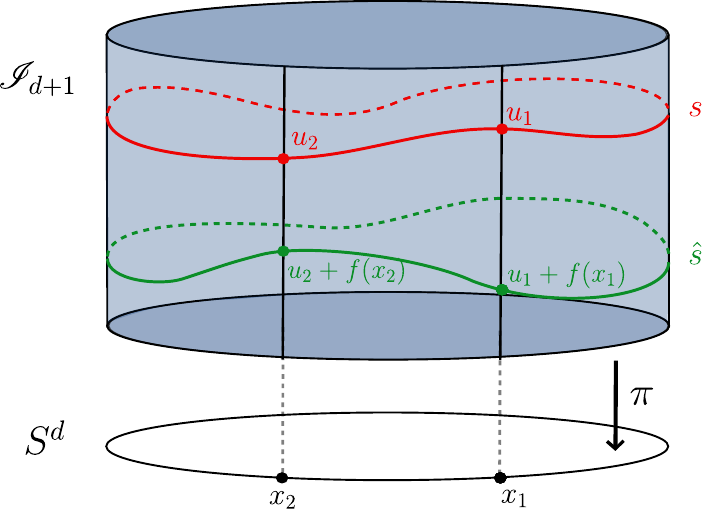}
\caption{The cuts $s$ and $\hat{s}$ are related by a supertranslation.}
\label{fig:supertranslation maps cut to cut}
\end{figure}

\subsection{Carrollian geometry}\label{Carrolliangeometry}

It is well known that the Galilean group can be obtained as a non-relativistic limit, more precisely an İnönü-Wigner contraction, of the Poincaré group. There is, however, another natural possible kinematical group \cite{Bacry:1968zf,Figueroa-OFarrill:2017tcy} which can be obtained from 
Poincaré by an İnönü-Wigner contraction: this is the Carroll group  \cite{Levy-Leblond:1965dsc,SenGupta:1966qer}, named by L\'evy-Leblond in reference to Lewis Carroll  for the apparent craziness of the causality relationship in the corresponding Carrollian spacetimes (see Figure \ref{Redqueen}).\\

We will here review some basic elements of Carrollian geometries, having in mind to highlight their close relationship with the BMS group and its Poincaré subgroups. For original references on Carrollian geometries, see \cite{Duval:2014uoa,Bekaert:2015xua,Morand:2018tke} for careful comparison of Carrollian versus Galilean geometries and 
\cite{Figueroa-OFarrill:2018ilb,Figueroa-OFarrill:2022nui} for the corresponding homogeneous models (see also \cite{Hartong:2015xda,Figueroa-OFarrill:2022mcy,Campoleoni:2022ebj} for the curved Cartan geometry of Carroll gravity). For original results on conformal Carrollian geometries and their relationship to null infinity, see e.g. \cite{Penrose:1962ij,Geroch1977,Ashtekar:1987tt,Ashtekar:2014zsa} and references therein (see also \cite{Herfray:2021qmp} for a description of how these relate to Carrollian homogeneous models).

\begin{figure}
\centering
\captionsetup{width=0.9\textwidth}
    \includegraphics[width=0.4\textwidth]{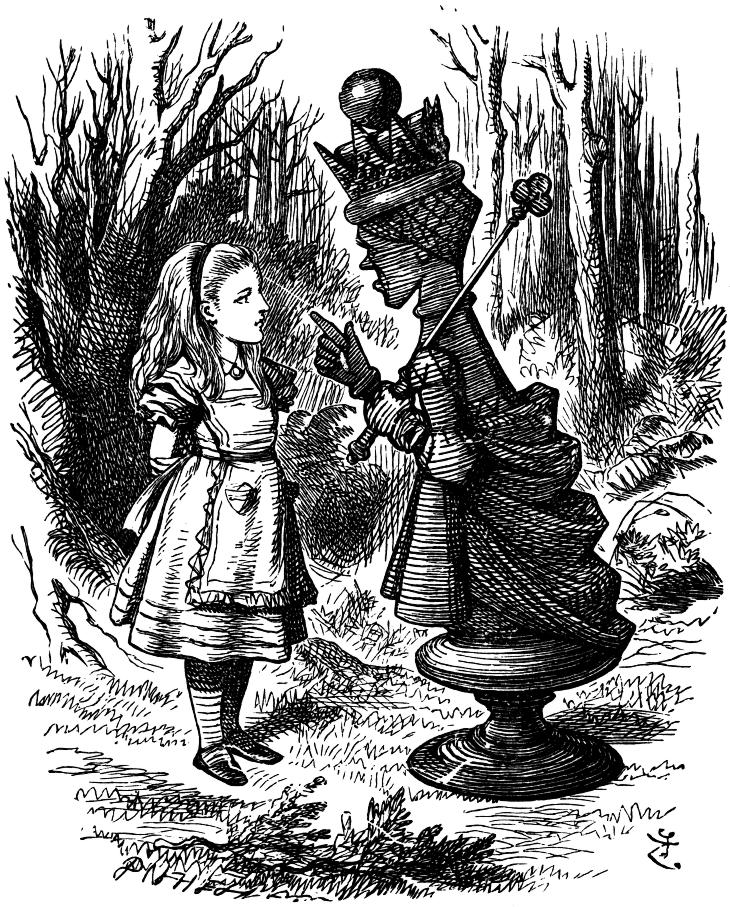}
    \caption{\label{Redqueen}The Red Queen lecturing Alice: ``You may call it `nonsense' if you like, but \emph{I}'ve heard nonsense, compared with which that would be as sensible as a dictionary!'' (Illustration by John Tenniel \cite{Alice2})}
\end{figure}

\subsubsection{Temporal structure of Carrollian geometry}

In the flat Carroll spacetime, \textit{all inertial observers are at rest}. In particular, this means that there is no possibility to exchange information, except via Carrollian tachyons, between these observers. This property is another motivation (according to Freeman Dyson \cite{Dyson:1972sd}) for the nickname ``Carroll'' and is captured by the Red Queen of the chessboard in the novel ``Through the Looking-Glass'' \cite{Alice2}:
\begin{center}
\textit{``Now, here, you see, it takes all the running you can do,\\ to keep in the same place.''}
\end{center}
Accordingly, in Carrollian geometry the 
worldlines of observers only belong to a fixed congruence of curves in spacetime. These curves can be thought of as the vertical orbits of a principal $\mathbb{R}$-bundle. More precisely, the timelike structure of a Carrollian spacetime $\mathscr{M}$ is defined as follows.

\begin{Definition} 
A \emph{field of observers} is a nowhere-vanishing vector field $n = n^\mu \partial_\mu$ on the spacetime $\mathscr{M}$.
\end{Definition}

As discussed in the previous section (and under the same assumptions), the timelike structure endows the spacetime manifold $\mathscr{M}$  with a structure of principal $\mathbb{R}$-bundle over the spatial manifold $\bar{\mathscr{M}}$, with fundamental vector field $n$. Furthermore, it provides a distinction between various types of vectors in Carrollian geometry:
\begin{align}
\left\{
\begin{array} { lll } 
{ V ^ { \mu } = f n ^ { \mu } } &\text { with } \quad &\left\{\begin{array}{ll}
f \neq 0 & \textit{ Timelike}, \\
f>0 & \textit{Future-oriented}, \\
\end{array}\right.
\\[1.5em]
{ V ^ { \mu } \neq f n ^ { \mu } }& \hfill &\textit{ Spacelike},
\end{array}
\right.
\end{align} 
with $f \in C^\infty(\mathscr{M})$. The integral lines of the field of observers $n$ are called \textit{Carrollian worldlines}. The fibration is nothing but the congruence of the Carrollian worldlines of inertial observers. In a Carrollian spacetime, all observers are at rest and are causally disconnected since all light-cones close in the $c\to 0$ limit.

\subsubsection{Carrollian relativity}

Any affine parameter $u$ of the field of observers $n$ (i.e. $n = \partial/ \partial u$) is called a \textit{Carrollian time}. One should note that it is defined up to shifts $u'=u+f(x)$ corresponding to the freedom to set the departure of the clock for every observer.

The image of a section $\iota:\bar{\mathscr{M}}\hookrightarrow\mathscr{M}$ of the fibration $\pi : \mathscr{M} \twoheadrightarrow \bar{\mathscr{M}}$ is a spacelike submanifold $\iota(\bar{\mathscr{M}})\subset\mathscr{M}$ of the spacetime manifold. These spacelike submanifolds can be thought of as simultaneity slices for a given choice of Carrollian time. The vertical automorphisms of the principal $\mathbb R$-bundle take the form $u'=u+f(x)$. They map simultaneity slices to each other. \\ 

The \textit{Carrollian relativity principle} could be summarised as follows: simultaneity is relative (leading to a democracy principle: all simultaneity slices are equal) but rest is absolute.
This contrasts with its dual limit, the Galilean ($c\to\infty$) limit, where the \textit{Galilean relativity principle} states that: motion is relative (leading to the democracy principle: all inertial observers are equal) but simultaneity is absolute.\\

Note that the notions of observers and clocks decouple in both the Carrollian and Galilean limits, and are
respectively relative or absolute depending on the limit under consideration.

\begin{Example}\label{example simultaneity slices}
At null infinity $\mathscr{I}_{d+1}$, the simultaneity slices described above coincide with its cuts (cf. Definition \ref{defcuts}). Since supertranslations are the vertical automorphisms of the associated principal $\mathbb R$-bundle (cf. Example \ref{supertr}), they map one cut to another, and thus one simultaneity slice to another.
\end{Example}

\subsubsection{Spatial structure of Carrollian geometry}

The spatial world imagined by Carroll on the other side of the looking-glass is a chessboard \cite{Alice2} which provides, in effect, a collection of rulers for its inhabitants, i.e. a metric along the spacelike directions.

\begin{Definition} \label{Carrmetr}
A \emph{Carrollian metric} is a positive semi-definite metric $\gamma$ on the spacetime $\mathscr{M}$ whose kernel is spanned by the fundamental vector field
\begin{equation}
\left\{\begin{array}{l}
\gamma_{\mu \nu} V^\mu V^\nu \geqslant 0 \\
\gamma_{\mu \nu} V^\mu=0 \quad \Leftrightarrow \quad V^\mu=f n^\mu.
\end{array}\right.
\end{equation}	
\end{Definition}

\begin{Remark}
There is a one-to-one correspondence\footnote{This one-to-one correspondence holds precisely because the Carrollian metric is \textit{basic}, i.e. both invariant ($\mathcal{L}_n\gamma_{\mu\nu}=0$) and horizontal ($\gamma_{\mu \nu}n^\nu=0$)}. between \textit{invariant} Carrollian metrics $\gamma_{\mu\nu}$  on $\mathscr{M}$ (i.e. $\mathcal{L}_n\gamma_{\mu\nu}=0$) and Riemannian metrics $\bar{\gamma}_{ij}$ on the base $\mathscr{\bar{M}}$. In particular, $\gamma = \pi^*\bar{\gamma}$. 
\end{Remark}

\subsubsection{Weak Carrollian structure}

The temporal and spatial structures of a Carrollian spacetime have been defined, which can be summarised as follows (see \cite{Henneaux:1979vn} for the early version with different terminology).

\begin{Definition}\label{weakCarroll}
[Henneaux, 1979] A \emph{Carrollian structure} is composed of two pieces of data: a field of observers $n$ and a Carrollian metric $\gamma$.
\end{Definition}
\noindent 
In the remainder of these notes, we will focus on \emph{invariant} Carrollian structures (i.e. such that $\gamma$ is invariant).   

\begin{Example}\label{flatCarroll}
An example of these structures is given by the \textit{flat Carroll spacetime}, $\mathbb{R}\times \mathbb{R}^{d}$. In this case, the field of observers is given by $n = \partial/ \partial u$ and the flat Carrollian metric is equal to the pullback of the metric on the Euclidean space
 \begin{equation}\label{Flat carroll metric}
     ds^2_{\mathbb{R}\times\mathbb{R}^{d}}=\delta_{ij}dx^idx^j = d\ell^2_{\mathbb{R}^{d}}\,,
 \end{equation}
 where the notation $d\ell^2$ is meant to distinguish positive-definite, Riemannian line elements from Lorentzian or Carrollian line elements written $ds^2$.
\end{Example} 

\begin{Example}\label{ScriCarroll}
Another example of Carrollian structure is given by future null infinity, $\mathscr{I}^+_{d+1}\simeq\mathbb{R}\times S^{d}$. Making use of the notations of Example \ref{Example: radiative spacetimes}, the pair $(n,\gamma)$ is obtained by restricting the normal $n=\hat{g}^{-1}d \omega\big|_{\omega=0}$ and taking the pull-back $\gamma=\iota^*\hat{g}$ of the unphysical metric  along $\mathscr{I}^+_{d+1} = \{\omega=0\}$. The resulting geometry is in fact always identical to that of flat null infinity, see Example \ref{exampleflatscri}. In this case, the Carrollian time coincides with the retarded time $u$, thus $n =\partial_u$, and the Carrollian metric is equal to the pullback 
\begin{equation}\label{Scricarroll metric}
     ds^2_{\mathscr{I}_{d+1}}=\bar\gamma_{ij}dy^idy^j = d\ell^2_{S^{d}}
 \end{equation}
 of the round metric on the celestial sphere $S^d$. Alternatively\footnote{This really is an equivalent description because null infinity $\mathscr{I}_{d+1}$ is a conformal Carrollian geometry (see Section \ref{ConformalCarrolliangeometry}). One essentially does not lose any information by working in this patch as long as one makes sure that all constructions are $SO(d+1,1)$-invariant.}, one can work in a patch ($\simeq\mathbb{R}\times\mathbb{R}^d$) covered by adapted stereographic coordinates $(u,x^i)$ together with the flat metric \eqref{Flat carroll metric}. We will frequently use this explicit description in what follows.
\end{Example}

\begin{Definition}
 A \emph{Carrollian isometry} is a diffeomorphism of $\mathscr{M}$ preserving the field of observers, $n^\prime = n$, and the Carrollian metric, $\gamma^\prime = \gamma$. 	
\end{Definition}
\noindent For an invariant Carrollian structure, these Carrollian isometries of $\mathscr{M}$ project onto isometries of the Riemannian metric on the base $\bar{\mathscr{M}}$ (in the sense that the image by the projection $\pi:\mathscr{M}\twoheadrightarrow\bar{\mathscr{M}}$ of a diffeomorphism of $\mathscr{M}$ preserving $n$ and $\gamma$ is a diffeomorphism of $\bar{\mathscr{M}}$ preserving $\bar\gamma$). 

\begin{Example}
Concrete and simple  examples of Carrollian isometries are the vertical automorphisms of the principal $\mathbb{R}$-bundle,
\begin{equation}
    u^\prime = u + f(x), \quad x^\prime = x, \label{vertical automorphism}
\end{equation}
that we outlined in Section \ref{principRbundes}, and which are interpreted as supertranslations in the BMS context. One can convince oneself by checking that the generator $X = f(x) \frac{\partial}{\partial u}$ indeed preserves the fundamental vector field $n$ and the invariant metric $\gamma$.
\end{Example}

\subsubsection{Strong Carrollian structure}

In the relativistic case, thanks to the uniqueness of the Levi-Civita connection, parallel transport is completely determined by the metric structure of spacetime. In the Galilean case, since the cometric is degenerate, there is no unique affine connection\footnote{\label{torsionless}Note that the affine connections will always be implicitly assumed torsionless in these notes (see, however, Appendix A of \cite{Bekaert:2015xua} for a systematic study of torsionful strong Carrollian structures).} compatible with the Galilean structure. A stronger structure, including the gift of an affine connection, is therefore needed in order to introduce the gravitational field and to describe the motion of free-falling observers in a Galilean spacetime. For a Carrollian spacetime, since the metric is degenerate (cf. Definition \ref{Carrmetr}), one also needs to introduce a stronger structure in order to completely describe the Carrollian geometry.  

To define a general Carrollian structure, we have seen earlier that one needs two objects: a field of observers $n$  and a Carrollian metric $\gamma$. The principal $\mathbb{R}$-bundle $\mathscr{M}$ where this structure lives can further be endowed with a \textit{compatible affine connection}. This defines a strong Carrollian structure, first discussed in this context in \cite{Duval:2014lpa}, which can be summarised as follows.

\begin{Definition}[Duval-Gibbons-Horvathy-Zhang, 2014] \label{strongCarroll}
	 A \textit{strong Carrollian structure} is composed of three data: a field of observers $n$, a Carrollian metric $\gamma$ and an affine connection $\nabla$ which is compatible, in the sense that $\nabla n=0=\nabla\gamma$.
\end{Definition}

\begin{Example}\label{flatCarrspacetime}
The weak Carrollian structure of the flat Carrollian spacetime (see Example \ref{flatCarroll}) can be supplemented with a flat affine connection, which reads $\Gamma^\rho_{\mu \nu}=0$ in some adapted coordinates $x^\mu=(u,x^i)$.
\end{Example}

Two basic facts about strong Carrollian structures are:
\begin{itemize}
    \item If the Carrollian metric is not invariant ($\mathcal{L}_n\gamma \neq 0$), then there exists no associated strong structure (this is only true because the affine connection is assumed torsionless, cf. Footnote \ref{torsionless}).
    \item $\nabla$ and $\nabla'$ are two affine connections compatible with the same Carrollian structure $(n,\gamma)$ if and only if
\begin{equation}\label{Change of strong Carrollian geometry}
     (\Gamma^\rho_{\mu \nu})'=\Gamma^\rho_{\mu \nu}+ S_{\mu\nu}\,n^\rho
\end{equation}
for some symmetric horizontal tensor field, i.e. $S_{\mu\nu}=S_{\nu\mu}$ and $S_{\mu\nu}n^\nu=0$. In particular, if $(\Gamma^\rho_{\mu \nu})'=0$ is the flat connection in an adapted coordinate system $x^\mu=(u,x^i)$, then $S_{\mu\nu} = \nabla_{\mu} \nabla_{\nu} u$.
\end{itemize}   

\begin{Example}\label{Strong carroll at null infinity}
For any asymptotically flat spacetime\footnote{See Example \ref{Example: radiative spacetimes}.} $(g,M)$, the Levi-Civita connection $\hat{\nabla}$ of the unphysical spacetime $(\hat{g},\hat{M})$ induces, at null infinity, a strong Carrollian structure, see e.g. \cite{Geroch1977}. In practice, this induced connection $\nabla$ is determined by \eqref{Change of strong Carrollian geometry} with $\nabla'$ the flat connection in adapted coordinates $(u,x^i)$ and $S$ the pullback to null infinity of the asymptotic shear of the null geodesic congruence generated by the null vector field $\ell=g^{-1}\nabla u\big|_{\mathscr{I}}$. For a four-dimensional asymptotically flat spacetime (see Figure \ref{fig: radiative spacetime}), such a Carrollian connection encodes\footnote{However at the price of breaking conformal invariance, see Example \ref{example: strong conformal carroll geometry at scri} for the corresponding conformally invariant description.} gravitational radiation \cite{Ashtekar:1981hw,Ashtekar:1981sf}.
\end{Example}

\begin{Definition}[Duval-Gibbons-Horvathy-Zhang, 2014]\label{Def: strong carrollian isom}
	A \emph{strong Carrollian isometry} is a diffeomorphism of $\mathscr{M}$ preserving the field of observers $n^\prime = n$, the Carrollian metric $\gamma^\prime = \gamma$ and the affine connection $\nabla^\prime = \nabla$.
\end{Definition}

\begin{Example} Consider the flat strong Carrollian structure on $\mathbb{R}\times\mathbb{R}^d$ (see Example \ref{flatCarrspacetime}).
	As is well-known, preserving the flat affine connection $\Gamma^\rho_{\mu\nu} = 0$ implies that the second derivatives of the transformations are zero. This last condition restricts the transformations to be affine in the coordinates, thereby leaving only the Carrollian time translations and Carrollian boosts
	\begin{equation}
    u^\prime = u + a + b_i\, x^i, \quad x^\prime = x,
\end{equation}
among the vertical isometries. The condition to preserve the flat affine connection removes the arbitrariness of the vertical automorphism given in Equation \eqref{vertical automorphism}. The remaining isometries arise from the base isometries: the spatial translations and rotations of $\mathbb{R}^d$.  
\end{Example}

The Carrollian isometries can also be described via their infinitesimal counterparts.

\begin{Definition}
	A (strong) \emph{Carroll-Killing vector field} is a vector field $X\in\mathfrak{X}(\mathscr{M})$ on the Carrollian spacetime such that
\begin{equation*}
{\cal L}_X n=0\,,\quad {\cal L}_X\gamma=0 \quad(\text{and}\quad {\cal L}_X\nabla=0\quad\text{in the strong case}).
\end{equation*}
\end{Definition}
Note that ${\cal L}_X n=0\Leftrightarrow{\cal L}_n X=0$, so a Carroll-Killing vector field is necessarily invariant.

\begin{Remark}
    In general, a strong Carrollian geometry will not have any isometry. On the other hand, the isometry algebra of a strong Carrollian geometry has maximal dimension if and only if there exists a local coordinate system where the connection is flat, i.e. $\Gamma=0$. This can be proved, for example, by Cartan's method of equivalence \cite{Herfray:2021qmp}.
\end{Remark}

\begin{Example}
	The Lie algebra of strong Carroll-Killing vector fields on flat Carrollian spacetime is the finite-dimensional Carroll algebra $\mathfrak{carr}(d+1)$. The non-trivial commutation relations of the generators of the Carroll algebra are
\cite{Levy-Leblond:1965dsc,SenGupta:1966qer} 
\begin{subequations}\label{Carroll}
\begin{align}
[\,P_i\,,\,B_j\,]=\delta_{ij}\,H\,, \\[0pt]
[\,P_i\,,\,J_{jk}\,]=\delta_{i[j}\,P_{k]}\,,\\[0pt] 
[\,B_i\,,\,J_{jk}\,]=\delta_{i[j}\,B_{k]}\,,\\[0pt]
[\,J_{ij}\,,\,J_{kl}\,]=\delta_{[i[k}J_{l]j]}\,,
\end{align}
\end{subequations}
with $J_{ij}=x_{[i}\partial_{j]}$, $P_i=\partial_i$, $H=\partial_u$, $B_i=x_i\partial_u$, which are respectively the generators of spatial rotations, spatial translations, time translations, and Carrollian boosts. As one can see the Carroll algebra has the structure of a semidirect sum $\mathfrak{carr}(d+1)=\mathfrak{so}(d)\inplus\mathfrak{h}_d$, of the Heisenberg algebra $\mathfrak{h}_d$ generated by the generators $\{H, P_i, B_j\}$, on which the $\mathfrak{so}(d)$, spanned by the generators $\{J_{ij}\}$, acts via rotations of the indices.
\end{Example}

\subsection{Conformal Carrollian geometry}\label{ConformalCarrolliangeometry}

The conformal analogue of a Carrollian structure has been introduced in \cite{Duval:2014lpa} (see \cite{Penrose:1962ij,Geroch1977} for its seminal version with different terminology). 

\subsubsection{Weak conformal Carrollian structure}

The conformal analogue of a Carrollian structure can be defined as follows.\footnote{\label{Foot: on conformal carroll}As emphasised in \cite{Duval:2014lpa}, the notion of a conformal Carrollian structure is more general than that given in Definition \ref{PenroseGeroch}, since the field of observers $n$ and Carrollian metric $\gamma$ may transform independently owing to the degeneracy of the latter. The case considered here
corresponds to the Carrollian limit of a relativistic conformal structure.}

\begin{Definition}\label{PenroseGeroch}
[Penrose, 1962; Geroch, 1977] A \emph{conformal Carrollian structure }is an equivalence class $[n, \gamma]$ of Carrollian structures, i.e. pairs $(n, \gamma)$, with respect to the
following equivalence relation: 
\begin{equation}\label{Weak conformal Carrollian structure}
    n \sim \Omega^{-1} \,n \text{ and } \gamma \sim \Omega^2 \,\gamma\,,
\end{equation}
for any positive function $\Omega >0$, called a \textit{conformal factor}.
\end{Definition}

\begin{Example}\label{example: weak conformal carroll geometry at scri}
As discussed in Examples \ref{exampleflatscri} and \ref{ScriCarroll}, a choice of unphysical metric $\hat{g}$ for an asymptotically flat spacetime $(g=\omega^{-2}\hat{g},M)$ induces a weak Carrollian geometry $(n,\gamma)$ along null infinity. Since the unphysical metric is only determined up to a Weyl transformation $\hat{g} \mapsto \hat\Omega^2 \hat{g}$, $\omega \mapsto \hat{\Omega} \omega$, where in general $\hat\Omega = \Omega + \mathcal{O}(\omega)$, this rescaling changes the induced Carrollian geometry as in Definition \ref{PenroseGeroch} and therefore preserves the equivalence class $[n, \gamma]$. Einstein's equations imply that this weak conformal Carrollian structure can be chosen to be\footnote{Namely Einstein's equations imply, at leading order, that $\mathcal{L}_n \gamma \propto \gamma$. However, since $\Omega$ is here \emph{a priori} not invariant, one can use this freedom to choose a representative $\gamma$ such that $\mathcal{L}_n \gamma=0$. This representative $\gamma$ is unique, up to rescaling by invariant functions $\Omega$.} invariant: $\mathcal{L}_n \gamma =0$.
\end{Example}

Now, one can introduce the corresponding symmetries preserving the latter conformal structure.	
\begin{Definition}\label{confCarriso}
A \emph{conformal Carrollian isometry} is a diffeomorphism of $\mathscr{M}$ such that
the conformal Carrollian structure $[n, \gamma]$ remains unchanged, i.e.
\begin{equation*}
    n^\prime = \Omega^{-1} n\  \text{ and } \gamma^\prime = \Omega^2 \gamma,
\end{equation*}
for a certain conformal factor $\Omega > 0$.
\end{Definition}

\begin{Remark}\label{invariantconformalfactor}
	For an invariant Carrollian metric ($\mathcal{L}_n\gamma=0$), these conformal Carrollian isometries of the Carrollian structure on ${\mathscr{M}}$ project onto conformal isometries of the Riemannian metric on the base $\bar{\mathscr{M}}$. In particular, the conformal factor is invariant, $\mathcal{L}_n \Omega = 0$ (see Appendix \ref{invproof} for a proof).  As a conclusion, the conformal Carrollian isometries of an invariant structure must be of the form
\begin{align}\label{BMS group action on coordinates}
    u' &= \Omega(x)\Big( u + \mathcal{T}(x)\Big), & x^{\prime\,i}=G^i(x)&,
\end{align}
where $G:x  \mapsto x'$ is a conformal isometry of the base with $\bar\gamma' = \bar\Omega^2 \bar\gamma$ and $\Omega=\pi^*\bar\Omega$.
\end{Remark}

\begin{Example}\label{BMStransfosconfCarr}
At null infinity, we have that \cite{Geroch1977,Duval:2014uva}
\begin{center}
     \textit{BMS transformations = Conformal Carrollian isometries},
\end{center}	
so we take the point of view here to define the BMS symmetries intrinsically on the boundary as conformal Carrollian isometries of null infinity endowed with its (universal \cite{Geroch1977,Ashtekar:1987tt}) conformal Carrollian structure.
\end{Example}

\begin{Definition}\label{confCarrollKilling}
A \emph{conformal Carroll-Killing vector field} is a vector field $X\in\mathfrak{X}(\mathscr{M})$ such that
\begin{equation*}
{\cal L}_X n=-f\,n\,,\quad {\cal L}_X\gamma=2f\,\gamma\,,
\end{equation*}
with $f\in C^\infty(\mathscr{M})$. 
\end{Definition}

\begin{Remark}\label{infversionremark}
As one can see,  conformal Carroll-Killing vector fields on $\mathscr{M}$ are always projectable onto $\bar{\mathscr{M}}$. Furthermore,
for an invariant Carrollian metric, one has $f\in C_{\text{inv}}^\infty(\mathscr{M})$, i.e. $\mathcal{L}_n f = 0$, and any conformal Carroll-Killing vector field $X$ on $\mathscr{M}$ projects onto a conformal Killing vector field $\bar{X}$ on $\bar{\mathscr{M}}$, i.e. ${\cal L}_{\bar{X}}\bar{\gamma}=2\bar{f}\,\bar{\gamma}\,$. This is the infinitesimal version of Remark \ref{invariantconformalfactor} (see also Appendix \ref{invproof} for an independent proof). 
\end{Remark}
In the adapted coordinates $(u,x^i)$, the general expression of a conformal Carroll-Killing vector field is 
\begin{equation}
        X = Y + \Big( \tfrac{1}{d} \bar\nabla_i Y^i\, u + \mathcal{T}(x)\Big) \frac{\partial}{\partial u},
\end{equation}
where $\bar\nabla$ is the Levi-Civita connection of the metric $\bar\gamma$, and the conformal Carroll-Killing vector field $Y = Y^i(x) \frac{\partial}{\partial x^i}$ on $\mathscr{M}$ defines a conformal Killing vector field $\bar Y = Y^i(x) \frac{\partial}{\partial x^i}$ on $\bar{\mathscr{M}}$. 

\begin{Example}
The Lie algebra of conformal Carroll-Killing vector fields on null infinity $\mathscr{I}_{d+1}$ is isomorphic to the BMS algebra $\mathfrak{bms}_{d+2}$ (cf. Example \ref{BMStransfosconfCarr}). Introducing the infinitesimal generators, $J(\bar Y) = Y^i(x)\,\partial_i +\tfrac{1}{d} \bar\nabla_i Y^i(x)\, u \,\partial_u$ of conformal isometries of the celestial sphere $S^d$, and $\mathcal{T}(x) \,\partial_u$ of supertranslations, we have the commutation relations
\begin{subequations}\label{BMS}
\begin{align}
[\,J(\bar Y_1)\,,\,J(\bar Y_2)\,]&= J\Big( \{\bar Y_1,\bar Y_2\}\Big)\,,  \label{BMSa}\\[0pt]
[\,J(\bar Y)\,,\,\mathcal{T}\,\partial_u\,]&= \Big( Y^i \partial_i \mathcal{T} -\tfrac{1}{d}\,\bar\nabla_i Y^i\,\mathcal{T} \Big) \partial_u \,,\label{BMSb}\\[0pt]
[\,\mathcal{T}_1\,\partial_u \,,\,\mathcal{T}_2\,\partial_u\,]&=0\,.
\end{align}
\end{subequations}
where the curly bracket on the right-hand side of \eqref{BMSa} stands for the Lie bracket of vector fields.
Note in particular that, from Equation \eqref{BMSb}, one can see that the function $\mathcal{T}\in C^\infty(S^d)$ parametrising supertranslations should really be thought of as a conformal primary field of scaling dimension $-1$ since it has a nontrivial transformation law under $\mathfrak{so}(d+1,1)$.
\end{Example}

\subsubsection{Strong conformal Carrollian structure}\label{Strong_conformal_Carrollian_structure}

The following conformal analogue of the notion of strong Carrollian structure was introduced in \cite{Ashtekar:1981hw,Ashtekar:1981sf} a long time ago, from a different perspective and under a different name (see also \cite{Herfray:2020rvq,Herfray:2021qmp,Herfray:2021xyp} for an investigation on how these geometries precisely interpolate between Carroll and classical constructions in conformal geometry). 

\begin{Definition}[Ashtekar, 1981]\label{strongconformalCarrollianstructure}
	 A \textit{strong conformal Carrollian structure} is an equivalence class $[n, \gamma ,\nabla]$ of strong Carrollian structures, with respect to the
following equivalence relation :
\begin{align*}
    n^\rho &\sim \Omega^{-1}n^\rho, \qquad \gamma_{\mu\nu} \sim \Omega^2 \gamma_{\mu\nu}\,, &&\text{ for any } 0<\Omega \in C^\infty_\text{inv}(\mathscr{M})
    ,\\
    \Gamma^\rho_{\mu\nu} & \sim \Gamma^\rho_{\mu\nu} + 2\,\Omega^{-1}\partial^{}_{(\mu}\Omega \,\delta_{\nu)}^\rho- \gamma_{\mu\nu}\left( \Omega^{-1}\gamma^{\rho\sigma}\partial^{}_{\sigma}\Omega  + \Phi \,n^{\rho}\right) && \text{ for any } \Phi \in C^\infty(\mathscr{M}),
\end{align*} where the round bracket around two indices stands for  their symmetrisation with weight one, i.e. $S_{(\mu\nu)}=S_{\mu\nu}$ for a symmetric tensor and $\gamma^{\sigma\rho}$ is any\footnote{Since the metric is degenerate of kernel $n^{\rho}$, the inverse is not unique. Rather, two inverses $\gamma^{\sigma\rho}$ and $\left(\gamma^{\sigma\rho}\right)'$ must differ by a term of the form $\phi^{(\sigma} n^{\rho)}$ for some vector $\phi^\sigma$. As a result, this ambiguity can always be reabsorbed in the choice of $\Phi$.} inverse of $\gamma_{\mu\nu}$ i.e. any tensor such that $\gamma^{\nu\rho} \gamma_{\rho \mu} = \delta^{\nu}_{\mu}$.
\end{Definition}

\begin{Example}\label{example: strong conformal carroll geometry at scri}
As Ashtekar emphasised early on (see, e.g., \cite{Ashtekar:1981hw,Ashtekar:1981sf,Ashtekar:1987tt,Ashtekar:2014zsa}), in four dimensions the gravitational radiation reaching null infinity is precisely encoded by such a choice of strong conformal Carroll geometry: we saw in Example \ref{Strong carroll at null infinity} that any choice of unphysical metric $\hat{g}$ for the asymptotically flat spacetime $g=\omega^{-2}\hat{g}$ induces a strong Carrollian geometry $(n,\gamma,\nabla)$ along null infinity. However, the unphysical metric is only determined up to a Weyl transformation $\hat{g} \mapsto \hat{\Omega}^2 \hat{g}$, $\omega \mapsto \hat{\Omega}\omega$ where in general, $\hat{\Omega} = \Omega \left(1+ \Phi \omega + \mathcal{O}(\omega^2) \right)$. This rescaling changes the induced strong Carrollian geometry as in Definition \ref{strongconformalCarrollianstructure} and therefore preserves the equivalence class $[n, \gamma ,\nabla]$.
\end{Example}

\begin{Definition}\label{strongconformalCarrollianisometry}
	A \emph{strong conformal Carrollian isometry} is a diffeomorphism of $\mathscr{M}$ such that the strong conformal Carroll structure $[n,\gamma, \nabla]$ remains unchanged:
    \begin{align*}
    n^{\prime\rho} &=\Omega^{-1}n^\rho, \qquad \gamma_{\mu\nu}' =\Omega^2 \gamma_{\mu\nu}\,,\\
    \Gamma^{\rho\,\prime}_{\mu\nu} &= \Gamma^\rho_{\mu\nu}+ 2\,\Omega^{-1}\partial^{}_{(\mu}\Omega \,\delta_{\nu)}^\rho- \gamma_{\mu\nu}\left( \Omega^{-1}\gamma^{\rho\sigma}\partial^{}_{\sigma}\Omega  + \Phi \,n^{\rho}\right)\,,
\end{align*}
for certain functions $\Phi \in C^\infty(\mathscr{M})$ and  $0<\Omega \in C^\infty_\text{inv}(\mathscr{M})$.\end{Definition}

\begin{Example} Consider the conformally flat case of a strong conformal Carrollian structure, which is such that one representative is the strong Carrollian structure of Example \ref{flatCarrspacetime}.	Preserving the flat affine connection $\Gamma^\rho_{\mu\nu} = 0$ up to a trace, $(\Gamma^\rho_{\mu\nu})' = -\Phi \gamma_{\mu\nu}n^{\rho}$, implies that the trace-free part of the second derivatives of the transformations is zero. This restricts the vertical transformations to be affine up to a trace in the adapted coordinates, thereby leaving us with
	\begin{equation}
    u^\prime = u + a + \vec b \cdot\vec x + c\, |\vec x|^2, \quad \vec x^\prime = \vec x,
\end{equation}
as vertical automorphisms, i.e. time translations, Carroll boosts and temporal special conformal transformations. 
These vertical conformal isometries should be combined with the conformal isometries $x\mapsto x'$ of the base. The weaker condition of preserving the equivalence class of connections, rather than the complete affine connection, adds one extra vertical symmetry with respect to strong Carrollian isometries: the group of vertical Carrollian isometries extends from $\mathbb{R}^{d+1}$ to $\mathbb{R}^{1,d+1}$.
\end{Example}

The strong conformal Carrollian isometries can also be described via their infinitesimal counterparts.

\begin{Definition}
	A strong \emph{conformal Carroll-Killing vector field} is a vector field $X\in\mathfrak{X}(\mathscr{M})$ on the Carrollian spacetime such that
\begin{equation*}
{\cal L}_X n=-f\,n\,,\quad {\cal L}_X\gamma=2f\,\gamma\,, \quad\text{and}\quad {\cal L}_X\nabla= df\odot \delta-\gamma\otimes \text{grad}(f) +g\ \gamma\otimes n,
\end{equation*}
with $f\in C^\infty(\bar{\mathscr{M}})$, $g\in C^\infty(\mathscr{M})$.\footnote{We adopted again an index-free notation for brevity: $\odot$ denotes the symmetric tensor product and $\otimes$ the tensor product. Their definition is along the lines of the right-hand side expressions in Definition \ref{strongconformalCarrollianstructure}.}
\end{Definition}

In complete parallel with usual (non Carrollian) conformal geometry, a strong conformal Carrollian geometry will have no isometries in general. The isometry algebra of a strong conformal Carrollian geometry will have maximal dimension if and only if the equivalence class locally contains a flat connection. What is more, such conformal flatness of Carrollian geometry is equivalent in $d+1>3$ dimensions (respectively, $d+1=3$ and $d+1=2$) to the vanishing of a Carrollian version of the Weyl (resp. Cotton\footnote{This Carrollian Cotton tensor is equivalent to Newman-Penrose coefficients $\Psi_4^0$, $\Psi_3^0$, Im$\Psi_2^0$ (see \cite{Ashtekar:1981hw,Ashtekar:1981sf,Ashtekar:1987tt,Ashtekar:2014zsa} for the original result phrased in another terminology and \cite{Campoleoni:2023fug} for the relationship of this tensor with the Carroll limit of the usual Cotton tensor of conformal geometry).} and Möbius curvature\footnote{This Carrollian Möbius curvature is equivalent \cite{Herfray:2020rvq,Herfray:2021xyp} to the evolution equation induced on the mass and angular-momentum aspects by three-dimensional Einstein equations.}) tensor. A neat proof can again be obtained by the method of equivalence of Cartan \cite{Herfray:2021qmp}.

\begin{Example}\label{Poincaralgebra}
The Lie algebra of strong conformal Carroll-Killing vector fields on flat Carrollian spacetime is the finite-dimensional conformal Carroll algebra, denoted by $\mathfrak{confcarr}(d+1)$
whose generators obey the following commutation relations
\begin{subequations}
\begin{align}
[D, H] &= -H, & [D, K_i] &= K_i, & [K_i, P_j] &= -2\delta_{ij} D - 2J_{ij}, \\
[D, P_i] &= -P_i, & [K, P_i] &= -2B_i, & [K_i, B_j] &= -\delta_{ij} K, \\
[D, K] &= K, & [K_i, H] &= -2B_i, & [K_i, J_{jk}] &= \delta_{i[j} K_{k]}.
\end{align}
\end{subequations}
Here the new generators are the dilatations $D$, and the Carrollian (temporal and spatial) special conformal transformations $K$ and $K_i$. 
The generators $P_i$, $K_j$, $J_{ij}$ and $D$ of conformal transformations of $S^d$ span a Lie algebra isomorphic to the Lorentz algebra $\mathfrak{so}(d+1,1)$, whereas the generators $H$, $K$ and $B_i$ span an Abelian Lie algebra transforming in the vector representation $\mathbb{R}^{d+1,1}$ under the Lorentz algebra.
In the Cartesian coordinates on the Euclidean space $\mathbb{R}^d$ (i.e. the celestial sphere $S^d$ with one point removed) these translation generators read $H=\partial_u$, $\vec B=\vec x\,\partial_u$, $K=|\vec x|^2\partial_u$\,.  Altogether, the generators of the above conformal extension of the Carroll algebra span a Lie algebra isomorphic to the Poincar\'e algebra (see, e.g., Appendix A of \cite{Bekaert:2024itn} for the precise correspondence): 
\begin{equation}
    \mathfrak{confcarr}(d+1)\,\cong\,\mathfrak{iso}(d+1,1)\,.
\end{equation} 
As a result
\begin{center}
     \textit{Poincaré transformations = Strong conformal Carrollian isometries},
\end{center}	
so the Poincaré group can be defined intrinsically on the boundary as the group of strong conformal Carrollian isometries of null infinity endowed with a flat structure.
\end{Example}

\begin{Example}\label{Example: gravity vacuum at null infinity}
More generally, the strong conformal Carrollian structure induced at the null infinity of Kerr spacetimes is flat and its isometries thus define a Poincaré subgroup,
\begin{equation}\label{ISO subset BMS at infinity}
        ISO(d+1,1) \subset BMS_{d+2}.
\end{equation} 
For Minkowski spacetime $\mathbb{R}^{d+1,1}$, the Carrollian isometries of its conformal boundary coincide with the restriction of the global isometries of its interior.
\end{Example}

If the strong conformal Carrollian structure at null infinity is curved, then there will be no strong conformal Carrollian isometries in general.

\begin{Example}\label{Example: gravity vacuum at null infinity2}
For radiative spacetimes of dimension four (i.e. $d=2$), the presence of gravitational radiation implies that the strong conformal Carrollian structure at null infinity is not flat \cite{Ashtekar:1981sf,Ashtekar:2014zsa}.
\end{Example}

\section{Semidirect product structure of BMS group}\label{semiprodstrBMS}

Minkowski spacetime is the affine\footnote{Let us recall that an \textit{affine space} $A$ modelled on a vector space $V$ is a homogeneous space for the regular (i.e. free and transitive) action of $V$ on $A$, where $V$ is seen as an additive group. In other words, for any pair $(a_1,a_2)$ of points $a_1,a_2\in A$ there is a unique vector $v_{12}\in V$ mapping $a_1$ to $a_2$. In colloquial terms, an affine space is a vector space without an origin.} space $\mathbb{M}_{d+2}\simeq\mathbb{R}^{d+1,1}$ equipped with the Minkowski metric with signature $(-,+,\ldots,+)$.
The group of isometries of Minkowski spacetime, i.e. the diffeomorphisms preserving the Minkowski metric, is the Poincar\'{e} group $ISO(d+1,1)$. This group acts transitively on $\mathbb{M}_{d+2}$, thereby making it a homogeneous space. 

\subsection{Canonical subgroups}

We review the canonical normal subgroups of Poincar\'e and BMS groups.

\subsubsection{Canonical subgroups of the Poincar\'e group}\label{canonicsubgPoinc}

The following result is well-known.\footnote{\label{conncomp}Note that, strictly speaking, the group $ISO(d+1,1)$ should read in Proposition \ref{translationmaxnorm} as $ISO_0(d+1,1)$, the connected component of the Poincar\'e group. Similarly, Proposition \ref{Sachs} really applies to the connected component of the BMS group. In order to avoid overloading notations and lengthy discussions, we will implicitly assume, from now on, that $SO(d+1,1)$, $ISO(d+1,1)$ and $BMS_{d+2}$, etc., stand for their connected components.}
\begin{Proposition}\label{translationmaxnorm}
The translation group is the maximal normal subgroup
$\mathbb{R}^{d+1,1}$ of the Poincar\'e group $ISO(d+1,1)$.
\end{Proposition}
\noindent A proof is given in Appendix \ref{Sachsproof} for completeness.

By this proposition, the translation group is a canonical subgroup of the Poincar\'{e} group.
Since it is normal, we can consider the quotient group
$\frac{ISO(d+1,1)}{\mathbb{R}^{d+1,1}}$, which is isomorphic to the Lorentz group $SO(d+1,1)$. Furthermore, we have the canonical projection
\begin{equation}
    \pi\,:\, ISO(d+1,1)  \twoheadrightarrow \frac{ISO(d+1,1)}{\mathbb{R}^{d+1,1}}\,:\, g \longmapsto g\,\mathbb{R}^{d+1,1}
\end{equation}
where $g\,\mathbb{R}^{d+1,1}$ denotes a left coset of the translation subgroup. The kernel of this surjection $\pi$ is precisely the normal subgroup $\mathbb{R}^{d+1,1}$.

We then have the following canonical sequence of Lie groups
\begin{align}
    \mathbb{R}^{d+1,1} \overset{\iota}{\hookrightarrow} ISO(d+1,1) \overset{\pi}{\twoheadrightarrow}  {SO(d+1,1)}
\end{align}
where $\iota$ is the canonical inclusion and
$\pi$ the canonical surjection.
Equivalently, we can say that the Poincar\'{e} group $ISO(d+1,1)$ is a group extension of the Lorentz group $SO(d+1,1)$ by the translation group $\mathbb{R}^{d+1,1} $.

\subsubsection{Canonical subgroups of the BMS group}\label{canonicsubgBMS}

The previous analysis can be applied to the BMS group in an analogous way.

\paragraph{Supertranslation subgroup of the BMS group.}
We have seen in Section \ref{ConformalCarrolliangeometry} that one can define $BMS_{d+2}$ as the group of conformal Carrollian isometries of $\mathscr{I}_{d+1} \simeq \mathbb{R} \times S^d$ equipped with its usual conformal Carrollian structure. 
The translation group $\mathbb{R}^{d+1,1}$ is a normal subgroup of $BMS_{d+2}$, but a larger group, the supertranslation group\footnote{As we already highlighted a couple of times, supertranslations are actually conformal primary fields on the sphere. This can directly be seen from the action of the Lorentz group \eqref{BMSb}.} $C^\infty (S^d)$
is also a canonical normal subgroup, which contains the translations.
In fact, there is a direct analogue of the above Proposition \ref{translationmaxnorm} for the BMS group : 
\begin{Theorem}[Sachs, 1962] \label{Sachs}
The supertranslation group is the maximal normal subgroup of the BMS group.
\end{Theorem}
\noindent Strictly speaking, this is a slight generalisation of Lemma 2 in \cite{Sachs:1962zza}. A proof can be found in Appendix \ref{Sachsproof}.

We can therefore consider the quotient group $\frac{BMS(d+1, 1)}{C^\infty(S^d)}$
and with this we have a canonical projection 
\begin{align}\label{canproj}
    \pi\, :\, BMS(d+1, 1) \twoheadrightarrow \frac{BMS(d+1, 1)}{C^\infty(S^d)}\,:\,    g \mapsto g\,\,C^\infty(S^d)
\end{align}
whose kernel is the supertranslation subgroup $C^\infty(S^d)$.
Following Remark \ref{invariantconformalfactor}, this map $\pi$ corresponds to the projection of conformal Carrollian isometries of null infinity $\mathscr{I}_{d+1}$ onto conformal isometries of the celestial sphere $S^d$. The kernel of this projection is spanned by the vertical automorphisms of the principal bundle $\mathscr{I}_{d+1}\simeq\mathbb{R}\times S^d$, i.e. the supertranslations (cf. Example \ref{supertr}). 
Therefore, the quotient group $BMS(d+1, 1)/C^\infty(S^d)$
identifies with the conformal group of the celestial sphere $S^d$. 

For the sake of simplicity, from now on (if not specified), the dimension of the celestial sphere is implicitly assumed to be not smaller than two ($d\geqslant 2$) such that the conformal group of the celestial sphere $S^d$ is isomorphic\footnote{\label{foot: connectect comp2}Let us stress that, strictly speaking, this isomorphism holds for the connected component of the conformal group of the celestial sphere (see Footnote \ref{conncomp}). The classical theorem is $Conf(S^d)\cong O(d+1,1)/\mathbb{Z}_2$ where $\mathbb{Z}_2=\{+I,-I\}$ is the normal subgroup made of the identity and its opposite. Similarly, let us recall that the connected component $SO_0(d+1,1)$ of the Lorentz group does not act effectively on $S^d$ for $d$ even, so that $Conf_0(S^d)\cong SO_0(d+1,1)/\mathbb{Z}_2$ in this case.} to the Lorentz group and we have\footnote{\label{Foot d=1 groups}For $d=1$, the conformal group of the celestial circle $S^1$ contains all diffeomorphisms. A M\"obius structure on $S^1$ is needed in order to restrict $Diff(S^1)$ to one of its subgroups $SO(2,1)$. Therefore, if one keeps the definition of $BMS_3$ as conformal Carroll group (as in Example \ref{BMStransfosconfCarr}), then the discussion must be modified (see, e.g., \cite{Oblak:2016eij}). The statements in Sections \ref{semiprodstrBMS}-\ref{BMSUIRs} are valid for all $d\geqslant 1$ only if the limiting case $BMS_3$ is understood (in contrast with the usage in \cite{Barnich:2014kra,Barnich:2015uva,Oblak:2016eij}) as denoting the group extension of the Lorentz group $SO(2,1)$ by the supertranslation group $C^\infty(S^1)$.}
\begin{align} \label{BMSd/Cinfty}
	 \frac{BMS(d+1, 1)}{C^\infty(S^d)} \cong SO(d+1,1).
\end{align}

We then have the following canonical sequence 
\begin{align}
    C^\infty(S^d) \overset{\iota}{\hookrightarrow} BMS_{d+2} \overset{\pi}{\twoheadrightarrow}  {SO(d+1,1)}
\end{align}
where $\iota$ is the canonical inclusion (cf. Theorem \ref{Sachs}) and
$\pi$ the canonical surjection \eqref{canproj}.
Equivalently, we can say that $BMS_{d+2}$ is a group extension of the Lorentz group $SO(d+1,1)$ by the supertranslation group $C^\infty(S^d) $.

\paragraph{Translation subgroup of the supertranslation group.}
The following result is well-known\footnote{\label{Foot translationasinvariantsubmodule} See, e.g., part 2 of Theorem 1 in \cite{Redou2003} (together with Remark 1 there) for a proof of the characterisation in Proposition \ref{translationasinvariantsubmodule}.} in conformal geometry and in the representation theory of the Lorentz group.
\begin{Proposition}\label{translationasinvariantsubmodule}
The translation group $\mathbb{R}^{d+1, 1}$ is a canonical subgroup of the supertranslation group $C^\infty(S^d)$ characterised by the fact that it is the unique (closed) invariant subspace of the vector space $C^\infty(S^d)$ of supertranslations, under the action \eqref{BMSb} of the Lorentz group $SO(d+1,1)$.
\end{Proposition}

\begin{Remark}\label{translationambient}
In order to make this more concrete, let us consider the ambient formulation of the flat conformal sphere $S^d$ as a projective null cone in Minkowski spacetime $\mathbb{R}^{d+1, 1}$ (see, e.g., Section 2 of \cite{Eastwood:1991} for a review in any dimension).
The primary fields of scaling dimension $-1$ on $S^d$ (i.e. the supertranslations) can be realised as equivalence classes of functions $F(X)$ on $\mathbb{R}^{d+1, 1}$ of homogeneity degree $1$ in the Cartesian coordinates $X^A$ ($A=0,1,\ldots,d+1$) modulo functions vanishing on the null cone, i.e. $F(X)\sim F(X)+X^2 G(X)$. This provides an ambient realisation of the space $C^\infty(S^d)$ of supertranslations where the action of the Lorentz group $SO(d+1,1)$ is linearly realised on the Cartesian coordinates. The invariant subspace $\mathbb{R}^{d+1, 1}\subset C^\infty(S^d)$ corresponds to the linear functions $F(X)=T_A X^A$ which are the only homogeneous polynomial functions of degree one.
Obviously, these linear functions can be seen as the smooth solutions of the following system of two ambient equations: $(X^A\partial_A -1)F(X)=0$ and $\partial_A\partial_B F(X)=0$. The former equation implements the homogeneity condition while the latter equation picks the translations. In terms of the flat conformal sphere $S^d$, the homogeneity condition implements the scaling dimension $-1$ of supertranslations $\mathcal{T}$ while the second-order equation translates into the following conformally invariant equation: 
\begin{equation}\label{goodcuteqhom}
    \big(\bar\nabla_{(i}\bar\nabla_{j)}-\tfrac1{d}\bar\gamma_{ij}\bar\nabla^2\big)\mathcal{T}(x)=0\,.
\end{equation}
\end{Remark}

\begin{Example}\label{translationscoords}
In the Cartesian coordinates on the Euclidean plane $\mathbb{R}^d$ (i.e. the sphere $S^d$ deprived of a point), the finite translations obtained from the generators in Example \ref{Poincaralgebra} read as $u'=u+a+\vec b\cdot \vec x+c\,|\vec x|^2$ where $a,c\in\mathbb{R}$ and $\vec b\in\mathbb{R}^d$.
As one can see,  $\mathcal{T}(\vec x)=a+\vec b\cdot \vec x+c\,|\vec x|^2$ is indeed a solution of $ \big(\partial_{i}\partial_{j}-\tfrac1{d}\,\delta_{ij}\, \partial_k\partial^k\big)\mathcal{T}(\vec x)=0$, which is the flat version of \eqref{goodcuteqhom}.
\end{Example}

\subsection{Non-canonical subgroups}

As we have seen, the Poincar\'e  and BMS groups are group extensions of the Lorentz group but there is no canonical Lorentz subgroup. The Lorentz group $SO(d+1,1)$ acts on the celestial sphere $S^d$ as conformal transformations and it acts on supertranslations as primary fields on $S^d$ of scaling dimension $-1$. Accordingly, the BMS group $BMS_{d+2}$ has an abstract structure of semidirect product, as does the Poincar\'e group $ISO(d+1,1)$. Although they have an abstract structure of semidirect products, their decomposition as a concrete semidirect product depends on the choice of a Lorentz subgroup. These conceptual subtleties are actually of geometrical and physical significance. They are especially important for the BMS group.

\subsubsection{Non-canonical subgroups of the Poincar\'e  group}
\label{subsubsec non canonical subgroups of P}

Minkowski spacetime $\mathbb{M}_{d+2}\simeq\mathbb{R}^{d+1,1}$ is an affine space which can be viewed as a homogeneous space for the Poincar\'{e} transformations, since they act transitively on it. This implies that, as a homogeneous space, Minkowski spacetime can be seen as a (left) coset space
\begin{equation}
    \mathbb{M}_{d+2} \simeq \frac{ISO(d+1,1)}{\text{Stab}(p)},
\end{equation}
where $\text{Stab}(p)\subset ISO(d+1,1)$ is the stabiliser of a point $p \in \mathbb{M}_{d+2}$ for the action of the Poincar\'{e} group $ISO(d+1,1)$ on Minkowski spacetime $\mathbb{M}_{d+2}$.

For each event $p$, it is easy to see that the stabiliser of this point is the subgroup of Lorentz transformations $SO(d+1,1)$ with $p$ as origin: 
\begin{equation}
    \text{Stab}(p) = SO(d+1,1)_p\,.
\end{equation}
Of course, all such Lorentz subgroups are isomorphic 
\begin{equation}
    SO(d+1,1)_p \cong SO(d+1,1)\,.
\end{equation} 
In fact, all such stabilisers are isomorphic to each other (more precisely they are
 conjugate to each other via translations). Therefore, we can write
\begin{equation}
     \mathbb{M}_{d+2}\simeq \frac{ISO(d+1,1)}{SO(d+1,1)}
\end{equation}
without referring to the point $p$.
The important thing to note is that because Minkowski spacetime $\mathbb{M}_{d+2}$ is homogeneous, there is no preferred way to choose a point $p$ rather than another in order to compute the stabiliser. Therefore, there is no canonical Lorentz subgroup inside the Poincar\'{e} group.
Instead, we have a one-to-one correspondence
\begin{equation}
  \text{Point }  p \in \mathbb{M}_{d+2}\quad \longleftrightarrow\quad \text{Subgroup } SO(d+1,1)_p \subset ISO(d+1,1) .
\end{equation}
If we look at the canonical sequence 
\begin{align}
    \mathbb{R}^{d+1,1} \overset{\iota}{\hookrightarrow} ISO(d+1,1) \overset{\pi}{\twoheadrightarrow}{SO(d+1,1)},
\end{align}
choosing a point $p \in \mathbb{M}_{d+2}$ allows one to write an injection
\begin{equation}
   h_p :  SO(d+1,1)_p \hookrightarrow ISO(d+1,1)\,.
\end{equation}

Recall that a group $G$ is a semidirect product $G=H\ltimes N$ of a subgroup $H\subseteq G$ acting on a normal subgroup $N\trianglelefteq G$ if every element $g\in G$ has a \emph{unique} decomposition as a product $g=hn$ with $h\in H$ and $n\in N$.

A choice of point $p \in\mathbb{M}_{d+2}$ in Minkowski spacetime is equivalent to a decomposition  of the Poincar\'e group as a semidirect product
$ISO(d+1,1) \simeq SO(d+1,1)_p \ltimes \mathbb{R}^{d+1,1}$ of a Lorentz subgroup $SO(d+1,1)_p$ acting on the translation subgroup $\mathbb{R}^{d+1,1} $. This example is just a warm-up for the BMS group.

\subsubsection{Non-canonical subgroups of the BMS group}
\label{subsubsec non canonical subgroups of BMS}

Now we can look at the affine space of all cuts, i.e. $\Gamma(\mathscr{I}_{d+1})$.
We know from Example \ref{supertr} that this affine space 
is modelled on the vector space of all supertranslations $C^\infty(S^d)$. Furthermore, 
the group $BMS_{d+2}$  acts transitively on $\Gamma(\mathscr{I}_{d+1})$ via its action on null infinity.
This means that the affine space $\Gamma(\mathscr{I}_{d+1})$ of all cuts is a homogeneous space which can be 
viewed as a (left) coset space 
\begin{align}
	\Gamma(\mathscr{I}_{d+1}) \simeq \frac{BMS_{d+2}}{\text{Stab}(s)}
\end{align}
where $s:S^d\hookrightarrow\mathscr{I}_{d+1}$ is a cut. Let us spell out the corresponding geometric picture.

\paragraph{Lorentz subgroups of the BMS group.} 
On $\mathscr{I}_{d+1} \simeq \mathbb{R} \times S^d$, for each ``transverse copy'' of the sphere $S^d$, 
there is a corresponding subgroup $SO(d+1,1)$ of conformal isometries of this sphere.
This can be made more precise: for each cut $s : S^d \hookrightarrow  \mathscr{I}_{d+1}$, there is a corresponding subgroup $SO(d+1,1)_s$
of conformal Carrollian isometries preserving the submanifold $s(S^d)\subset \mathscr{I}_{d+1}$. Therefore:
\begin{align}
	\Gamma(\mathscr{I}_{d+1}) 
	\simeq \frac{BMS_{d+2}}{SO(d+1,1)_s}.
\end{align}

There is no preferred section of the principal $\mathbb R$-bundle $\mathscr{I}_{d+1}\simeq \mathbb{R} \times S^d$. This implies that there is no canonical Lorentz subgroup inside the BMS group.
Instead, we have a one-to-one correspondence 
\begin{equation}
    \text{Cut $s$ of }\mathscr{I}_{d+1}  \quad \longleftrightarrow \quad \text{Subgroup }SO(d+1,1)_s \subset BMS_{d+2}.
\end{equation}
A choice of cut $s \in\Gamma(\mathscr{I}_{d+1} )$ is also equivalent to a decomposition  of the BMS group as a semidirect product
$BMS_{d+2} \simeq SO(d+1,1)_s \ltimes C^\infty(S^d)$ of a Lorentz subgroup $SO(d+1,1)_s$ acting on the supertranslation subgroup $C^\infty(S^d)$.

Putting together all the relevant groups of Sections \ref{subsubsec non canonical subgroups of P} and \ref{subsubsec non canonical subgroups of BMS} gives the following diagram
\begin{center}
	\begin{tikzcd}
  C^\infty(S^d) \arrow[rr, hook] 
    & & BMS_{d+2}  \arrow[rd, two heads] &  \\
    & & & \ \ SO(d+1,1) \arrow[lu, blue, hook, "cut", shift left=3] 
    \arrow[ld, red, hook, "point"', shift right=3]\\
  \mathbb{R}^{d+1,1} \arrow[uu, hook]  \arrow[rr, hook] 
& &  ISO(d+1,1) \arrow[ru, two heads] & 
\end{tikzcd}
\end{center}
where the colored arrows are the non-canonical ones, since they depend on a choice (respectively of a cut or of a point).

\paragraph{Poincar\'e subgroups of the BMS group.} 

Consider a Lorentz subgroup $SO(d+1,1)_s\subset BMS_{d+2}$ of the  group. Proposition \ref{translationasinvariantsubmodule} implies that the corresponding Poincar\'e group is also a subgroup of the BMS group:
\begin{align}
   \underbrace{SO(d+1,1)_s \ltimes \mathbb{R}^{d+1,1}}_{\cong ISO(d+1,1)}\;\subset\; \underbrace{SO(d+1,1)_s \ltimes C^\infty(S^d)}_{\cong BMS_{d+2}}\,.
\end{align}
Similarly to the Lorentz group, this Poincar\'e subgroup is not a canonical subgroup of the BMS group.
We will see now that Poincar\'e subgroups are in one-to-one correspondence with the space of gravitational vacua (see, e.g., \cite{Ashtekar:1987tt,Ashtekar:2014zsa}). This is particularly transparent in the following group-theoretical definition.

\begin{Definition}[Ashtekar, 2014]\label{def grav vacua}
A \textit{gravitational vacuum} $\mathbb{M} \in \mathbb{V}$ is an equivalence class $[s]$ formed by all cuts related to a cut $s$ by a translation in $\mathbb{R}^{d+1, 1}\triangleleft BMS_{d+2}$. 
In other words, a gravitational vacuum is an orbit of a cut $s$ under the action of the translation subgroup $\mathbb{R}^{d+1, 1}$.
Therefore, the affine space $\mathbb{V}$ of all gravitational vacua is the homogeneous space
\begin{equation}\label{eq grav vacua}
	\mathbb{V} \simeq \frac{\Gamma(\mathscr{I}_{d+1})}{\mathbb{R}^{d+1,1}}\simeq \frac{BMS_{d+2}}{ISO(d+1,1)}. 
\end{equation}    
Accordingly, a gravitational vacuum is equivalent to a choice of Poincar\'e subgroup inside the BMS group.  
\end{Definition}

\begin{figure}
\centering
\captionsetup{width=0.9\textwidth}
\includegraphics[width=3.2in]{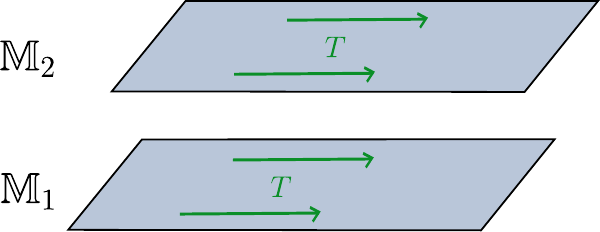}
\caption{The gravitational vacua $\mathbb{M}_i \in \mathbb{V}$ are the different orbits in the space of cuts under the translation group ($T\in\mathbb{R}^{d+1,1}$)}
\label{fig: orbits of transl}
\end{figure}

Choosing a particular Poincar\'e subgroup inside the
BMS group and quotienting by a Lorentz subgroup gives:
\begin{align}
	ISO(d+1,1) &\subset BMS_{d+2} \nonumber\\[0.3em]
     &\Downarrow \\[0.3em]
	\frac{ISO(d+1,1)}{SO(d+1,1)} &\subset 
	\frac{BMS_{d+2}}{SO(d+1,1)} \qquad\longleftrightarrow\qquad \mathbb{M}_{d+2} \subset \Gamma(\mathscr{I}_{d+1}).\nonumber
\end{align}
In fact, any orbit of the translation group inside the space $\Gamma(\mathscr{I}_{d+1})$ of all cuts can be identified with a Minkowski spacetime $\mathbb{M}_{d+2}$, i.e. with a particular gravitational vacuum, see Figure \ref{fig: orbits of transl}.

\begin{Remark}
In the light of Subsection \ref{Strong_conformal_Carrollian_structure}, one can see that a gravitational vacuum (cf. Definition \ref{def grav vacua}) is also equivalent to a choice of flat strong conformal Carrollian structure on null infinity $\mathscr{I}_{d+1}$ (this seminal observation first appeared in \cite{Ashtekar:1981sf,Ashtekar:1987tt}):
\begin{center}
     \textit{Gravitational vacua $\,\longleftrightarrow\,$ Flat strong conformal Carrollian structures}.
\end{center}	
This is consistent with the fact that such (equivalence classes of) compatible connections indeed span an affine space (i.e. $\mathbb{V}$ in Definition \ref{def grav vacua}). 
\end{Remark}

\begin{Example}
The example \ref{Example: gravity vacuum at null infinity} provides us with instances of gravity vacua as preferred Poincaré groups selected by the asymptotic geometry of a Kerr spacetime. More precisely, any stationary asymptotically simple four-dimensional spacetime in fact selects a particular gravitational vacuum \cite{Newman:1966ub}.
\end{Example}

The examples above give concrete realisations of gravity vacua. On the other hand, our general considerations were purely kinematical and only based on group-theoretical considerations, without any mention of dynamical gravity, thus the term \textit{Poincar\'e vacuum} might be more appropriate than ``gravitational vacuum''. Note that this choice of terminology (``Poincar\'e vacuum'') could be motivated as follows: in relativistic QFT, a vacuum is the unique state annihilated by all Poincar\'e generators. In other words, it spans the trivial representation of the Poincar\'e group. Of course, this standard axiom of QFT implicitly assumes that one is given a Poincar\'{e} group.
Accordingly, in a BMS-invariant QFT one expects a one-to-one correspondence between (Poincar\'e-invariant) vacua and Poincar\'{e} subgroups of BMS group.

\begin{figure}
\centering
\captionsetup{width=0.9\textwidth}
\includegraphics[width=3.2in]{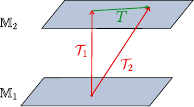}
\caption{$\mathcal{T}_1$ and $\mathcal{T}_2=\mathcal{T}_1+T$ are supertranslations corresponding to the same shift $[\mathcal{T}_1]=[\mathcal{T}_2]\in\Gamma(\mathscr{I}_{d+1})/\mathbb{R}^{d+1,1}$ of gravitational vacua, taking $\mathbb{M}_1$ to $\mathbb{M}_2$,  whereas $T\in\mathbb{R}^{d+1,1}$ is a translation}
\label{fig: shifts of grav vacua}
\end{figure}

The elements of the quotient group $\frac{C^\infty(S^d)}{\mathbb{R}^{d+1,1}}$ of the supertranslation group by its translation subgroup can be thought of as shifts of gravitational vacua, see Figure \ref{fig: shifts of grav vacua}. 
A summary of the relevant affine spaces (and the corresponding vector spaces they are modelled on) is provided in Table \ref{table2}.

\begin{sidewaystable}
		\begin{tabular}{|c|c|c||c|c|}
    \hline
    &&&&\\
			Affine space & Elements & Stabiliser & Vector space & Vectors
            \\&&&&\\\hline\hline
            &&&&\\
Minkowski spacetime & Event & Lorentz subgroup & Translation group & Translation\\
$\mathbb{M}_{d+2}$& $p\in \mathbb{M}_{d+2}$ & $SO(d+1,1)\subset ISO(d+1,1)$ & $\mathbb{R}^{d+1,1}$ & $T\in\mathbb{R}^{d+1,1}$
\\&&&&\\\hline
&&&&\\
Space of cuts & Cut & Lorentz subgroup & Supertranslation group & Supertranslation\\
$\Gamma(\mathscr{I}_{d+1})$ & $s:S^d\hookrightarrow \mathscr{I}_{d+1}$ & $SO(d+1,1)\subset BMS_{d+2}$ & $C^\infty(S^d)$ & $\mathcal{T}\in C^\infty(S^d)$
\\&&&&\\\hline
&&&&\\
Space of gravitational vacua & Gravitational vacuum & Poincar\'e subgroup  & Vacuum shift group & Shift of vacua\\
$\mathbb{V}$& $\mathscr{M}_{d+2}\simeq\mathbb{M}_{d+2}$ & $ISO(d+1,1)\subset BMS_{d+2}$ & $\frac{C^\infty(S^d)}{\mathbb{R}^{d+1,1}}$ & $[\mathcal{T}]$ with $\mathcal{T}\in C^\infty(S^d)$\\
&&&&\\\hline
		\end{tabular}
        \caption{Relevant affine spaces and the vector spaces they are modelled on.}
       \label{table2}
\end{sidewaystable}

\paragraph{Distance between vacua.}

Known facts about the unitary representations of $SO(d+1,1)$ (see, e.g., \cite{Todorov:1978rf,Basile:2016aen,Sun:2021thf} for some reviews) imply that the above vector space of vacua shifts carries a UIR of the Lorentz group. This has an important implication for the geometry of the infinite-dimensional space of gravitational vacua: it is endowed with a positive-definite metric.

\begin{Proposition}\label{UIRLorentz}
    The vector space $\frac{C^\infty(S^d)}{\mathbb{R}^{d+1,1}}$ of vacuum shifts is a Hilbert space carrying a unitary irreducible representation of $SO(d+1,1)$.\footnote{\label{FootHilbert}By which we mean that the quotient $\frac{C^\infty(S^d)}{\mathbb{R}^{d+1,1}}$ carries a Lorentz-invariant positive-definite scalar product (given by \eqref{hermitianproduct} in even $d$). However, strictly speaking, a Hilbert space is only obtained after the suitable completion. This distinction, between $\frac{C^\infty(S^d)}{\mathbb{R}^{d+1,1}}$ and its completion into a Hilbert space, will not play any deep role in these lecture notes.} This implies that the affine space $\mathbb{V}$ of gravitational vacua is endowed with a Lorentz-invariant distance.
\end{Proposition}

\begin{Remark}
The UIR of $SO(d+1,1)$ alluded to in Proposition \ref{UIRLorentz} belongs to the discrete series for $d=1$, to the boundary of the principal series for $d=2$, and to the exceptional series for $d>2$. In the conformal field theory terminology of \cite{Todorov:1978rf,Basile:2016aen,Sun:2021thf}, one might say that (in $d\geqslant 2$) the quotient space $\frac{C^\infty(S^d)}{\mathbb{R}^{d+1,1}}$ is spanned by the primary field $\partial_i\partial_j\mathcal{T}-\tfrac1{d}\delta_{ij}\partial^k\partial_k\mathcal{T}$ and all its descendants, where $\mathcal{T}$ is a supertranslation. The latter is a scalar primary field of scaling dimension $-1$, so the former primary field is labelled by the scaling dimension $+1$, with respect to the dilatation subgroup $SO(1,1)\subset SO(d+1,1)$, and by the rank-two symmetric representation of the rotation subgroup $SO(d)\subset SO(d+1,1)$.
\end{Remark}

For even $d$, the $SO(d+1,1)$-invariant scalar product on the space of supertranslations reads explicitly as
\begin{equation}\label{hermitianproduct}
    \langle\mathcal{T}_1,\mathcal{T}_2\rangle=\int_{S^d}\mathcal{T}_1(x)\,\hat{P}_{d+2}\mathcal{T}_2(x)\,\sqrt{\det\bar\gamma(x)}\,d^dx\,,
\end{equation}
where $\hat{P}_{d+2}$ is the GJMS operator\footnote{Note that our choice of terminology (``GJMS operator'') is a slight departure from the standard one in the conformal geometry literature, because the operator $\hat{P}_{d+2}$ does \textit{not} admit a Weyl-invariant generalisation (this fact is a particular case of \cite{Gover:2003am}). Nevertheless, $\hat{P}_{d+2}$ is an $SO(d+1,1)$-invariant completion of the operator $(\bar\nabla^2)^{\frac{d}2+1}$ on $S^d$.} on $S^d$ \cite{GJMS,Branson}. The latter is a conformally invariant differential operator of order $d+2$ mapping primary fields of scaling dimension $-1$ to primary fields of scaling dimension $d+1$. Explicitly, it reads as follows \cite{Branson}
\begin{equation}\label{GJMSd+2}
    \hat{P}_{d+2}
:=(-1)^{\frac{d}2+1}
\prod_{j=0}^{\frac{d}2}\Big(\bar\nabla^2-(\tfrac{d}{2}+j)(\tfrac{d}{2}-j-1)\Big)\,.
\end{equation}
The scalar product \eqref{hermitianproduct} is positive semi-definite on the vector space $C^\infty(S^d)$ of supertranslations, but it is positive-definite on the quotient space $\frac{C^\infty(S^d)}{\mathbb{R}^{d+1,1}}$ (see Appendix \ref{unitproof} for a proof). For odd $d$, the analogue of the scalar product \eqref{hermitianproduct} is bilocal (see, e.g., \cite{Todorov:1978rf}).

\vspace{1mm}

If $\langle\,[\mathcal{T}_1],[\mathcal{T}_2]\,\rangle$ denotes the Lorentz-invariant scalar product on the space of gravity shifts, then the Lorentz-invariant ``distance'' between two gravity vacua $\mathbb{M}_1$ and $\mathbb{M}_2$ is equal to the norm of the shift $[\mathcal{T}]:\mathbb{M}_1\mapsto \mathbb{M}_2$ relating them 
\begin{equation}
    d(\mathbb{M}_1,\mathbb{M}_2)=\sqrt{\langle\,[\mathcal{T}],[\mathcal{T}]\,\rangle}\,.
\end{equation}
\newpage

\section{Alice in Boundaryland}\label{Alicebdyland}

In this section, we would like to reconstruct Minkowski spacetime holographically, i.e. purely from geometrical data at null infinity. 
Let us imagine that Alice fell in a rabbit hole and is stuck in ``Boundaryland'' (the conformal Carrollian manifold $\mathscr{I}_{d+1}$) aka null infinity. Nevertheless, she would like to reconstruct her own world, ``Interiorland'' (the Lorentzian manifold $\mathbb{M}_{d+2}$) aka Minkowski spacetime, purely from her local experiments. 

This idea was very early on pushed to its ultimate consequences by Newman in the form of his theory of Heaven-space \cite{Newman:1976gc}, relating asymptotically shear-free null-geodesic congruences, solutions of the good cut equations along three-dimensional null infinity, and four-dimensional self-dual radiative (complexified) spacetimes (see, e.g., \cite{Adamo:2009vu} for a review). We will here stick to real, rather than complexified, spacetimes for which good cuts are much more elementary and identify with points of Minkowski spacetime.

\subsection{Gravitational vacua and good cuts at null infinity}\label{extrasubsect}

From Definition \ref{def grav vacua}, one can
see that the affine space of cuts is the total space of a principal bundle
$$\Pi:\Gamma(\mathscr{I}_{d+1})\twoheadrightarrow\mathbb{V}$$
over the affine space $\mathbb{V}$ of gravitational vacua with the translation group $\mathbb{R}^{d+1,1}$ as structure group.
In this picture, over any specific element $v\in \mathbb{V}$, the fibre $\Pi^{-1}(v)$ is an orbit of the translation group $\mathbb{R}^{d+1,1}$. Each such fibre is a gravitational vacuum which can be identified with a copy of Minkowski spacetime $\mathbb{M}_{d+2}$ inside the space  of all cuts. This motivates the following definition.

\begin{Definition}
Given a gravitational vacuum $v\in \mathbb{V}$, the corresponding fibre $\Pi^{-1}(v)\subset\Gamma(\mathscr{I}_{d+1})$ over this element is an orbit of the translation subgroup $\mathbb{R}^{d+1,1}\subset C^\infty(S^d)$. The elements $s$ of the fibre are called the \textit{good cuts} of the corresponding gravitational vacuum.
\end{Definition}

\begin{Remark}
Remember that a generic supertranslation will send one fibre to another one, therefore a family of good cuts is not a BMS-invariant notion.
\end{Remark}

\subsection{Good cuts at null infinity of compactified Minkowski spacetime}

Let us now fix a gravitational vacuum, realised as a given Minkowski spacetime $\mathbb{M}_{d+2}$ with conformal compactification $\hat{\mathbb{M}}_{d+2}$. Following Definition \ref{defcuts}, the cuts are  sections of its null boundary, say future null infinity $\mathscr{I}_{d+1}^+$. The good cuts are represented in Figure \ref{fig: cuts}.
In the space $\Gamma(\mathscr{I}^+_{d+1})$ of all cuts, one can define 
them as some specific sections of $\mathscr{I}_{d+1}^+$: the intersections $\mathscr{C}_d=\mathscr{N}_{d+1}^+\cap \mathscr{I}_{d+1}^+$ between two cones, the future null cone $\mathscr{N}_{d+1}^+$ of a
bulk point $p\in\mathbb{M}_{d+2}$ and future null infinity $\mathscr{I}_{d+1}^+$.
Note that, from the boundary point of view, this is a choice. 
In fact, from the point of view of the BMS group, there is \textit{a priori} no reason to define 
this family of cuts as ``the good ones''.

\begin{figure}
\centering
\includegraphics[width=2.5in]{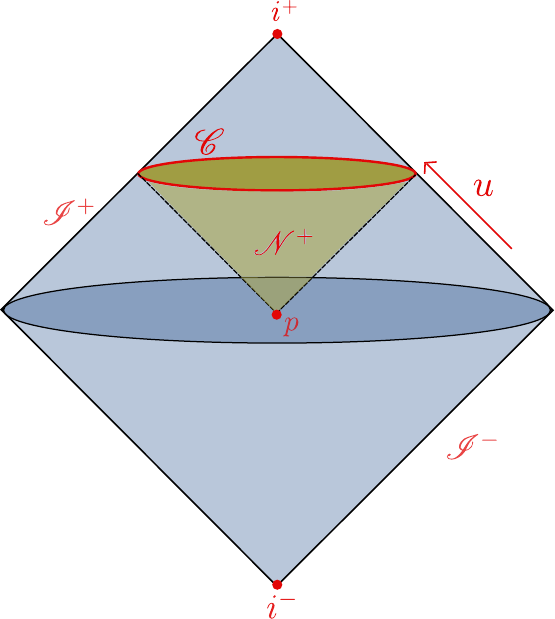}
\caption{A specific good cut $\mathscr{C}$ emanating from $p$}
\label{fig: cuts}
\end{figure}

The importance of the notion of good cut comes from the following property: given the conformally compactified Minkowski spacetime $\hat{\mathbb{M}}_{d+2}$, there is a one-to-one correspondence between these good cuts at future null infinity  and points in Minkowski
spacetime $\mathbb{M}_{d+2}$ (aka events).
This can be easily seen from Figure \ref{fig: cuts}: 
from one point $p \in \mathbb{M}_{d+2}$ in the bulk there is exactly one corresponding intersection $\mathscr{C}_d \simeq S^d$ between its future null cone $\mathscr{N}^+_{d+1}$ and future null infinity $\mathscr{I}_{d+1}^+$.
Conversely, given a good cut $\mathscr{C}_d\subset\mathscr{I}^+_{d+1}$, it defines, by construction, an interior light-cone $\mathscr{N}^+_{d+1}$ emanating from a bulk point $p$. \\

It is important to realise that it is indeed possible to reconstruct Minkowski spacetime purely at $\mathscr{I}_{d+1}$ seen as a conformal Carrollian manifold, if one is given a choice of Poincar\'e subgroup $ISO(d+1,1)\subset BMS_{d+2}$. This holographic reconstruction is the subject of the next subsection.

\subsection{Holographic reconstruction of Minkowski spacetime from good cuts at null infinity}

In the previous section, we have seen that there is a one-to-one correspondence between points in the bulk of Minkowski
spacetime and good cuts at $\mathscr{I}^{+}_{d+1}$. 
Nevertheless, our definition of good cut refers to a point $p$ of Minkowski spacetime.
We will see here that it is possible to characterise the good cuts intrinsically at $\mathscr{I}^{+}_{d+1}$, without
any reference to points in the bulk.
This will give a completely holographic reconstruction of Minkowski spacetime from good cuts at null infinity.

\vspace{1mm}
One needs to proceed in two steps: First, choose an initial simultaneity slice of the conformal Carrollian manifold $\mathscr{I}^{+}_{d+1}$, this defines a gravity vacuum $\mathbb{M}_{d+2}\in \mathbb{V}$. Second, produce all the corresponding good cuts by means of translations. More precisely:

\paragraph{Pick a cut.}
We first need to choose one simultaneity slice at $\mathscr{I}^{+}_{d+1}$. We recall from Example \ref{example simultaneity slices} that it is equivalent to choosing a specific cut $\mathcal{C}_d\subset\mathscr{I}^{+}_{d+1}$. This step is crucial since it effectively amounts to selecting a gravity vacuum. 

\paragraph{Produce the corresponding good cuts by translation.} 
As explained in Section 
\ref{extrasubsect}, we define the subset of good cuts as the orbit under translation of our chosen simultaneity slice. This orbit will then correspond to one
specific gravitational vacuum.
Correspondingly, the set of orbits under translation of all the simultaneity slices will be in one-to-one correspondence with the gravitational vacua:
\begin{equation}
   \begin{alignedat}{2}
    \{\text{Translation orbits of simultaneity slices}\}
    \longleftrightarrow \{\text{Good cuts subspaces}\} \\
    \longleftrightarrow \{\text{Minkowski spacetimes}\} 
    \longleftrightarrow \{\text{Poincar\'{e} subgroups of } BMS\}.
\end{alignedat} 
\end{equation}

In order to  see the good cuts at $\mathscr{I}^{+}_{d+1}$, 
we will now construct the action of the translation group 
on the simultaneity slice $u=u_0$ in the flat Carrollian spacetime $\mathbb{R}\times\mathbb{R}^d$ from Example \ref{flatCarroll} with adapted coordinates $(u,\vec x)$. Following the remark in Example \ref{ScriCarroll}, this flat Carrollian spacetime is seen as representing null infinity $\mathscr{I}_{d+1}\simeq\mathbb{R}\times S^d$ of which a single ray has been removed (corresponding to the point at infinity of the base Euclidean space of the flat Carrollian spacetime).
A general translation $T\in\mathbb{R}^{d+1,1}$ is parametrised by 
an element $(a, \vec{b}, c) $ in $\mathbb{R} \times \mathbb{R}^d \times \mathbb{R}$ following the realisation in Example \ref{translationscoords}.

\begin{itemize}
\item[a)] Under a time translation $u^\prime = u + a$, the simultaneity slice $u=u_0$ can be sent to another simultaneity slice $u=u_1$, at a different Carrollian time, see Figure \ref{fig:time transl}. 
\item[b)] Under a Carroll boost, $u^\prime = u + \vec{b}\cdot \vec{x} $, the original simultaneity slice $u=u_0$ is tilted to the simultaneity slice $u = u_0+ \vec{b}\cdot \vec{x}$, as can be seen in Figure \ref{fig:carroll boost}. 
\item[c)] Under a Carroll temporal special conformal transformation, $u^\prime = u + c\, |\vec{x}|^2$, a simultaneity slice $u = u_0$ is sent to a paraboloid of revolution $u = u_0 + c\, |\vec{x}|^2$ around the time axis, as can be seen in Figure \ref{fig:spatial transl}.
\end{itemize}

\begin{figure}[h!]
  \centering
  \captionsetup{width=0.8\textwidth}
  \begin{subfigure}[b]{0.43\textwidth}
    \centering
\includegraphics[width=2in]{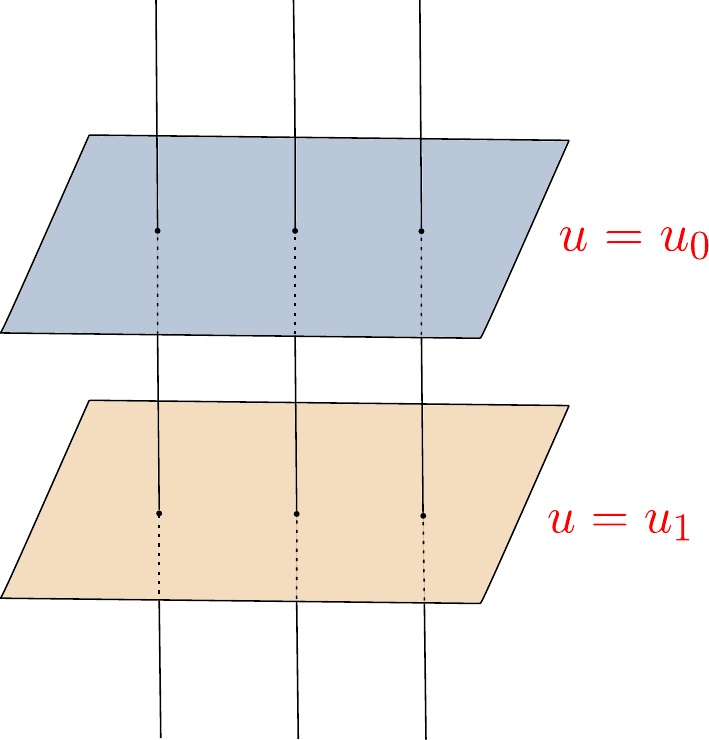}
    \caption{}
    \label{fig:time transl}
  \end{subfigure}
  \hfill
  \begin{subfigure}[b]{0.47\textwidth}
    \centering
    \includegraphics[width=2in]{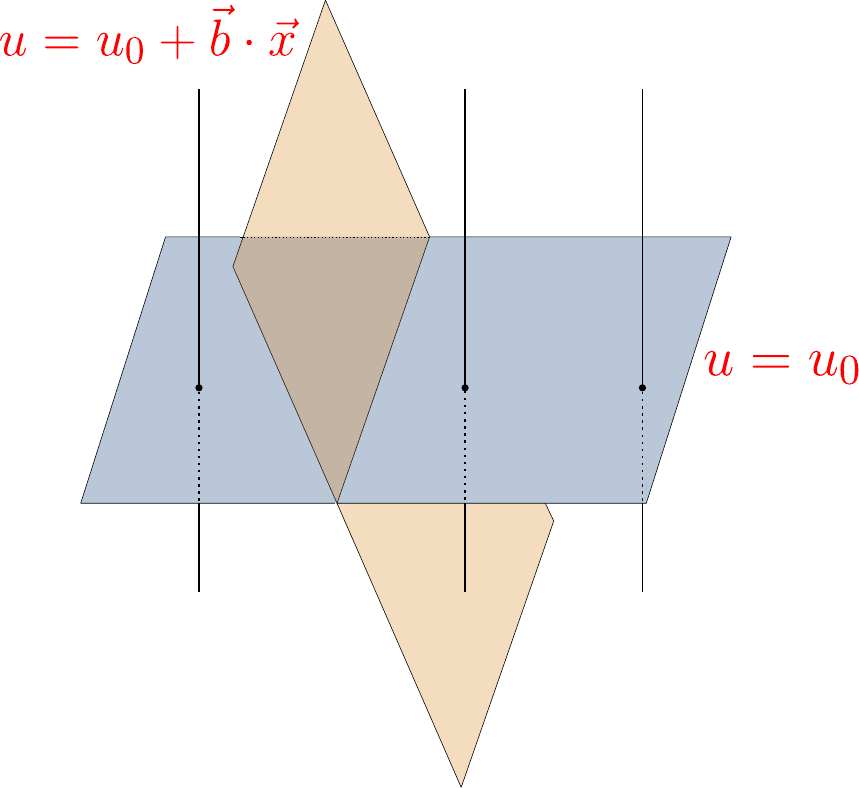}
    \caption{}
    \label{fig:carroll boost}
  \end{subfigure}
  \begin{subfigure}[b]{0.47\textwidth}
    \centering
    \includegraphics[width=2in]{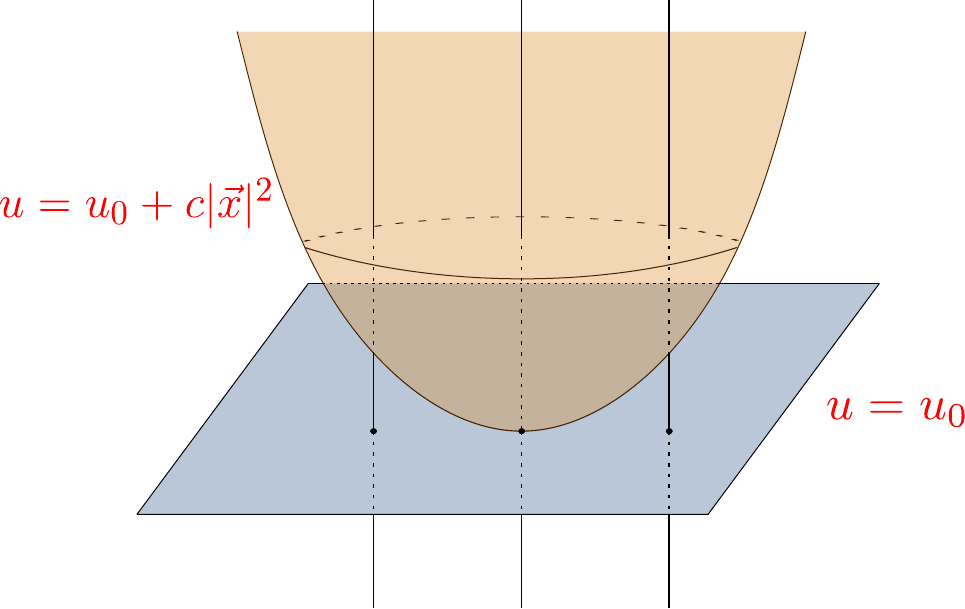}
    \caption{}
    \label{fig:spatial transl}
  \end{subfigure}
  \caption{The image of a simultaneity slice (grey hyperplane) in flat Carroll spacetime under (a) a time translation, (b) a Carroll boost, (c) a Carroll temporal special conformal transformation.}
  \label{fig:figures_cote_a_cote2}
\end{figure}

Therefore, declaring that this orbit defines the set of all good cuts,
one can conclude that the good cuts are in one-to-one correspondence with 
the paraboloids of revolution around their vertical axis 
(including their degeneration into hyperplanes 
transverse to the time axis). In fact, any such 
paraboloid can be described by an equation of the form 
\begin{equation}
u= a + \vec{b}\cdot \vec{x} + c\, \vec{x}^2 
=c\,\big|\,\vec x+\tfrac1{2c}\vec b\,\big|^2+\big(a-\tfrac1{4c}|\vec{b}|^2\big)
\end{equation}
with $(a, \vec{b}, c) \in \mathbb{R}^{d+2}$. Alternatively, these good cuts can be obtained as the kernel of \eqref{goodcuteqhom}.

\newpage

\section{Unitary irreducible representations of the BMS group}\label{BMSUIRs}

In order to investigate the implementation of BMS symmetry in QFT, one may rest in a safe harbor by adopting the group-theoretical approach of Weinberg's textbook for the foundation of relativistic QFT: if one wants to combine the rules of quantum mechanics with special relativity, then one is naturally led to the standard construction of zero-, one- and multi-particle states \cite{Weinberg:1995mt}. 
First, the trivial representation of the Poincar\'e group defines the Poincar\'e vacuum annihilated by all generators: the zero-particle state. Second, unitary irreducible representations of the Poincar\'e group are identified with the possible Hilbert spaces of one-particle states. Third, suitably tensoring the Hilbert space of one-particle states defines the Fock space of all multi-particle states. This motivates the systematic study of the unitary irreducible representations of the BMS group as the natural analogue of usual particles on flat spacetime if one takes into account the asymptotic symmetries of the latter.

The classification of the unitary irreducible representations was performed for $BMS_4$ by McCarthy in the seventies \cite{Mccarthy:1972ry,McCarthy_72-I,McCarthy_73-II,McCarthy_73-III,McCarthy:1974aw,McCarthy_75,McCarthy_76-IV,McCarthy_78,McCarthy_78errata} (see also \cite{Cantoni:1966,Girardello:1974sq,Bekaert:2024uuy,Bekaert:2025kjb}) and, more recently, for $BMS_3$ (including super-rotations) in \cite{Barnich:2014kra,Barnich:2015uva,Oblak:2016eij,Melas:2017jzb}.
We will briefly sketch the method of induced representations by starting with the Poincaré group (Wigner's classification) and then proceed to the BMS group (McCarthy's classification). This method allows for a smooth generalisation in any dimension (see, e.g., the review \cite{Bekaert:2006py} for the Poincar\'e group) so our presentation holds in any dimension equal to or higher than four, but a complete classification is beyond reach for the BMS group in higher dimensions.\footnote{\label{Foot: spinorial reps}Spinorial representations are not discussed here as such, but they could be included if projective representations of Poincar\'e and BMS groups (or, equivalently,  strict representations of their universal covering) are allowed, see e.g. \cite{Bekaert:2025kjb}}.

\subsection{Unitary irreducible representations of the Poincaré group}\label{UIRsPoinc}

The Poincar\'e group  $ISO(d+1,1)$ is the semidirect product of a semisimple Lie group, $SO(d+1,1)$, acting on an Abelian normal subgroup $\mathbb{R}^{d+1,1}$. Furthermore, it is finite-dimensional so Wigner's method of induced representation \cite{Wigner:1939cj} applies in a straightforward way. 

Let us briefly sketch this method (see, e.g., \cite{Bekaert:2006py} for more details): 
		\begin{enumerate}
			\item Consider the UIRs of the translation subalgebra ${\mathbb R}^{d+1,1}$. By Schur's lemma, they are one-dimensional representations 	
			labelled by real eigenvalues $p_A$ of the Hermitian generators $\hat{P}_A$ which identify with the momentum eigenstates.
			\item Identify the orbit and stabiliser of these eigenvalues under the action of the Lorentz group $SO(d+1,1)$. They are respectively the \textit{mass shells} $\mathcal{O}_p$ and the \textit{little groups} $\ell_p$ (see Table \ref{table Wigner reps} for the four possible cases).
			\item Choose an UIR of the little group $\ell_p$. This representation corresponds to the ``spinning'' degrees of freedom or, in other words, the ``physical components'' of the wave function. 
            \item Induce the UIR of the Poincar\'e group $ISO(d+1,1)$ from this UIR of the little group $\ell_p\subset SO(d+1,1)$.
		\end{enumerate}
As Wigner proved in \cite{Wigner:1939cj}, all UIRs of the Poincar\'e group are obtained as induced representations in this way.
        
In first quantisation, the elements of the infinite-dimensional Hilbert space carrying a UIR of the Poincar\'e group are interpreted as the square-integrable wave functions of an elementary particle. These wave functions have support on the orbit $\mathcal{O}_p$ of the momentum $p_A$ under the action of the Lorentz group and take values in the UIR of the little group $\ell_p$. The resulting Hilbert space is the space of one-particle states; the usual Fock space of QFT is then obtained by considering the direct sums of suitable tensor products of this Hilbert space.

    \begin{table}[h]
    \begin{center}
		\begin{tabular}{|c|c|c|c|}
			\hline
			UIR & Orbit (as mass-shell) & Little group & Orbit (as coset) \\\hline\hline
			\small{Massive} & \tiny{2-sheeted hyp} $p^2=-m^2$ & $SO(d+1)$ & $\frac{SO(d+1,1)}{SO(d+1)}\simeq H_{d+1}$\\\hline
			\small{Massless} & \tiny{light-cone} $p^2=0$ & $ISO(d)$ & $\frac{SO(d+1,1)}{ISO(d)}\simeq \mathscr{I}_{d+1}$ \\\hline
			\small{Tachyonic} & \tiny{1-sheeted hyp} $p^2=+m^2$ & $SO(d,1)$ & $\frac{SO(d+1,1)}{SO(d,1)}\simeq dS_{d+1}$ \\\hline
			\small{Zero-momentum} & \tiny{origin} $p_A=0$ & $SO(d+1,1)$ & point \\\hline
		\end{tabular}
        \caption{The four possible cases of orbits and stabilisers for real momenta}
        \label{table Wigner reps}
            \end{center}
	\end{table}
    
\subsection{Unitary irreducible representations of the BMS group}\label{ssection: UIR of BMS}

The group  $BMS_{d+2}$ is the semidirect product of a semisimple Lie group, $SO(d+1,1)$, acting on an Abelian normal subgroup, $C^\infty(S^d)$, so the method of induced representations also applies, although some more care is in order since the Abelian group is infinite-dimensional.

The method of induced representations mimics the one in Section \ref{UIRsPoinc}: 
		\begin{enumerate}
			\item Consider the UIRs of the supertranslation subalgebra $C^\infty(S^d)$. Again, by Schur's lemma they are			
			labelled by real eigenvalues of the generators $\hat{\mathcal{T}}$. These eigenvalues can be seen as linear maps that associate to each supertranslation generator $\hat{\mathcal{T}}$ a real number. In this sense, the eigenvalues belong to the dual of the vector space of supertranslations. Accordingly, they can be thought of as primary fields (actually, distributions) on the celestial sphere $S^d$ with scaling dimension $d+1$. They are called \textit{supermomenta}, by analogy with the Poincar\'e case, and will be denoted $\mathcal{P}(x)$. 
            
The pairing between supertranslations and supermomenta reads explicitly as 
\begin{equation}
                \langle\mathcal{P},\mathcal{T}\rangle=\int_{S^d}d^dx\,\sqrt{\small\bar\gamma(x)}\,  \mathcal{P}(x)\,\mathcal{T}(x)\,.\label{pairingexplicit}
\end{equation}
            
            \begin{Proposition}
            There is a canonical linear projection map 
            \begin{equation}
\pi:\mathcal{P}\mapsto p \in \mathbb{R}^{d+1,1}
            \end{equation}
            that sends any supermomentum $\mathcal{P}$ onto its corresponding momentum $p=\pi(\mathcal{P})$.
            \end{Proposition}
            This proposition is the dual counterpart of Proposition \ref{translationasinvariantsubmodule}. In practice, the associated momentum $p_A$ is obtained by considering the eigenvalues of the generators $\hat{P}_A$ among the supertranslation generators $\hat{\mathcal{T}}$. It reads explicitly as 
            \begin{equation}     p^A=\int_{S^d} d^dx\,\sqrt{\small\bar\gamma(x)}\, \, X^A(x)\, \mathcal{P}(x)\,,\label{pexplicit}
            \end{equation}
            where $X^A(x)$ are the parametric equations of a unit sphere $S^d$ embedded in $\mathbb{R}^{d+1,1}$ as the intersection of the null cone $X^2=0$ and an affine hyperplane
            of the form $N_A X^A=1$ for a unit timelike vector ($N^2=-1$). See Appendix \ref{ambient} for a proof of this formula \eqref{pexplicit}.
			\item Identify the orbit and stabiliser of these supermomenta under the action of the conformal group $SO(d+1,1)$.
			The stabiliser of a supermomentum $\mathcal{P}$ is called a \textit{BMS little group} and denoted  $\ell_{\mathcal P}$. It is always a subgroup $\ell_{\mathcal P}\subseteq\ell_{\pi(\mathcal{P})}$ of the corresponding \textit{Poincar\'e little group} $\ell_{\pi(\mathcal{P})}$.
            \item Choose an UIR of the BMS little group $\ell_{\mathcal P}$.
			\item Induce the UIR of the BMS group $BMS_{d+2}$ from this UIR of the BMS little group $\ell_{\mathcal P}$.
		\end{enumerate}
	The wave function of a BMS particle has support on the orbit $\mathcal{O}_{\mathcal{P}}$ of the supermomentum $\mathcal{P}$ under the action of the Lorentz group and takes values in the UIR of the BMS little group $\ell_{\mathcal{P}}$. The resulting Hilbert space is the space of one-particle states; a Fock space is then obtained by considering tensor products and direct sums.

\begin{Remark}[On the finite-dimensionality of the UIRs of BMS group]
    The orbit of a supermomentum $\mathcal P$ is merely a coset of the (finite-dimensional) Lorentz group
\begin{equation}\label{BMSmasshel}
\mathcal{O}_{\mathcal{P}}\simeq \frac{SO(d+1,1)}{\ell_{\mathcal{P}}}    
\end{equation}
so all orbits (or  ``BMS mass shells'') are finite-dimensional. In this sense, the Hilbert spaces of states for any BMS particle are relatively simple. For instance, they are always separable (since the wave functions live on a finite-dimensional manifold). 
\end{Remark}

 \begin{Remark}[On the topology of the BMS group]\label{topologyremark}
 The fact that the BMS group is infinite-dimensional, namely that supertranslations $\mathcal{T}(x)$ are functions on the sphere $S^d$, implies that usually there is an implicit ambiguity in the definition of the BMS group, related to the fact that one is typically vague about the functional space to which supertranslations belong: if one thinks of the BMS group as a group of diffeomorphisms, as in Example \ref{BMStransfosconfCarr}, then it is more or less implicit (or at least natural) that supertranslations $\mathcal{T}(x)$ should be smooth functions. However, why not generalize the group and consider instead square-integrable functions? Or, in the opposite direction, why not restrict it to analytic functions on $S^d$? In physics, this is usually not much of a concern: in practice, and as far as induced representations of BMS are concerned, this would here amount to deciding whether we want to consider square-integrable or rather distributional supermomenta. However, for the problem  of \emph{classifying} all UIRs of the group, it is surely relevant to know which group one is talking about (analytic, smooth or square-integrable supertranslations, etc). The state of the art of BMS classification theory is a result of Piard \cite{Piard,Piard1977279} which, building on McCarthy's work, proves that \emph{all} UIRs of the BMS group with $L^2$ supertranslations can in fact be obtained as induced representations. There are currently no results proving something similar for  BMS group with smooth supertranslations. Nevertheless, it seems reasonable (since all representations of $L^2$ supertranslations are induced) to expect that counter-examples, i.e.  representations of smooth supertranslations which are not induced (if they exist), are likely to be very exotic and therefore likely to be physically excluded (similarly to Wigner's continuous-spin or tachyonic UIRs of the Poincar\'e group, which are usually discarded); but, still, it is an unsettled question. 
 
 Let us also close this discussion with a word of warning: it might be tempting to think that representations with square-integrable supermomentum are enough (i.e. that restricting to square-integrable supertranslations is enough) and that it is only of mathematical interest to consider distributional supermomenta (i.e. representations of smooth supertranslations). As we shall now see (cf. Example \ref{masslesshardrep}), however, this cannot be done without excluding usual massless particles. Therefore, it would seem that any physically relevant group of supertranslations should at least be contained in smooth functions (i.e. that physically relevant supermomenta should at least contain distributions).
\end{Remark}

\subsubsection{Hard representations of the BMS group}\label{section: hard representation}

     Hard representations are the representations defined by the following straightforward procedure.
\begin{Definition}   
   Consider some field $\Phi(x^A)$ on Minkowski spacetime, realising one of Wigner's Poincar\'e representations. One then identifies its \emph{scattering data}, i.e. the corresponding asymptotic data encoding the value of the field, at timelike infinity  for a massive field, or at null infinity for a massless field. The BMS group, being an asymptotic symmetry, then naturally acts on this scattering data and defines in this way a (very special) BMS representation; the corresponding \emph{hard representation}.
\end{Definition}
 We refer to \cite{Bekaert:2024uuy,Bekaert:2025kjb,Donnay:2026urd} for a systematic presentation of hard BMS representations and how they fit into McCarthy's classification. We only here sketch essential features.

\paragraph{Hard massless BMS representations}\mbox{}

\noindent For massless particles, it was already known to Sachs \cite{Sachs:1962zza} that, in four dimensions, massless fields at null infinity define a representation of $BMS_4$ (see \cite{Bekaert:2022ipg} for the extension to higher dimensions). As far as we know, however, it was only in \cite{Chatterjee:2017zeb} that this representation was explicitly related to McCarthy's classification. The corresponding hard supermomentum is a delta function (see, Table \ref{table hard BMS reps} for a summary and see the example below for more details).

\begin{Example}\label{masslesshardrep}
If $p_A$ is a null momentum, $p^2=0$, then $p^A = \omega \,X^A(y)$ where $\omega$ is the energy and $X^A(y)$ is the normalised null vector defined below Equation \eqref{pexplicit}, with $y$ the point on the celestial sphere $S^d$ corresponding to the direction in which the momentum points. The scattering data of a massless scalar field at null infinity will be of the form 
\begin{equation}\label{Carrollian scalar field}
    \Phi(u,y) = i\int d\omega\left( e^{-i \omega u}a^{\dagger}(\omega,y) - e^{i \omega u}a(\omega,y) \right)
\end{equation}
where $a(p)= a(\omega,y)$ are the usual Fourier modes of a massless real scalar field. The BMS group acts on null infinity via conformal Carrollian isometries (cf. Example \ref{BMStransfosconfCarr}) and therefore naturally acts on \eqref{Carrollian scalar field} via pullback. The resulting action of a supertranslation $\mathcal{T}$ is given by
\begin{align}\label{Hard massless representation}
   a(\omega,y) &\mapsto e^{i \omega \mathcal{T}(y)} a(\omega,y) = e^{i\int d^dx\,\small\sqrt{\bar\gamma(x)}\,  \mathcal{P}(x)\,\mathcal{T}(x)} a(\omega, \vec{y})
\end{align}
where $\mathcal{P}(x) = \omega\, \delta^{(d)}(x,y)$ is the \textit{massless hard supermomentum} corresponding to the momentum $p_A = \omega\, X_A(y)$ and $\delta^{(d)}(x,y)$ is the Dirac distribution on the sphere $S^d$ defined by $\int d^dx\,\small\sqrt{\bar\gamma(x)}\,  \delta^{(d)}(x,y)\,\mathcal{T}(x)=\mathcal{T}(y)$.
\end{Example}

\paragraph{Hard massive BMS representations}\mbox{}

\noindent For massive particles, one can associate to any momentum $p_A$ an associated \textit{massive hard supermomentum} $\mathcal{P}$ which reads as $\mathcal{P}(x)\propto m^{d+2}\big(p_A X^A(x)\big)^{-(d+1)}$, where the precise normalisation factor will not be relevant here. As a result, the action on the usual Fourier modes $a(p)$ will be of the form
\begin{align}\label{Hard massive representation}
   a(p) &\,\mapsto\, e^{i\int d^dx\,\small\sqrt{\bar\gamma(x)}\,  \mathcal{P}(x)\,\mathcal{T}(x)}\; a(p).
\end{align}
The corresponding representation for $d=2$ was first realised in \cite{Longhi:1997zt}. On the other hand, the asymptotic action of BMS on massive fields was worked out much more recently \cite{Campiglia:2015kxa} and the explicit relationship of these works to McCarthy's classification only done in \cite{Borthwick:2024skd,Bekaert:2025kjb}. See Table \ref{table hard BMS reps} for a summary.

    \begin{table}[h]
    \begin{center}
		\begin{tabular}{|c|c|c|c|c|}
			\hline
			UIR & Orbit (as mass-shell) & Hard supermomentum & Little group & Orbit (as coset) \\\hline\hline \rule{0pt}{15pt}

			\small{Massive} & \footnotesize{2-sheeted hyp} $p^2=-m^2$ &  $\mathcal{P}(x) \propto \big(p_A X^A(x)\big)^{-(d+1)}$ & $SO(d+1)$ & $\frac{SO(d+1,1)}{SO(d+1)}\simeq H_{d+1}$\\[0.5em]\hline \rule{0pt}{15pt}
			\small{Massless} & \footnotesize{light-cone} $p^A = \omega \,X^A(y)$ &  $\mathcal{P}(x) = \omega\, \delta^{(d)}(x,y)$   & $ISO(d)$ & $\frac{SO(d+1,1)}{ISO(d)}\simeq \mathscr{I}_{d+1}$ \\[0.5em] \hline
		\end{tabular}
        \caption{The hard BMS supermomentum, associated little groups and orbits.}
        \label{table hard BMS reps}
            \end{center}
	\end{table}

\vspace{4mm}

As should be clear from comparing Table \ref{table Wigner reps} and Table \ref{table hard BMS reps}, the hard supermomenta are such that their BMS and Poincar\'e little groups coincide: $\ell_{\mathcal{P}}=\ell_{\pi(\mathcal{P})}$. Therefore, the orbits of hard supermomenta coincide with usual mass shells: $\mathcal{O}_{\mathcal{P}}\simeq \mathcal{O}_{\pi(\mathcal{P})}$. This implies that there is a one-to-one correspondence between
\begin{center}
hard UIRs of BMS $\longleftrightarrow$ non-zero momentum UIRs of Poincar\'e.
\end{center}

In particular, the wavefunctions of Poincaré and hard BMS representations \emph{are identical}. Let us emphasise the implication of this result: for a generic BMS representation, the choice of a gravitational vacuum $ISO(d+1,1) \subset BMS_{d+2}$ will in general decompose the BMS wavefunction into infinitely many usual Poincaré particles. In practice, this decomposition will depend on the chosen gravitational vacuum. However, this does \emph{not} happen for hard BMS particles: the same wavefunction can describe both a Poincaré particle and a BMS particle. What is more, this applies irrespectively of the Poincaré group chosen. In this sense, hard representations are insensitive to the choice of gravity vacua: \emph{for hard particles all Poincaré groups are equal}. 

\begin{Remark}[More on the topology of the BMS group] \label{topologyremark2} Let us expand further on the topology of the BMS group by continuing the discussion initiated in Remark \ref{topologyremark}.
    In McCarthy's original work \cite{Mccarthy:1972ry,McCarthy_72-I,McCarthy_73-II,McCarthy_73-III} the square-integrable supertranslations were (incorrectly) praised as having the property of recovering Wigner's usual  massive and massless  representations (with discrete-spin) while excluding the tachyonic and continuous-spin representations. The latter are usually considered to be unphysical ``predictions'' of Wigner's classification and McCarthy claimed that the BMS group with $L^2$ supertranslations was doing the fantastic job of cherry-picking just the physical representations. It was soon understood \cite{Girardello:1974sq} that continuous-spin representations could be obtained via representations of the BMS group with smooth supertranslations and McCarthy then generalised his work to this case \cite{McCarthy_75,McCarthy_76-IV,McCarthy_78,McCarthy_78errata}. This is the classical way the story is told, however this way of phrasing the problem overlooks an essential aspect: the real trouble is that BMS representations with $L^2$ supertranslations simply \emph{cannot capture} the usual massless particles. To see this, it is enough to consider Example \ref{masslesshardrep}: the supermomentum of the hard massless representation is a distribution and therefore does not belong to the dual space of $L^2$ functions on the sphere (rather, it is a dual element of smooth supertranslations). Another way to see the problem is to consider McCarthy's classification of admissible little groups: in McCarthy's initial classification \cite{Mccarthy:1972ry,McCarthy_72-I,McCarthy_73-II,McCarthy_73-III}, there was no supermomentum with stabiliser $ISO(2)$ and, as a result, there was no orbit of the form $\frac{SO(3,1)}{ISO(2)}\simeq \mathscr{I}_{3}$ and, therefore, no possibility of having wavefunctions that could be identified with Wigner's massless particles. Thus, the problem with $L^2$ supertranslations has nothing specific to do with continuous-spin representations (as could be wrongly inferred from the title of \cite{Girardello:1974sq}), it is merely that they do not capture usual massless particles at all! How could McCarthy have missed this? The answer is the following: in his classification, he realised that there was, as opposed to Wigner's classification, a supermomentum of stabilizer $U(1)$. He (wrongly) identified the representations of this $U(1)$ group with the helicity of usual massless particles. However, this is clearly incorrect since the corresponding wavefunctions would then be wavefunctions on $\frac{SO(3,1)}{U(1)}$ which is a five-dimensional orbit and therefore very far from the three-dimensional mass-shell of Wigner's massless particle.\footnote{Considering the formidable amount of work done in the seminal series of articles by McCarthy, it is easy to forgive this misconception. That this is a very tempting mistake is all the more demonstrated by the fact that it keeps coming back, see e.g. \cite{Ruzziconi:2026isv} where a supermomentum with stabiliser $U(1)$ is inaccurately interpreted as related to BMS hard massless representations.}
\end{Remark}

\subsubsection{Soft representations of the BMS group}\label{section: soft representation}

Hard supermomenta $\mathcal{P}(x)$ are in one-to-one correspondence with usual momenta $p^A$, see Table \ref{table Wigner reps} and \ref{table hard BMS reps}. In other words, they provide some (non-linear) inverse to the projection \eqref{pexplicit}.

On the opposite extreme, one can define \textit{soft supermomenta} as those supermomenta $\mathcal{P}$ which are such that their associated momentum is zero, $\pi(\mathcal{P})=0$.
Thus their Poincar\'e little group is the whole Lorentz group $SO(d+1,1)$. Therefore, a priori the dimensions of the orbits of soft supermomenta range from 0 to $\frac{(d+2)(d+1)}{2}$. 

\begin{Proposition}\label{Soft supermomentum characterisation}
    For even $d$, all soft supermomenta $\mathcal P$ belong to the image of the GJMS operator \eqref{GJMSd+2}. More precisely,
\begin{equation}\label{equivsoft}
\pi(\mathcal{P})=0
\qquad\Longleftrightarrow\qquad\exists !\;[\mathcal{N}]\in\frac{C^\infty(S^d)}{\mathbb{R}^{d+1,1}}\quad\text{such that}\quad \mathcal{P}=\hat{P}_{d+2}\,\mathcal{N}\,.
\end{equation}
\end{Proposition} 

\begin{Remark}\label{Remark on distributional soft supermomenta}
    This result holds both in the smooth category (i.e. for smooth supermomenta and vacua shifts) and in its continuous dual (i.e. for supermomenta and vacua shifts which are distributions). The corresponding proofs can be found in Appendix \ref{ImGJMS} (for more detailed discussions, see \cite{Bekaert:2025kjb,Bekaert2026}).
\end{Remark}

The zero-momentum representations of the Poincar\'e group $ISO(d+1,1)$ are unfaithful UIRs of the Poincar\'e group (since the translations act trivially) but they are faithful UIRs of the Lorentz group $SO(d+1,1)\simeq\frac{ISO(d+1,1)}{\mathbb{R}^{d+1,1}}$. In other words, the zero-momentum representations of the Poincar\'e group identify with the UIRs of the Lorentz group. There is a similar but much richer situation for the BMS group. 

The following quotient group was called the ``Komar group'' by McCarthy in \cite{McCarthy_72-I}, as a credit to \cite{PhysRevLett.15.76}.
\begin{Definition}[McCarthy, 1972]\label{Def: Komar group}
The \textit{Komar group} is the quotient group of the BMS group by the normal subgroup of translations
		$$\frac{BMS_{d+2}}{\mathbb{R}^{d+1,1}}\simeq SO(d+1,1)_s\ltimes\frac{C^\infty(S^d)}{\mathbb{R}^{d+1,1}}$$
which, given a cut $s\in\Gamma(\mathscr{I}_{d+1})$, decomposes as the semidirect product of the Lorentz group acting on the normal subgroup of shifts of gravitational vacua.
\end{Definition}			
		
\begin{Proposition}\label{Prop: on reps of Komar group}
    The soft representations of the BMS group $BMS_{d+2}$ of non-zero supermomentum are unfaithful UIRs of the BMS group (since the translations act trivially) and are UIRs of the Komar group $\frac{BMS_{d+2}}{\mathbb{R}^{d+1,1}}$.
\end{Proposition}

Furthermore, induced soft representations are in fact \textit{faithful} UIRs of the Komar group. A proof is provided in Appendix \ref{RepKomar}.

\subsubsection{Generic representations of the BMS group}\label{ssection: generic BMS reps}

To conclude, let us consider supermomenta which are generic (i.e. neither zero, nor hard nor soft). Their Poincar\'e little group $\ell_{\pi(\mathcal{P})}$ is a subgroup of dimension $\frac{d(d+1)}2$, see the three non-trivial cases in Table \ref{table Wigner reps}. Therefore, the dimension of the BMS little group $\ell_{\mathcal P}\subseteq\ell_{\pi(\mathcal{P})}$ is bounded from above by this value. Consequently, the dimension of the generic BMS mass shells \eqref{BMSmasshel} is between  $d+1$ and $\frac{(d+2)(d+1)}{2}$. By analogy with what happens for $d=2$ \cite{McCarthy_75,McCarthy_78errata,Girardello:1974sq}, one might speculate that the dimensions of the orbits of soft supermomenta also range from $d+1$ to $\frac{(d+2)(d+1)}{2}$.

Finally, despite the fact that a generic supermomentum is neither hard nor soft, it can always be, uniquely and in a Lorentz-invariant way,  decomposed into a hard and a soft contribution \cite{Bekaert:2024uuy,Bekaert:2025kjb}:
\begin{equation}\label{hard/soft decomposition}
    \mathcal{P}(x) = \textrm{P}_p(x) + \hat{P}_{d+2}\,\mathcal{N}(x)
\end{equation}
where $\textrm{P}_p(x)$ is the unique hard supermomentum (see Table \ref{table hard BMS reps}) corresponding to $p= \pi\left(\mathcal{P}\right) =\pi\left(\textrm{P}_p\right)$ and $\hat{P}_{d+2}\,\mathcal{N}$ is soft (see previous subsection). Remember that $\hat{P}_{d+2}$ is the GJMS operator \eqref{GJMSd+2}.
\begin{proof}
Let $\mathcal{P}(x)$ be a generic supermomentum and let $p=\pi(\mathcal{P})$ be its associated momentum, as given by \eqref{pexplicit}. Then, as a result of Section \ref{section: hard representation}, there exists a unique hard supermomentum $\textrm{P}_p(x)$ such that $p = \pi(\textrm{P}_p)$. Since the projection \eqref{pexplicit} is linear, this implies that $\pi\left(\mathcal{P}- \textrm{P}_p\right)=0$ is soft. We conclude, as a result of Proposition \ref{Soft supermomentum characterisation}, that there exists a unique $[\mathcal{N}]\in\frac{
C^\infty(S^d)}{\mathbb{R}^{d+1,1}}$ such that $\mathcal{P}- \textrm{P}_p = \hat{P}_{d+2}\,\mathcal{N}$.
\end{proof}

Note that the hard supermomentum $\textrm{P}_p(x)$, see Table \ref{table hard BMS reps}, depends non-linearly on the associated momentum $p^A$. Therefore the supermomentum $\textrm{P}_{p_1+p_2}(x)$ associated to $p_1^A+p_2^A$ does not correspond to a mere sum $\textrm{P}_{p_1}(x) + \textrm{P}_{p_2}(x)$ of the corresponding hard supermomenta, but contains a soft term:
\begin{equation}
    p_3^A = p_1^A + p_2^A\quad\Longrightarrow\quad\textrm{P}_3(x)=\textrm{P}_1(x)+\textrm{P}_2(x)+\hat{P}_{d+2}\mathcal{S}(\vec{x})\,.
\end{equation}
As a result, the decomposition \eqref{hard/soft decomposition} has the interesting property of being also non-linear:
\begin{equation}\label{Non linearity of the decomposition}
    \mathcal{P}_1(x) + \mathcal{P}_2 (x) = \textrm{P}_3(x) + \hat{P}_{d+2}\Big( \mathcal{N}_1(x) + \mathcal{N}_2(x) + \mathcal{S}(x)\Big)\,.
\end{equation}

As it turns out, the term $\mathcal{S}(\vec{x})$ appearing in the last expression is in fact directly related to \emph{Weinberg's soft factor} \cite{Weinberg:1995mt} (see \cite{Bekaert:2025kjb,Donnay:2026urd} for detailed proofs). This is not a coincidence but a manifestation of the fact that the infrared physics of gravity is closely tied up with its asymptotic symmetry group, see \cite{Strominger:2013jfa,He:2014laa} and the review \cite{Strominger:2017zoo}. 
The present lecture notes are not the place to dwell in detail on the relationship between representation theory of the BMS group and infrared physics; we rather refer to \cite{Donnay:2026urd} for a detailed discussion of Weinberg's soft theorems and infrared divergences in view of BMS representation theory. The interested reader will also find in \cite{Donnay:2026urd} a careful discussion on how infrared-finite dressed-states can be related to the generic (neither hard nor soft) BMS representation. See also the conclusion for some more comments.

\newpage
\section{Conclusion}\label{conclu}

The goal of these lecture notes was to introduce basic elements of BMS group theory to a geometrically oriented audience, which might be willing to decipher the interplay between the algebraic features of the BMS group, notably its different reductions to different Poincaré groups (the ``gravity vacua''), and its geometrical realisations in terms of Carrollian geometries. In particular, we hope that the present lectures can serve as a bridge between the classical reviews of Ashtekar \cite{Ashtekar:1987tt,Ashtekar:2014zsa,Ashtekar:2018lor} and the early works of McCarthy on representation theory of the BMS group \cite{Mccarthy:1972ry,McCarthy_72-I,McCarthy_73-II,McCarthy_73-III,McCarthy:1974aw,McCarthy_75,McCarthy_76-IV,McCarthy_78}; each of them capturing essential aspects of the BMS group.

\paragraph{Further Reading}

The lectures have been written in such a way as to emphasise the unifying features common to all the groups $BMS_{d+2}$ for $d\geqslant 2$. The price to pay, which is a limit of our presentation, is that it plays down the essential difference with the lower-dimensional BMS group $d=1$; the interested reader is referred to \cite{Barnich:2014kra,Barnich:2015uva,Oblak:2016eij} for detailed discussions on the representation theory of $BMS_{3}$. Another shortcoming of our presentation is that the format and scope of the present article cannot do full justice to the functional analysis details involved in BMS representation theory. To mitigate this choice, we nevertheless   tried to warn the reader on those matters, by means of suitable remarks, whenever we felt that the functional space at stake had some particular relevance,  especially when some physics might be involved. For instance, distributional supermomenta are certainly needed to incorporate the massless hard representation, cf. Remark \ref{topologyremark}. We refer to \cite{McCarthy_75,McCarthy_76-IV,McCarthy_78,Bekaert:2025kjb,Donnay:2026urd} for further reading and discussion on the UIRs of $BMS_4$.

\vspace{4mm}
As should have become clear to the reader, the present lecture notes do not, by any means, aim to provide a comprehensive review of BMS representation theory. Rather, we hope that it can serve as a bridge between vastly different parts of the literature. In view of the repeated claim \cite{Strominger:2013jfa,He:2014laa,Strominger:2017zoo,Choi:2017ylo} that the BMS group controls the infrared behavior of gravitational scattering, we hope that the present introduction can be of some help to the community. Incidentally, let us briefly comment on some consequences of BMS representation theory for the investigation of infrared divergences in gravitational scattering. 

\paragraph{Generic BMS representations and infrared divergences.}

It turns out to be a consequence of the non-linearity \eqref{Non linearity of the decomposition} of the decomposition \eqref{hard/soft decomposition} that the usual states of QFT, corresponding to the hard BMS representations, cannot by themselves preserve supermomentum.\footnote{For instance, conservation of momentum $p_1 + p_2 = p_3+ p_4$ does \emph{not} imply conservation of supermomentum, i.e. $P_{p_1} + P_{p_2} \neq P_{p_3}  + P_{p_4}$.} Therefore, if any BMS-invariant S-matrix can ever be defined\footnote{Where here ``well-defined'' stands for infrared-finite.}, it will have to accommodate generic (non hard) BMS representations. Interestingly, the corresponding generic BMS-states turn out to be extremely rich: First of all, as was shown in \cite{Bekaert:2024uuy}, they can always be interpreted as quantum superpositions of usual particles, each of them propagating on a different gravity vacua. In this sense, they naturally produce states which mix ``hard'' and ``soft'' degrees of freedom.\footnote{Surprisingly, this can be achieved while retaining wavefunctions defined on a finite-dimensional space. The constructions and detailed discussions in \cite{Bekaert:2024uuy,Donnay:2026urd} provide concrete counter-examples to the claims (made e.g. in \cite[Section 10.1.8]{Oblak:2016eij} and \cite{Bosma:2026qcb}) that all BMS representations are too small to leave room for soft gravitons, or that it is necessary to go beyond the BMS group to describe dressed states.} What is more, the states made available by using generic BMS representations have enough room to, at least formally, solve a classical problem: as emphasised very early on by Ashtekar \cite{Ashtekar:1978zz}, the usual (hard) QFT states exclude some observables (such as the memory effect) which would classically be a mandatory part of any scattering process; this failure of the quantum theory to accommodate a perfectly reasonable infrared observable can then be understood as a source of infrared divergences, see e.g. \cite{Ashtekar:2018lor,Prabhu:2022zcr,Prabhu:2024lmg,Prabhu:2024zwl} for a discussion along those lines. As it turns out, however, generic BMS particles do not have this problem \cite{Bekaert:2024uuy}. One is therefore allowed to speculate that states belonging to generic BMS representations might be a key to defining infrared-finite S-matrix elements, incorporating BMS-symmetries as one of their defining features. Here again, the interested reader is referred to \cite{Donnay:2026urd} for a refined discussion.

\section*{Acknowledgements}

X.B. is grateful to the group ``Physique de l’Univers, Champs et Gravitation'' of Mons University for warm hospitality and kind invitation to deliver some lectures on BMS group theory in June 2024, which are the basis of these notes. We thank Andrea Campoleoni for discussions on the intricate relations between radiation, supertranslation and memory effect in higher dimensions. We are also grateful to Thomas Basile for his careful reading of a preliminary version of these notes. Finally, we thank the three anonymous referees for numerous valuable improvements.

L.M. is a FRIA grantee of the Fund for Scientific Research — FNRS, Belgium. The work of L.M. was supported by FNRS under Grant FC.55077. 
N.P is a FNRS grantee of the Fund for Scientific Research — FNRS, Belgium.

\appendix

\section*{Appendix}
\addcontentsline{toc}{section}{Appendix}

\setcounter{section}{1}

\subsection{Proof of the invariance of the conformal factor}\label{invproof}

\paragraph{Conformal Carrollian isometries.}
As mentioned in Remark \ref{invariantconformalfactor}, for an invariant Carrollian metric, $\mathcal{L}_n\gamma=0$, the conformal factor is invariant, $\mathcal{L}_n \Omega = 0$.
This follows from the two conditions in Definition \ref{confCarriso}. Indeed, the first condition (i.e. $n^\prime = \Omega^{-1} n$) asserts the preservation of the direction of the fundamental vector field $n$, and thus the preservation of the fibration. Therefore, a conformal Carrollian isometry is necessarily an automorphism of the fibre bundle. In turn, this implies that it projects onto a diffeomorphism of the base. For an invariant metric, this implies that the second condition (i.e. $\gamma^\prime = \Omega^2 \gamma$) admits a well-defined projection $\bar\gamma^\prime = \bar\Omega^2 \bar\gamma$. Finally, this implies that $\Omega= \pi^*\bar\Omega$ must be invariant.

\paragraph{Conformal Carroll-Killing vector fields.} The remark  \ref{infversionremark} is the infinitesimal counterpart of Remark \ref{invariantconformalfactor}.
We have to check that, for an invariant Carrollian metric, the function $f$ appearing in Definition \ref{confCarrollKilling} obeys $\mathcal{L}_n f = 0$.
This can be shown directly as follows. The first condition in Definition \ref{confCarrollKilling} implies $\{n,X\}=f\,n$ (where the curly bracket stands for the Lie bracket of vector fields), while the second condition implies $[{\cal L}_n,{\cal L}_X]\gamma=(2{\cal L}_{n}f)\,\gamma$ since $\mathcal{L}_n \gamma = 0$. One can also check that ${\cal L}_{fn}\gamma=0$ for any function $f$, since $\gamma$ is invariant and its contraction with $n$ vanishes.
Combining these various results leads to $0={\cal L}_{fn}\gamma={\cal L}_{\{n,X\}}\gamma=[{\cal L}_n,{\cal L}_X]\gamma=(2{\cal L}_{n}f)\,\gamma$. This shows that $\mathcal{L}_n f = 0$. 

\subsection{Proof of Sachs theorem}\label{Sachsproof}

We will prove in one shot both Proposition \ref{translationmaxnorm} 
and Theorem \ref{Sachs} 
by showing the following lemma, which can be applied for $G$ the Poincar\'e (or BMS) group, $N$ the (super)translation and $Q$ the Lorentz group (strictly speaking, their connected components).

\vspace{1mm}

\noindent\textbf{Lemma.} Let us consider a group $G$ with a proper normal subgroup $N\triangleleft\; G$ and corresponding quotient group $Q=\frac{G}{N}$ (i.e. a group extension $G$ of $Q$ by $N$).\\
If the quotient group $Q$ is simple, then the proper normal subgroup $N$ is maximal.

\vspace{1mm}

\begin{proof}
Let us proceed by contradiction. Assume that $N$ is not a maximal normal proper subgroup of $G$, that is to say there exists a proper normal subgroup $N'\triangleleft G$ of the group $G$ that contains the normal group $N\triangleleft G$ as a subgroup, i.e. $N\subset N'$. In group theory, the ``correspondence theorem'' (see, e.g., Theorem 7.14 in \cite{Humphreys1996}) asserts in particular that 
$$
N\subset N'\triangleleft G\quad\Longrightarrow\quad \frac{N'}{N}\triangleleft \frac{G}{N}\,.
$$
Hence, $\frac{N'}{N}$ is a proper normal subgroup of the quotient group $Q$. However, the latter group $Q$ is assumed to be simple in the lemma. Therefore, the quotient group $\frac{N'}{N}$ must be trivial, i.e. $\frac{N'}{N}\simeq\{e\}$. In turn, this implies that $N\simeq N'$, which is a contradiction.
\end{proof}

\subsection{Spectrum of the GJMS operator}\label{unitproof}

The issue of the positive-definiteness of the scalar product \eqref{hermitianproduct} reduces to the 
positive-definiteness of the GJMS operator \eqref{GJMSd+2}.
This property can be seen concretely by investigating its spectrum. Let us recall that the set of eigenvalues of the Laplacian on $S^d$ is 
\begin{equation}
    \underset{L^2(S^d)}{\text{Spec}}\big(\bar\nabla^2\big)=\{-\ell(\ell+d-1)\mid\ell\in\mathbb{N}\}\,.
\end{equation}
Therefore, the spectrum the GJMS operator \eqref{GJMSd+2} is
\begin{eqnarray}
    \underset{L^2(S^d)}{\text{Spec}}\big(\hat{P}_{d+2})&=&\Big\{\prod_{j=0}^{\frac{d}2}\big[\ell(\ell+d-1)+(\tfrac{d}{2}+j)(\tfrac{d}{2}-j-1)\big]\mid\ell\in\mathbb{N}\Big\}\nonumber\\
    &=&\Big\{\prod_{k=-1}^{\frac{d}2-1}\big[(\ell+k)(\ell-k+d-1)\big]\mid\ell\in\mathbb{N}\Big\}\\
    &=&\Big\{\ell(\ell-1)\prod_{p=1}^{\frac{d}2-1}(\ell+p)\prod_{q=0}^{\frac{d}2}(\ell-q+d)\mid\ell\in\mathbb{N}\Big\}\,.\nonumber
\end{eqnarray}
All factors in the last line are manifestly positive, except for the first two factors which vanish for the first two spherical harmonics $\ell=0,1$. The latter harmonics span the Lorentz-invariant subspace $\mathbb{R}^{d+1,1}\subset C^\infty(S^d)$ of translations. Therefore, the operator $\hat{P}_{d+2}$ is positive-definite on the quotient space $\frac{C^\infty(S^d)}{\mathbb{R}^{d+1,1}}$. 

\subsection{Ambient realisation of supermomenta}\label{ambient}

Recall from Remark \ref{translationambient}, that supertranslations can be thought as the restriction of functions $\mathcal{T}(X)$ on Minkowski spacetime $\mathbb{R}^{d+1, 1}$ of homogeneity degree $1$ in the Cartesian coordinates $X^A$, pullbacked on the light-cone $\eta_{AB}X^AX^B=0$. Similarly,
an equivalent realisation of supermomenta is to think about them as restrictions of homogeneous functions $\mathcal{P}(X)$ of degree $-(d+1)$ on Minkowski spacetime ${\mathbb R}^{d+1,1}$, also pullbacked on the light-cone $X^2=0$. In this formulation, the pairing \eqref{pairingexplicit} between supertranslations and supermomenta reads
\begin{equation}
                \langle\mathcal{P},\mathcal{T}\rangle
                =\int_{\mathbb{R}^{d+1,1}} d^{d+2}X\,\,\delta(X^2)\,\delta(N\cdot X-1)\,  \mathcal{P}(X)\,\mathcal{T}(X)\label{pairinganyd}
\end{equation}
where, in the second line, $N\in\mathbb{R}^{d+1,1}$ is a fixed unit timelike vector ($N^2=-1$). It can be shown that the integral does not actually depend on the choice of $N$; this result is sometimes called the ``covariance lemma'' (see page 83 of \cite{Todorov:1978rf}) because it ensures the $SO(d+1,1)$-invariance of the pairing \eqref{pairinganyd}. 

From Remark \ref{translationambient}, one knows that translations can be realised as linear functions $\mathcal{T}(X)=T_AX^A$ on $\mathbb{R}^{d+1,1}$. Inserting them into Equation \eqref{pairinganyd} shows that the associated momentum can be computed explicitly via the integral
\begin{equation}     
p^A=\int_{\mathbb{R}^{d+1,1}}\! d^{d+2}X\,\,\delta(X^2)\,\delta(N\cdot X-1)\, X^A\, \mathcal{P}(X)\,,
\end{equation}
which reproduces Equation \eqref{pexplicit}.
            
\subsection{Image of the GJMS operator}\label{ImGJMS}

Here, we want to prove Proposition \ref{Soft supermomentum characterisation} in the smooth case. In other terms, we want to show that if $\mathcal{P}$
is a smooth supermomentum satisfying $\pi \left(\mathcal{P}\right) =0$, i.e. if
\begin{equation}\label{Pdecomp}
    \mathcal{P} = \sum_{\ell=2}^\infty \mathcal{P}_{\ell},
\end{equation}
where $\mathcal{P}_{\ell}$ is a linear combination of the $\ell$-th spherical harmonics, then there exists a smooth conformal primary field $\mathcal{N}\in C^\infty(S^d)$ of scaling dimension $-1$ such that $\mathcal{P} = \hat{P}_{d+2} \mathcal{N}$. Making use of Section \ref{unitproof}, we are therefore looking for
\begin{equation}\label{Ndecomp}
    \mathcal{N} = \sum_{\ell=0}^\infty \mathcal{N}_{\ell}
\end{equation}
such that
\begin{equation}
    \mathcal{P}_{\ell} = \ell(\ell-1)\prod_{p=1}^{\frac{d}2-1}(\ell+p)\prod_{q=0}^{\frac{d}2}(\ell-q+d) \;\mathcal{N}_{\ell}.
\end{equation}
However, this relation is invertible for $\ell \geqslant 2$. Furthermore, the sequence of $L^2(S^d)$-norms $\lVert\mathcal{N}_\ell\rVert$ of the corresponding $\ell$-th spherical harmonics $\mathcal{N}_\ell$ is rapidly decreasing (in the sense of Definition 1.12 in \cite{Morimoto1998}) because the sequence of norms $\lVert\mathcal{P}_\ell\rVert$ of the spherical harmonics $\mathcal{P}_\ell$ appearing in the decomposition \eqref{Pdecomp} of $\mathcal{P}$ is rapidly decreasing (since $\mathcal{P}$ is smooth, cf. the first part of Theorem 2.45 in \cite{Morimoto1998}). Therefore, the function $\mathcal N$ defined by the series \eqref{Ndecomp} is smooth (cf. the second part of Theorem 2.45 in \cite{Morimoto1998}).
This concludes the proof of Proposition \ref{Soft supermomentum characterisation} for the smooth case. 

\begin{Remark}
    For distributions, the proof is identical: the only difference is that the above sequences are slowly increasing (in the sense of Definition 1.16 in \cite{Morimoto1998}) and that one should make use of the corresponding results (cf. Theorem 2.50 in \cite{Morimoto1998}). See \cite{Bekaert2026} for a detailed discussion.
\end{Remark}

\begin{Remark}
Similarly, one can also prove the first part of Proposition \ref{Soft supermomentum characterisation} for supermomenta $\mathcal P$ and vaccum shifts $\mathcal N$ which are analytic functions or their duals (i.e. hyperfunctions) by making use, respectively, of Theorems 4.55 and 4.57 in \cite{Morimoto1998} about the exponential decay or sub-exponential growth of the corresponding sequences of $L^2(S^d)$-norms. 
\end{Remark}

\begin{Remark}
Finally, one can easily prove Proposition \ref{Soft supermomentum characterisation} for supermomenta $\mathcal{P}$ and vaccum shifts $\mathcal N$ which are square-integrable. In fact, $\mathcal{P}$ is square-integrable iff $\lVert\mathcal{P}\rVert^2=\sum_\ell\,\lVert\mathcal{P}_\ell\rVert^2<\infty$. Therefore, the comparison test implies that $\lVert\mathcal{N}\rVert^2=\sum_\ell\,\lVert\mathcal{N}_\ell\rVert^2<\infty$, hence $\mathcal{N}$ is square-integrable. However, the equivalence \eqref{equivsoft} does not hold because the GJMS operator $\hat{P}_{d+2}$ does not necessarily send a square-integrable function $\mathcal{N}\in L^2(S^d)$ to a square-integrable supermomentum $\mathcal P$, e.g. for $d=2$ one should assume that $\mathcal{N}\in H^2(S^2)$ belongs to the Sobolev space of order two in order to guarantee that $\mathcal{P}\in L^2(S^d)$ (see, e.g., Theorem 6.30 in \cite{Warner1983}).
\end{Remark}

\subsection{Faithful representations of the Komar group}\label{RepKomar}

We here want to prove that induced soft representations with non-zero supermomentum are faithful representations of the Komar group. This comes as a corollary of the following lemma (applied for $K$ the Komar group, $N$ the group of vacuum shifts and $H$ the Lorentz group).

\vspace{1mm}

\noindent\textbf{Lemma.} Let us consider a Lie group $K=H\ltimes N$ which is the semidirect product of a connected Lie group $H$ acting on an Abelian normal subgroup $N$. If the Abelian normal subgroup $N$ is a vector space carrying a faithful irreducible representation of the group $H$ and is such that the dual space $N^*$ also carries a faithful irreducible representation of $H$, then the non-trivial induced representations of $K$ are faithful as well.

\vspace{1mm}

\begin{proof}
Consider a \textit{non-trivial} induced representation of $K=H\ltimes N$. This vector space is spanned by some fields $f$ on a \textit{non-trivial} $H$-orbit $\mathcal{O}_p\neq \{0\}$, where $p\in N^*\setminus\{0\}$ is a \textit{non-zero} linear form on $N$. The restriction of the induced representation to the normal Abelian subgroup $N$ acts by multiplication, i.e. $f(v)\mapsto \exp\big(i\,v(n)\big)f(v)$, $\forall v\in\mathcal{O}_p$, $\forall n\in N$. 
By continuity, a one-parameter subgroup generated by an element $n\in N$ acts trivially on the induced representation if and only if the element $n$ is annihilated by all linear forms in the $H$-orbit, i.e. $v(n)=0$, $\forall v\in\mathcal{O}_p$. By linearity, this condition holds if and only if the element $n$ is annihilated by the vector space $V=\text{span}\,\mathcal{O}_p$, denoting the linear span of the $H$-orbit. 

The linear span $V$ of the $H$-orbit  $\mathcal{O}_p$ is a non-trivial invariant subspace $V\subseteq N^*$ of the contragredient module, hence $V=N^*$
(since $V\neq\{0\}$ and $N^*$ is irreducible). Therefore, the restriction of the induced representation to the normal Abelian subgroup $N$ is faithful (since the kernel of the representation would have to be annihilated by the whole $N^*$).

Finally, since $N^*$ carries a faithful representation of $H$, the only vector $v$ of $V=N^*$ which is annihilated by every element of $H$ is the zero vector $v=0$. Hence, no point of the $H$-orbit is fixed by every element of $H$ (since $\mathcal{O}_p\neq \{0\}$). Therefore, the induced representation is also faithful with respect to $H$.
\end{proof}

\begin{Remark}
If $N$ is assumed \textit{finite-dimensional}, then $N$ is a faithful irreducible $H-$module if and only if the contragredient module $N^*$ is also a faithful irreducible $H$-module. Therefore, in the finite-dimensional case this extra assumption (i.e. that the contragredient module $N^*$ is also a faithful irreducible $H$-module)  would not be necessary in the lemma. Furthermore, in the \textit{infinite-dimensional} case, if $N$ is a Hilbert space then this extra assumption is also not necessary, since the $N$ and its continuous dual $N^\prime$ are isomorphic to each other, due to Riesz representation theorem. Therefore, the lemma indeed applies to the group $N$ of vacuum shifts endowed with a structure of Hilbert space, as in Proposition \ref{UIRLorentz}.
\end{Remark}

\newpage

\bibliographystyle{utphys}
\bibliography{BMSUIR}

@article{Ruzziconi:2026bix,
    author = "Ruzziconi, Romain",
    title = "{Carrollian physics and holography}",
    eprint = "2602.02644",
    archivePrefix = "arXiv",
    primaryClass = "hep-th",
    doi = "10.1016/j.physrep.2026.03.005",
    journal = "Phys. Rept.",
    volume = "1182",
    pages = "1--87",
    year = "2026"
}

@article{Nguyen:2025zhg,
    author = "Nguyen, Kevin",
    title = "{Lectures on Carrollian holography}",
    eprint = "2511.10162",
    archivePrefix = "arXiv",
    primaryClass = "hep-th",
    doi = "10.1140/epjc/s10052-026-15952-1",
    journal = "Eur. Phys. J. C",
    volume = "86",
    number = "6",
    pages = "733",
    year = "2026"
}

@article{Bagchi:2025vri,
    author = "Bagchi, Arjun and Banerjee, Aritra and Dhivakar, Prateksh and Mondal, Saikat and Shukla, Ashish",
    title = "{The Carrollian kaleidoscope}",
    eprint = "2506.16164",
    archivePrefix = "arXiv",
    primaryClass = "hep-th",
    doi = "10.1140/epjc/s10052-026-15437-1",
    journal = "Eur. Phys. J. C",
    volume = "86",
    number = "4",
    pages = "429",
    year = "2026"
}

@book{Warner1983,
  author    = {Warner, Frank W.},
  title     = {Foundations of Differentiable Manifolds and Lie Groups},
  series    = {Graduate Texts in Mathematics},
  volume    = {94},
  publisher = {Springer-Verlag},
  address   = {New York-Berlin},
  year      = {1983},
  note      = {Corrected reprint of the 1971 edition}
}

@book{Morimoto1998,
  author    = {Mitsuo Morimoto},
  title     = {Analytic Functionals on the Sphere},
  series    = {Translations of Mathematical Monographs},
  volume    = {178},
  publisher = {American Mathematical Society},
  address   = {Providence, RI},
  year      = {1998},
  isbn      = {978-0-8218-0585-5},
  doi       = {10.1090/mmono/178}
}

@article{Redou2003,
author = {Pascal Redou},
title = {Representations of the Conformal Lie Algebra in the Space of Tensor Densities on the Sphere},
journal = {Journal of Nonlinear Mathematical Physics},
volume = {10},
number = {2},
pages = {136--140},
year = {2003},
doi = {10.2991/jnmp.2003.10.2.1},
url = {https://doi.org/10.2991/jnmp.2003.10.2.1},
publisher = {Atlantis Press}
}

@article{Henneaux:1979vn,
    author = "Henneaux, Marc",
    title = "{Geometry of Zero Signature Space-times}",
    reportNumber = "PRINT-79-0606 (PRINCETON)",
    journal = "Bull. Soc. Math. Belg.",
    volume = "31",
    pages = "47--63",
    year = "1979"
}

@article{Gover:2003am,
    author = "Gover, A. Rod and Hirachi, Kengo",
    title = "{Conformally invariant powers of the Laplacian: A Complete non-existence theorem}",
    eprint = "math/0304082",
    archivePrefix = "arXiv",
    doi = "10.1090/S0894-0347-04-00450-3",
    journal = "J. Am. Math. Soc.",
    volume = "17",
    pages = "389--405",
    year = "2004"
}

@article{Campoleoni:2020ejn,
    author = "Campoleoni, Andrea and Francia, Dario and Heissenberg, Carlo",
    title = "{On asymptotic symmetries in higher dimensions for any spin}",
    eprint = "2011.04420",
    archivePrefix = "arXiv",
    primaryClass = "hep-th",
    reportNumber = "NORDITA 2020-103",
    doi = "10.1007/JHEP12(2020)129",
    journal = "JHEP",
    volume = "12",
    pages = "129",
    year = "2020"
}

@article{Satishchandran:2019pyc,
    author = "Satishchandran, Gautam and Wald, Robert M.",
    title = "{Asymptotic behavior of massless fields and the memory effect}",
    eprint = "1901.05942",
    archivePrefix = "arXiv",
    primaryClass = "gr-qc",
    doi = "10.1103/PhysRevD.99.084007",
    journal = "Phys. Rev. D",
    volume = "99",
    number = "8",
    pages = "084007",
    year = "2019"
}

@article{Pate:2017fgt,
    author = "Pate, Monica and Raclariu, Ana-Maria and Strominger, Andrew",
    title = "{Gravitational Memory in Higher Dimensions}",
    eprint = "1712.01204",
    archivePrefix = "arXiv",
    primaryClass = "hep-th",
    doi = "10.1007/JHEP06(2018)138",
    journal = "JHEP",
    volume = "06",
    pages = "138",
    year = "2018"
}

@article{Mao:2017wvx,
    author = "Mao, Pujian and Ouyang, Hao",
    title = "{Note on soft theorems and memories in even dimensions}",
    eprint = "1707.07118",
    archivePrefix = "arXiv",
    primaryClass = "hep-th",
    doi = "10.1016/j.physletb.2017.08.064",
    journal = "Phys. Lett. B",
    volume = "774",
    pages = "715--722",
    year = "2017"
}

@article{Kapec:2015vwa,
    author = "Kapec, Daniel and Lysov, Vyacheslav and Pasterski, Sabrina and Strominger, Andrew",
    title = "{Higher-dimensional supertranslations and Weinberg{\textquoteright}s soft graviton theorem}",
    eprint = "1502.07644",
    archivePrefix = "arXiv",
    primaryClass = "gr-qc",
    reportNumber = "CALT-TH-2015-006",
    doi = "10.4310/AMSA.2017.v2.n1.a2",
    journal = "Ann. Math. Sci. Appl.",
    volume = "02",
    pages = "69--94",
    year = "2017"
}

@article{Branson,
  title={Sharp inequalities, the functional determinant, and the complementary series},
  author={Branson, Thomas P},
  journal={Trans. Am. Math. Soc.},
  volume={347},
  number={10},
  pages={3671--3742},
  year={1995}
}

@article{GJMS,
author = {Graham, C. Robin and Jenne, Ralph and Mason, Lionel J. and Sparling, George A. J.},
title = {Conformally Invariant Powers of the Laplacian, I: Existence},
journal = {J. Lond. Math. Soc. },
volume = {s2-46},
number = {3},
pages = {557-565},
doi = {https://doi.org/10.1112/jlms/s2-46.3.557},
year = {1992}
}

@article{PhysRevLett.15.76,
  title = {Quantized Gravitational Theory and Internal Symmetries},
  author = {Komar, Arthur},
  journal = {Phys. Rev. Lett.},
  volume = {15},
  issue = {2},
  pages = {76--78},
  numpages = {0},
  year = {1965},
  publisher = {American Physical Society},
  doi = {10.1103/PhysRevLett.15.76},
}

@article{Sun:2021thf,
    author = "Sun, Zimo",
    title = "{A note on the representations of SO(1,d + 1)}",
    eprint = "2111.04591",
    archivePrefix = "arXiv",
    primaryClass = "hep-th",
    doi = "10.1142/S0129055X24300073",
    journal = "Rev. Math. Phys.",
    volume = "37",
    number = "01",
    pages = "2430007",
    year = "2025"
}

@book{Todorov:1978rf,
    author = "Todorov, I. T. and Mintchev, M. C. and Petkova, V. B.",
    title = "{Conformal Invariance in Quantum Field Theory}",
    publisher = "Scuola Normale Superiore",
    year = "1978"
}

@article{Basile:2016aen,
    author = "Basile, Thomas and Bekaert, Xavier and Boulanger, Nicolas",
    title = "{Mixed-symmetry fields in de Sitter space: a group theoretical glance}",
    eprint = "1612.08166",
    archivePrefix = "arXiv",
    primaryClass = "hep-th",
    doi = "10.1007/JHEP05(2017)081",
    journal = "JHEP",
    volume = "05",
    pages = "081",
    year = "2017"
}

@article{Bekaert:2006py,
    author = "Bekaert, Xavier and Boulanger, Nicolas",
    title = "{The unitary representations of the Poincar\'e group in any spacetime dimension}",
    eprint = "hep-th/0611263",
    archivePrefix = "arXiv",
    doi = "10.21468/SciPostPhysLectNotes.30",
    journal = "SciPost Phys. Lect. Notes",
    volume = "30",
    pages = "1",
    year = "2021"
}

@article{Figueroa-OFarrill:2017tcy,
    author = "Figueroa-O'Farrill, Jos{\'e} M.",
    title = "{Higher-dimensional kinematical Lie algebras via deformation theory}",
    eprint = "1711.07363",
    archivePrefix = "arXiv",
    primaryClass = "hep-th",
    reportNumber = "EMPG-17-12",
    doi = "10.1063/1.5016616",
    journal = "J. Math. Phys.",
    volume = "59",
    number = "6",
    pages = "061702",
    year = "2018"
}

@article{Sachs:1962zza,
    author = "Sachs, R.",
    title = "{Asymptotic symmetries in gravitational theory}",
    doi = "10.1103/PhysRev.128.2851",
    journal = "Phys. Rev.",
    volume = "128",
    pages = "2851--2864",
    year = "1962"
}

@article{Ashtekar:2014zsa,
    author = "Ashtekar, Abhay",
    title = "{Geometry and physics of null infinity}",
    eprint = "1409.1800",
    archivePrefix = "arXiv",
    primaryClass = "gr-qc",
    reportNumber = "IGC-14-9-1",
    doi = "10.4310/sdg.2015.v20.n1.a5",
    journal = "Surveys Diff. Geom.",
    volume = "20",
    number = "1",
    pages = "99--122",
    year = "2015"
}

@article{Bekaert:2024itn,
    author = "Bekaert, Xavier and Campoleoni, Andrea and Pekar, Simon",
    title = "{Holographic Carrollian conformal scalars}",
    eprint = "2404.02533",
    archivePrefix = "arXiv",
    primaryClass = "hep-th",
    doi = "10.1007/JHEP05(2024)242",
    journal = "JHEP",
    volume = "05",
    pages = "242",
    year = "2024"
}

@article{Ashtekar:1981hw,
    author = "Ashtekar, A.",
    title = "{Radiative Degrees of Freedom of the Gravitational Field in Exact General Relativity}",
    doi = "10.1063/1.525169",
    journal = "J. Math. Phys.",
    volume = "22",
    pages = "2885--2895",
    year = "1981"
}

@article{Ashtekar:1981sf,
    author = "Ashtekar, A.",
    title = "{Asymptotic Quantization of the Gravitational Field}",
    doi = "10.1103/PhysRevLett.46.573",
    journal = "Phys. Rev. Lett.",
    volume = "46",
    pages = "573--576",
    year = "1981"
}

@book{Ashtekar:1987tt,
    author = "Ashtekar, A.",
    title = "{Asymptotic Quantization: Based on 1984 Naples Lectures,}",
    year = "1987",
    publisher = "Bibliopolis"
}

@article{Duval:2014uva,
    author = "Duval, C. and Gibbons, G. W. and Horvathy, P. A.",
    title = "{Conformal Carroll groups and BMS symmetry}",
    eprint = "1402.5894",
    archivePrefix = "arXiv",
    primaryClass = "gr-qc",
    doi = "10.1088/0264-9381/31/9/092001",
    journal = "Class. Quant. Grav.",
    volume = "31",
    pages = "092001",
    year = "2014"
}

@article{Duval:2014lpa,
    author = "Duval, C. and Gibbons, G. W. and Horvathy, P. A.",
    title = "{Conformal Carroll groups}",
    eprint = "1403.4213",
    archivePrefix = "arXiv",
    primaryClass = "hep-th",
    doi = "10.1088/1751-8113/47/33/335204",
    journal = "J. Phys. A",
    volume = "47",
    number = "33",
    pages = "335204",
    year = "2014"
}

@article{Bacry:1968zf,
    author = "Bacry, H. and Levy-Leblond, J.",
    title = "{Possible kinematics}",
    doi = "10.1063/1.1664490",
    journal = "J. Math. Phys.",
    volume = "9",
    pages = "1605--1614",
    year = "1968"
}

@article{SenGupta:1966qer,
    author = "Sen Gupta, N. D.",
    title = "{On an analogue of the Galilei group}",
    doi = "10.1007/BF02740871",
    journal = "Nuovo Cim. A",
    volume = "44",
    number = "2",
    pages = "512--517",
    year = "1966"
}

@article{Levy-Leblond:1965dsc,
    author = "L{\'e}vy-Leblond, Jean-Marc",
    title = "{Une nouvelle limite non-relativiste du groupe de Poincar{\'e}}",
    journal = "Ann. Inst. H. Poincar\'e Phys. Theor. A",
    volume = "3",
    number = "1",
    pages = "1--12",
    year = "1965"
}

@article{Eastwood:1991,
author = {Michael G. Eastwood and C. Robin Graham},
title = {{Invariants of conformal densities}},
volume = {63},
journal = {Duke Math. J.},
number = {3},
publisher = {Duke University Press},
pages = {633 -- 671},
year = {1991},
doi = {10.1215/S0012-7094-91-06327-1},
}

@article{Prabhu:2022zcr,
    author = "Prabhu, Kartik and Satishchandran, Gautam and Wald, Robert M.",
    title = "{Infrared finite scattering theory in quantum field theory and quantum gravity}",
    eprint = "2203.14334",
    archivePrefix = "arXiv",
    primaryClass = "hep-th",
    doi = "10.1103/PhysRevD.106.066005",
    journal = "Phys. Rev. D",
    volume = "106",
    number = "6",
    pages = "066005",
    year = "2022"
}

@article{Prabhu:2024lmg,
    author = "Prabhu, Kartik and Satishchandran, Gautam",
    title = "{Infrared finite scattering theory: Amplitudes and soft theorems}",
    eprint = "2402.18637",
    archivePrefix = "arXiv",
    primaryClass = "hep-th",
    doi = "10.1103/PhysRevD.110.085022",
    journal = "Phys. Rev. D",
    volume = "110",
    number = "8",
    pages = "085022",
    year = "2024"
}

@article{Prabhu:2024zwl,
    author = "Prabhu, Kartik and Satishchandran, Gautam",
    title = "{Infrared finite scattering theory: scattering states and representations of the BMS group}",
    eprint = "2402.00102",
    archivePrefix = "arXiv",
    primaryClass = "hep-th",
    doi = "10.1007/JHEP08(2024)055",
    journal = "JHEP",
    volume = "08",
    pages = "055",
    year = "2024"
}

@book{Weinberg:1995mt,
    author = "Weinberg, Steven",
    title = "{The quantum theory of fields. Vol. 1: Foundations}",
    doi = "10.1017/CBO9781139644167",
    isbn = "978-0-521-67053-1, 978-0-511-25204-4",
    publisher = "Cambridge University Press",
    month = "6",
    year = "2005"
}

@article{Piard1977279,
title = {Representations of the {B}ondi-{M}etzner-{S}achs group with the Hilbert topology},
journal = {Rep. Math. Phys.},
volume = {11},
number = {2},
pages = {279-283},
year = {1977},
issn = {0034-4877},
doi = {https://doi.org/10.1016/0034-4877(77)90068-4},
author = {Alain Piard},
}

@article{Wigner:1939cj,
    author = "Wigner, Eugene P.",
    editor = "Kim, Y. S. and Zachary, W. W.",
    title = "{On Unitary Representations of the Inhomogeneous Lorentz Group}",
    doi = "10.2307/1968551",
    journal = "Annals Math.",
    volume = "40",
    pages = "149--204",
    year = "1939"
}

@article{Piard,
title = {Unitary representations of semi-direct product groups with infinite dimensional Abelian normal subgroup},
journal = {Rep. Math. Phys.},
volume = {11},
number = {2},
pages = {259-278},
year = {1977},
issn = {0034-4877},
doi = {https://doi.org/10.1016/0034-4877(77)90067-2},
author = {Alain Piard}
}

@article{Hollands:2016oma,
    author = "Hollands, Stefan and Ishibashi, Akihiro and Wald, Robert M.",
    title = "{BMS supertranslations and memory in four and higher dimensions}",
    eprint = "1612.03290",
    archivePrefix = "arXiv",
    primaryClass = "gr-qc",
    doi = "10.1088/1361-6382/aa777a",
    journal = "Class. Quant. Grav.",
    volume = "34",
    number = "15",
    pages = "155005",
    year = "2017"
}

@article{Garfinkle:2017fre,
    author = "Garfinkle, David and Hollands, Stefan and Ishibashi, Akihiro and Tolish, Alexander and Wald, Robert M.",
    title = "{The memory effect for particle scattering in even spacetime dimensions}",
    eprint = "1702.00095",
    archivePrefix = "arXiv",
    primaryClass = "gr-qc",
    doi = "10.1088/1361-6382/aa777b",
    journal = "Class. Quant. Grav.",
    volume = "34",
    number = "14",
    pages = "145015",
    year = "2017"
}

@article{Bekaert:2024uuy,
    author = "Bekaert, Xavier and Donnay, Laura and Herfray, Yannick",
    title = "{BMS particles}",
    eprint = "2412.06002",
    archivePrefix = "arXiv",
    primaryClass = "hep-th",
    doi = "10.1103/8376-fync",
    journal = "Phys. Rev. Lett.",
    volume = "135",
    number = "13",
    pages = "131602",
    year = "2025"
}

@article{Bekaert:2025kjb,
    author = "Bekaert, Xavier and Herfray, Yannick",
    title = "{BMS Representations for Generic Supermomentum}",
    eprint = "2505.05368",
    archivePrefix = "arXiv",
    primaryClass = "hep-th",
    doi = "10.1007/s00220-025-05513-0",
    journal = "Commun. Math. Phys.",
    volume = "407",
    number = "2",
    pages = "35",
    year = "2026"
}

@article{Girardello:1974sq,
    author = "Girardello, L. and Parravicini, G.",
    title = "{Continuous spins in the {B}ondi-{M}etzner-{S}achs group of asymptotic symmetry in general relativity}",
    doi = "10.1103/PhysRevLett.32.565",
    journal = "Phys. Rev. Lett.",
    volume = "32",
    pages = "565--568",
    year = "1974"
}

@article{McCarthy_72-I,
 ISSN = {00804630},
 author = {P. J. McCarthy},
 journal = {Proc. R. Soc. (London) A},
 number = {1583},
 pages = {517--535},
 publisher = {The Royal Society},
 title = {Representations of the {B}ondi-{M}etzner-{S}achs Group. I. Determination of the Representations},
 volume = {330},
 year = {1972}
}

@article{Mccarthy:1972ry,
    author = "McCarthy, P. J. M.",
    title = "{Asymptotically flat space-times and elementary particles}",
    doi = "10.1103/PhysRevLett.29.817",
    journal = "Phys. Rev. Lett.",
    volume = "29",
    pages = "817--819",
    year = "1972"
}

@article{McCarthy_73-II,
 ISSN = {00804630},
 author = {P. J. McCarthy},
 journal = {Proc. R. Soc. (London) A},
 number = {1594},
 pages = {317--336},
 publisher = {The Royal Society},
 title = {Representations of the {B}ondi-{M}etzner-{S}achs Group. II. Properties and Classification of the Representations},
 volume = {333},
 year = {1973}
}

@article{McCarthy_73-III,
 ISSN = {00804630},
 author = {P. J. McCarthy and M. Crampin},
 journal = {Proc. R. Soc. (London) A},
 number = {1602},
 pages = {301--311},
 publisher = {The Royal Society},
 title = {Representations of the {B}ondi-{M}etzner-{S}achs group. III. {P}oincar\'e Spin multiplicities and irreducibility},
 volume = {335},
 year = {1973}
}

@article{McCarthy_76-IV,
 ISSN = {00804630},
 author = {M. Crampin and P. J. McCarthy},
 journal = {Proc. R. Soc. (London) A},
 number = {1664},
 pages = {55--70},
 publisher = {The Royal Society},
 title = {Representations of the {B}ondi-{M}etzner-{S}achs group. IV. {C}antoni representations are induced},
 volume = {351},
 year = {1976}
}

@article{McCarthy_75,
 ISSN = {00804630},
  author = {P. J. McCarthy},
 journal = {Proc. R. Soc. (London) A},
 number = {1635},
 pages = {489--523},
 publisher = {The Royal Society},
 title = {The {B}ondi-{M}etzner-{S}achs group in the nuclear topology},
 volume = {343},
 year = {1975}
}

@article{McCarthy_78,
 author = {P. J. McCarthy},
 journal = {Proc. R. Soc. (London) A},
 pages = {141–171},
 publisher = {The Royal Society},
 title = {Lifting of projective representations of the {B}ondi-{M}etzner-{S}achs group},
 volume = {358},
 year = {1978}
}

@article{McCarthy_78errata,
 author = {P. J. McCarthy},
 journal = {Proc. R. Soc. (London) A},
number = {1695},
 pages = {495-498},
 publisher = {The Royal Society},
 title = {Hyperfunctions and asymptotic symmetries},
 volume = {358},
 year = {1978}
}

@article{McCarthy:1974aw,
    author = "Crampin, M. and McCarthy, P. J.",
    title = "{Physical significance of the topology of the {B}ondi-{M}etzner-{S}achs group}",
    doi = "10.1103/PhysRevLett.33.547",
    journal = "Phys. Rev. Lett.",
    volume = "33",
    pages = "547--550",
    year = "1974"
}

@article{Ashtekar:2018lor,
    author = "Ashtekar, Abhay and Campiglia, Miguel and Laddha, Alok",
    title = "{Null infinity, the BMS group and infrared issues}",
    eprint = "1808.07093",
    archivePrefix = "arXiv",
    primaryClass = "gr-qc",
    doi = "10.1007/s10714-018-2464-3",
    journal = "Gen. Rel. Grav.",
    volume = "50",
    number = "11",
    pages = "140--163",
    year = "2018"
}

@article{Newman:1966ub,
    author = "Newman, E. T. and Penrose, R.",
    title = "{Note on the {B}ondi-{M}etzner-{S}achs group}",
    doi = "10.1063/1.1931221",
    journal = "J. Math. Phys.",
    volume = "7",
    pages = "863--870",
    year = "1966"
}

@book{Strominger:2017zoo,
    author = "Strominger, Andrew",
    title = "{Lectures on the Infrared Structure of Gravity and Gauge Theory}",
publisher = {Princeton University Press},
    eprint = "1703.05448",
    archivePrefix = "arXiv",
    primaryClass = "hep-th",
    isbn = "978-0-691-17973-5",
    year = "2017"
}

@article{Strominger:2013jfa,
    author = "Strominger, Andrew",
    title = "{On BMS Invariance of Gravitational Scattering}",
    eprint = "1312.2229",
    archivePrefix = "arXiv",
    primaryClass = "hep-th",
    doi = "10.1007/JHEP07(2014)152",
    journal = "JHEP",
    volume = "07",
    pages = "152",
    year = "2014"
}

@article{He:2014laa,
    author = "He, Temple and Lysov, Vyacheslav and Mitra, Prahar and Strominger, Andrew",
    title = "{BMS supertranslations and Weinberg\textquoteright{}s soft graviton theorem}",
    eprint = "1401.7026",
    archivePrefix = "arXiv",
    primaryClass = "hep-th",
    doi = "10.1007/JHEP05(2015)151",
    journal = "JHEP",
    volume = "05",
    pages = "151",
    year = "2015"
}

@article{Bekaert:2022ipg,
    author = "Bekaert, Xavier and Oblak, Blagoje",
    title = "{Massless scalars and higher-spin BMS in any dimension}",
    eprint = "2209.02253",
    archivePrefix = "arXiv",
    primaryClass = "hep-th",
    doi = "10.1007/JHEP11(2022)022",
    journal = "JHEP",
    volume = "11",
    pages = "022",
    year = "2022"
}

@article{Longhi:1997zt,
	archiveprefix = {arXiv},
	author = {Longhi, G. and Materassi, M.},
	doi = {10.1063/1.532782},
	eprint = {hep-th/9803128},
	journal = {J. Math. Phys.},
	pages = {480--500},
	reportnumber = {DFF-301-03-98},
	title = {{A Canonical realization of the BMS algebra}},
	volume = {40},
	year = {1999}}

@article{Campiglia:2015kxa,
	archiveprefix = {arXiv},
	author = {Campiglia, Miguel and Laddha, Alok},
	doi = {10.1007/JHEP12(2015)094},
	eprint = {1509.01406},
	journal = {JHEP},
	pages = {094},
	primaryclass = {hep-th},
	title = {{Asymptotic symmetries of gravity and soft theorems for massive particles}},
	volume = {12},
	year = {2015}}

@article{bondi_gravitational_1962,
	author = {Bondi, Hermann and {Van der Burg}, M. G. J. and Metzner, A. W. K.},
	doi = {10.1098/rspa.1962.0161},
	journal = {Proc. R. Soc. (London) A},
	number = {1336},
	pages = {21--52},
	title = {Gravitational Waves in General Relativity, {{VII}}. {{Waves}} from Axi-Symmetric Isolated System},
	volume = {269},
	year = {1962}}

@Incollection{Geroch1977,
author="Geroch, Robert",
editor="Esposito, F. Paul
and Witten, Louis",
title="Asymptotic Structure of Space-Time",
bookTitle="Asymptotic Structure of Space-Time",
year="1977",
publisher="Springer US",
address="Boston, MA",
pages="1--105",
isbn="978-1-4684-2343-3",
doi="10.1007/978-1-4684-2343-3_1",
}

@article{Ashtekar:1978zz,
    author = "Ashtekar, A. and Hansen, R. O.",
    title = "{A unified treatment of null and spatial infinity in general relativity. I - Universal structure, asymptotic symmetries, and conserved quantities at spatial infinity}",
    doi = "10.1063/1.523863",
    journal = "J. Math. Phys.",
    volume = "19",
    pages = "1542--1566",
    year = "1978"
}

@article{Chatterjee:2017zeb,
    author = "Chatterjee, Atreya and Lowe, David A.",
    title = "{BMS symmetry, soft particles and memory}",
    eprint = "1712.03211",
    archivePrefix = "arXiv",
    primaryClass = "hep-th",
    reportNumber = "BROWN-HET-1725",
    doi = "10.1088/1361-6382/aab5cc",
    journal = "Class. Quant. Grav.",
    volume = "35",
    number = "9",
    pages = "094001",
    year = "2018"
}

@article{Cantoni:1966,
  title={A Class of Representations of the Generalized {B}ondi-{M}etzner Group},
  author={Vittorio Cantoni},
  journal={J. Math. Phys.},
  year={1966},
  volume={7},
  pages={1361-1364},
}

@article{Barnich:2014kra,
    author = "Barnich, Glenn and Oblak, Blagoje",
    title = "{Notes on the BMS group in three dimensions: I. Induced representations}",
    eprint = "1403.5803",
    archivePrefix = "arXiv",
    primaryClass = "hep-th",
    doi = "10.1007/JHEP06(2014)129",
    journal = "JHEP",
    volume = "06",
    pages = "129",
    year = "2014"
}

@article{Barnich:2015uva,
    author = "Barnich, Glenn and Oblak, Blagoje",
    title = "{Notes on the BMS group in three dimensions: II. Coadjoint representation}",
    eprint = "1502.00010",
    archivePrefix = "arXiv",
    primaryClass = "hep-th",
    doi = "10.1007/JHEP03(2015)033",
    journal = "JHEP",
    volume = "03",
    pages = "033",
    year = "2015"
}

@phdthesis{Oblak:2016eij,
    author = "Oblak, Blagoje",
    title = "{BMS particles in three dimensions}",
    eprint = "1610.08526",
    archivePrefix = "arXiv",
    primaryClass = "hep-th",
    doi = "10.1007/978-3-319-61878-4",
    school = "U. Brussels, Brussels U.",
    year = "2016"
}

@article{Borthwick:2024skd,
    author = {Borthwick, Jack and Chantreau, Ma{\"e}l and Herfray, Yannick},
    title = "{Ti and Spi, Carrollian extended boundaries at timelike and spatial infinity}",
    eprint = "2412.15996",
    archivePrefix = "arXiv",
    primaryClass = "gr-qc",
    doi = "10.1088/1361-6382/ae0d3f",
    journal = "Class. Quant. Grav.",
    volume = "42",
    number = "20",
    pages = "205012",
    year = "2025"
}

@article{Melas:2017jzb,
    author = "Melas, Evangelos",
    title = "{On the representation theory of the Bondi\textendash{}Metzner\textendash{}Sachs group and its variants in three space\textendash{}time dimensions}",
    eprint = "1703.05980",
    archivePrefix = "arXiv",
    primaryClass = "math-ph",
    doi = "10.1063/1.4993198",
    journal = "J. Math. Phys.",
    volume = "58",
    number = "7",
    pages = "071705",
    year = "2017"
}

@article{Choi:2017ylo,
    author = "Choi, Sangmin and Akhoury, Ratindranath",
    title = "{BMS Supertranslation Symmetry Implies Faddeev-Kulish Amplitudes}",
    eprint = "1712.04551",
    archivePrefix = "arXiv",
    primaryClass = "hep-th",
    doi = "10.1007/JHEP02(2018)171",
    journal = "JHEP",
    volume = "02",
    pages = "171",
    year = "2018"
}

@article{Herfray:2020rvq,
    author = "Herfray, Yannick",
    title = "{Asymptotic shear and the intrinsic conformal geometry of null-infinity}",
    eprint = "2001.01281",
    archivePrefix = "arXiv",
    primaryClass = "gr-qc",
    doi = "10.1063/5.0003616",
    journal = "J. Math. Phys.",
    volume = "61",
    number = "7",
    pages = "072502",
    year = "2020"
}

@article{Herfray:2021xyp,
    author = "Herfray, Yannick",
    title = "{Tractor Geometry of Asymptotically Flat Spacetimes}",
    eprint = "2103.10405",
    archivePrefix = "arXiv",
    primaryClass = "gr-qc",
    doi = "10.1007/s00023-022-01174-0",
    journal = "Ann. H. Poincar\'e",
    volume = "23",
    number = "9",
    pages = "3265--3310",
    year = "2022"
}

@article{Herfray:2021qmp,
    author = "Herfray, Yannick",
    title = "{Carrollian manifolds and null infinity: a view from Cartan geometry}",
    eprint = "2112.09048",
    archivePrefix = "arXiv",
    primaryClass = "gr-qc",
    doi = "10.1088/1361-6382/ac635f",
    journal = "Class. Quant. Grav.",
    volume = "39",
    number = "21",
    pages = "215005",
    year = "2022"
}

@book{Humphreys1996,
    author = "Humphreys, J.F.",
    title = "{A course in group theory}",
    year = "1996",
    publisher = "Oxford University Press"
}

@book{Dieudonne1970,
    author = "Dieudonn\'e, J.",
    title = "{\'{E}l\'ements d'analyse, tome {III}}",
    year = "1970",
    publisher = "Gauthier-Villars"
}

@article{Ramaswamy1981,
    author = "{Ramaswamy, Sriram and Sen, Amitabha}",
    title = "{Dual‐mass in general relativity}",
    journal = "{J. Math. Phys.}",
    volume = "{22}",
    number = "{11}",
    pages = "{2612-2619}",
    year = "{1981}",    
    issn = "{0022-2488}",
    doi = "{10.1063/1.524839}",    
}

@article{Penrose:1962ij,
    author = "Penrose, Roger",
    title = "{Asymptotic properties of fields and space-times}",
    doi = "10.1103/PhysRevLett.10.66",
    journal = "Phys. Rev. Lett.",
    volume = "10",
    pages = "66--68",
    year = "1963"
}

@article{Dyson:1972sd,
    author = "Dyson, F. J.",
    title = "{Missed opportunities}",
    doi = "10.1090/S0002-9904-1972-12971-9",
    journal = "Bull. Am. Math. Soc.",
    volume = "78",
    pages = "635--639",
    year = "1972"
}

@article{Figueroa-OFarrill:2018ilb,
    author = "Figueroa-O'Farrill, Jos{\'e} and Prohazka, Stefan",
    title = "{Spatially isotropic homogeneous spacetimes}",
    eprint = "1809.01224",
    archivePrefix = "arXiv",
    primaryClass = "hep-th",
    reportNumber = "EMPG-18-01",
    doi = "10.1007/JHEP01(2019)229",
    journal = "JHEP",
    volume = "01",
    pages = "229",
    year = "2019"
}

@article{Morand:2018tke,
    author = "Morand, Kevin",
    title = "{Embedding Galilean and Carrollian geometries I. Gravitational waves}",
    eprint = "1811.12681",
    archivePrefix = "arXiv",
    primaryClass = "hep-th",
    doi = "10.1063/1.5130907",
    journal = "J. Math. Phys.",
    volume = "61",
    number = "8",
    pages = "082502",
    year = "2020"
}

@article{Bekaert:2015xua,
    author = "Bekaert, Xavier and Morand, Kevin",
    title = "{Connections and dynamical trajectories in generalised Newton-Cartan gravity II. An ambient perspective}",
    eprint = "1505.03739",
    archivePrefix = "arXiv",
    primaryClass = "hep-th",
    doi = "10.1063/1.5030328",
    journal = "J. Math. Phys.",
    volume = "59",
    number = "7",
    pages = "072503",
    year = "2018"
}

@article{Figueroa-OFarrill:2022nui,
    author = "Figueroa-O'Farrill, Jos{\'e}",
    title = "{Non-lorentzian spacetimes}",
    eprint = "2204.13609",
    archivePrefix = "arXiv",
    primaryClass = "math.DG",
    reportNumber = "EMPG-21-16",
    doi = "10.1016/j.difgeo.2022.101894",
    journal = "Differ. Geom. Appl.",
    volume = "82",
    pages = "101894",
    year = "2022"
}

@book{Alice2,
    author = "Carroll, L.",
    title = "{Through the Looking-Glass, and What Alice Found There}",
    year = "1871",
    publisher = "Macmillan \& Co"
}

@article{Duval:2014uoa,
    author = "Duval, C. and Gibbons, G. W. and Horvathy, P. A. and Zhang, P. M.",
    title = "{Carroll versus Newton and Galilei: two dual non-Einsteinian concepts of time}",
    eprint = "1402.0657",
    archivePrefix = "arXiv",
    primaryClass = "gr-qc",
    doi = "10.1088/0264-9381/31/8/085016",
    journal = "Class. Quant. Grav.",
    volume = "31",
    pages = "085016",
    year = "2014"
}

@article{Hartong:2015xda,
    author = "Hartong, Jelle",
    title = "{Gauging the Carroll Algebra and Ultra-Relativistic Gravity}",
    eprint = "1505.05011",
    archivePrefix = "arXiv",
    primaryClass = "hep-th",
    doi = "10.1007/JHEP08(2015)069",
    journal = "JHEP",
    volume = "08",
    pages = "069",
    year = "2015"
}

@article{Campoleoni:2023fug,
    author = "Campoleoni, Andrea and Delfante, Arnaud and Pekar, Simon and Petropoulos, P. Marios and Rivera-Betancour, David and Vilatte, Matthieu",
    title = "{Flat from anti de Sitter}",
    eprint = "2309.15182",
    archivePrefix = "arXiv",
    primaryClass = "hep-th",
    reportNumber = "CPHT-RR054.082023",
    doi = "10.1007/JHEP12(2023)078",
    journal = "JHEP",
    volume = "12",
    pages = "078",
    year = "2023"
}

@article{Newman:1976gc,
    author = "Newman, E. T.",
    title = "{Heaven and Its Properties}",
    doi = "10.1007/BF00762018",
    journal = "Gen. Rel. Grav.",
    volume = "7",
    pages = "107--111",
    year = "1976"
}

@article{Adamo:2009vu,
    author = "Adamo, T. M. and Kozameh, C. N. and Newman, E. T.",
    title = "{Null Geodesic Congruences, Asymptotically Flat Space-Times and Their Physical Interpretation}",
    eprint = "0906.2155",
    archivePrefix = "arXiv",
    primaryClass = "gr-qc",
    doi = "10.12942/lrr-2009-6",
    journal = "Living Rev. Rel.",
    volume = "12",
    pages = "6",
    year = "2009"
}

@article{Figueroa-OFarrill:2022mcy,
    author = "Figueroa-O'Farrill, Jos{\'e} and Have, Emil and Prohazka, Stefan and Salzer, Jakob",
    title = "{The gauging procedure and carrollian gravity}",
    eprint = "2206.14178",
    archivePrefix = "arXiv",
    primaryClass = "hep-th",
    reportNumber = "EMPG-22-10",
    doi = "10.1007/JHEP09(2022)243",
    journal = "JHEP",
    volume = "09",
    pages = "243",
    year = "2022"
}

@article{Campoleoni:2022ebj,
    author = "Campoleoni, Andrea and Henneaux, Marc and Pekar, Simon and P{\'e}rez, Alfredo and Salgado-Rebolledo, Patricio",
    title = "{Magnetic Carrollian gravity from the Carroll algebra}",
    eprint = "2207.14167",
    archivePrefix = "arXiv",
    primaryClass = "hep-th",
    doi = "10.1007/JHEP09(2022)127",
    journal = "JHEP",
    volume = "09",
    pages = "127",
    year = "2022"
}

@article{Bergshoeff:2022eog,
    author = "Bergshoeff, Eric and Figueroa-O'Farrill, Jos{\'e} and Gomis, Joaquim",
    title = "{A non-lorentzian primer}",
    eprint = "2206.12177",
    archivePrefix = "arXiv",
    primaryClass = "hep-th",
    reportNumber = "EMPG-22-08",
    doi = "10.21468/SciPostPhysLectNotes.69",
    journal = "SciPost Phys. Lect. Notes",
    volume = "69",
    pages = "1",
    year = "2023"
}

@article{Donnay:2026urd,
    author = "Donnay, Laura and Herfray, Yannick",
    title = "{Infrared physics of QED and gravity from representation theory}",
    eprint = "2603.06297",
    archivePrefix = "arXiv",
    primaryClass = "hep-th",
    month = "3",
    year = "2026"
}

@article{Ciambelli:2025unn,
    author = "Ciambelli, Luca and Jai-akson, Puttarak",
    title = "{Foundations of Carrollian geometry}",
    eprint = "2510.21651",
    archivePrefix = "arXiv",
    primaryClass = "hep-th",
    reportNumber = "RIKEN-iTHEMS-Report-26, RIKEN-iTHEMS-Report-25",
    doi = "10.1016/j.physrep.2026.05.005",
    journal = "Phys. Rept.",
    volume = "1188",
    pages = "1--51",
    year = "2026"
}

@article{Bosma:2026qcb,
    author = "Bosma, Jorrit and Mann, Jeremy A. and Marteau, Charles and Oblak, Blagoje and Petropoulos, Marios",
    title = "{Dressed States Call for Logarithmic Asymptotic Symmetries}",
    eprint = "2607.16179",
    archivePrefix = "arXiv",
    primaryClass = "hep-th",
    reportNumber = "CPHT-RR019.07202",
    month = "7",
    year = "2026"
}

@article{Ruzziconi:2026isv,
    author = "Ruzziconi, Romain and West, Peter",
    title = "{Extended BMS representations and strings}",
    eprint = "2601.00662",
    archivePrefix = "arXiv",
    primaryClass = "hep-th",
    doi = "10.1007/JHEP05(2026)167",
    journal = "JHEP",
    volume = "05",
    pages = "167",
    year = "2026"
}

@article{Bekaert2026,
    author = "Bekaert, Xavier and Benazout, Bilal and Herfray, Yannick and Raj, Aadharsh",
    title = "{Unitary irreducible representations of QED asymptotic symmetry group}",
journal= "{in preparation}",
    year = "2026"
}

\end{document}